\documentclass[10pt,journal,compsoc]{IEEEtran}
\usepackage[utf8]{inputenc}
\usepackage[utf8]{inputenc}
\usepackage[T1]{fontenc}
\def\ps@headings{%
\def\@oddhead{\mbox{}\scriptsize\rightmark \hfil \thepage}%
\def\@evenhead{\scriptsize\thepage \hfil \leftmark\mbox{}}%
\def\@oddfoot{}%
\def\@evenfoot{}}
\makeatother
\usepackage{commath,amsmath}
\usepackage{amssymb}
\usepackage{amsthm}
\usepackage{dsfont}
\allowdisplaybreaks
\usepackage{algorithm}
\usepackage[noend]{algpseudocode}
\usepackage{float}
\usepackage{graphicx}
\usepackage{mathtools}
\usepackage[thinc]{esdiff}
\usepackage{epsf}
\usepackage{epsfig, amsmath, amssymb, graphicx}
\usepackage{subcaption}
\usepackage{multirow}
\usepackage[english]{babel}
\DeclareGraphicsExtensions{.pdf, .bmp, .eps, .ps, .jpeg, .png}

\newtheorem{theorem}{Theorem}
\newtheorem{lemma}{Lemma}
\newtheorem{corollary}[theorem]{Corollary}

\newenvironment{definition}[1][Definition]{\begin{trivlist}
\item[\hskip \labelsep {\bfseries #1}]}{\end{trivlist}}

\long\def\comment#1{}

\makeatletter
\def\BState{\State\hskip-\ALG@thistlm}
\makeatother

\comment{
\documentclass[10pt, conference]{IEEEtran}
\makeatletter
\def\ps@headings{%
\def\@oddhead{\mbox{}\scriptsize\rightmark \hfil \thepage}%
\def\@evenhead{\scriptsize\thepage \hfil \leftmark\mbox{}}%
\def\@oddfoot{}%
\def\@evenfoot{}}
\makeatother
\pagestyle{empty}
\usepackage{amsmath}
\usepackage{algorithmicx, algpseudocode}
\usepackage{algorithm}
\usepackage{graphicx}
\usepackage{times}
\usepackage{epsf}
\usepackage{epsfig, amsmath, amssymb, graphicx}
\usepackage{subcaption}
\usepackage{multirow}
\usepackage[english]{babel}
\DeclareGraphicsExtensions{.pdf, .bmp, .eps, .ps, .jpeg, .png}

\newtheorem{theorem}{Theorem}
\newtheorem{lemma}{Lemma}
\newtheorem{corollary}[theorem]{Corollary}

\long\def\comment#1{}
}

\begin{document}

\title{Optimizing Fund Allocation for Game-based Verifiable Computation Outsourcing}

\author{Pinglan~Liu,~\IEEEmembership{Student Member,~IEEE,}
        Xiaojuan~Ma,~
        and~Wensheng~Zhang,~\IEEEmembership{Member,~IEEE}
\thanks{
P. Liu, X. Ma, and W. Zhang are with 
the Computer Science Department, Iowa State University, Ames, IA, 50011.
E-mail: \{pinglan,xiaojuan,wzhang\}@iastate.edu}
}

\IEEEtitleabstractindextext{%
\begin{abstract}
This paper considers the setting where a cloud server services a static set or a dynamic sequence of tasks submitted by multiple clients. Every client wishes to assure honest execution of tasks by additionally employing a trusted third party (TTP) to re-compute the tasks with a certain probability. The cloud server makes a deposit for each task it takes, each client allocates a budget (including the wage for the server and the cost for possibly hiring TTP) for each task submitted, and every party has its limited fund for either deposits or task budgets. We study how to allocate the funds optimally to achieve the three-fold goals: a rational cloud server honestly computes each task; the server's wage is maximized; the overall delay for task verification is minimized. We apply game theory to formulate the optimization problems, and develop the optimal or heuristic solutions for three application scenarios. For each of the solutions, we analyze it through either rigorous proofs or extensive simulations. To the best of our knowledge, this is the first work on optimizing fund allocation for verifiable outsourcing of computation in the setting of one server and multiple clients, based on game theory.

%
\end{abstract}

\begin{IEEEkeywords}
Verifiable Computation Outsourcing, Game theory, Cloud Computing, Fund Allocation, Optimization.
\end{IEEEkeywords}}

\maketitle

%
\IEEEpeerreviewmaketitle

\comment{
\IEEEraisesectionheading{\section{Introduction}\label{sec:introduction}}
\IEEEPARstart{T}{he} popularity of cloud services and decentralized platforms promote the development and prosperity of computation outsourcing. Clients export heavy computational task to executors who have available and intensive computational resources to handle it. The executors could be cloud service providers or nodes in the decentralized network. Executors aim to efficiently utilize the available computational resources and achieve an optimal benefits from the computation. Clients desire to outsource tasks without paying more than the predefined budget and get all correct computation results as less delay as possible.

In order to guarantee the correctness of computation result, we propose an outsourcing mechanism that client requires executor to pay a deposit which will be returned back only if there is no dishonest behavior of executor caught by the client. The executor computes and reveals result to client. The client raises a trusted third party (TTP) to check the correctness of the result received from the executor with a low and adaptable frequency. Client will pay executor a wage and the deposit will be returned as long as the result is accepted by the client. The schema ensures that economically rational executors prefer to act honestly compared with returning an incorrect result and taking the risk of raising the TTP as well as losing the deposit.

The goal of client is to minimize the total computation delay and executor aims to maximize the wages. That is, client needs to find out an optimal strategy to allocate its total budget to different tasks and pay hired executors for the correct computation in order to minimize total computation delay. Executors need to find out an optimal strategy to allocate its total deposits to different tasks and achieve maximal total wages received from clients. Either client or executor has infinite ways to allocate the deposits to tasks as long as the budget and deposit assigned to task satisfies some predefined conditions. Client assigns a number of budget to each task. The budget is divided into two parts, one part is paid to executor for the computation, the other part is used to hire a TTP with a probability. We show that the wage paid to executor is decided by the deposit the executor promises to pay and the budget assigned to task. Generally, the more deposit paid on a task, the lower probability that client raises a TTP and the lower cost spent on TTP which leads to higher wage received by the executor. The more budget assigned to a task, the lower probability of raising a TTP which results in a lower latency.

In order to help executor balance the deposit and wage and help client balance the budget and latency, we propose a game theoretical approach for client and executor to find optimal allocation strategy independently and then dynamically update its strategy based on the other's updated strategy. We analyze the game-based allocation algorithm and present that the scheme leads to optimal results compared with other straw-man allocation methods. Furthermore, we consider two types of task configurations to evaluate our scheme, the first setting is having all tasks predefined and known to client and executor. The other type has tasks coming dynamically and tasks are computed parallel but release in either sequential model or parallel model. The result indicates that release sequentially results in higher latency and higher wage while release tasks in parallel leads to lower latency and lower wage.
}

\section{Introduction}

\IEEEPARstart{T}{he} popularity of cloud services and decentralized platforms 
promote the development and prosperity of computation outsourcing. 
Clients export heavy computational tasks (such as data mining and machine learning) 
to executors who have available and intensive computational resources to handle them. 
The executors could be cloud service providers or nodes in the decentralized network. 
Executors aim to efficiently utilize the available computational resources 
and achieve optimal benefits from the computation. 
Clients desire to outsource tasks without paying more than 
a certain predefined budget and 
get correct computation results with as short delay as possible.  

In this paper, we focus on computation outsourcing in the cloud computing environment.
Specifically, we consider a system composed of 
a cloud service provider (abbreviated as {\em server} hereafter) and 
multiple cloud clients (abbreviated as {\em clients} hereafter). 
Each client submits to the server 
either a static set or a dynamic sequence of tasks.
The former corresponds to the application scenarios where the client processes tasks in batch
while the latter the scenarios where tasks are generated and processed in real time. 

To assure that the server returns correct computation results,
certain verifiable outsourced computation mechanisms should be in place. 
A large variety of schemes have been proposed in the literature to verify outsourced computation.
They can rely on 
cryptography \cite{gennaro2010non, parno2012delegate, catalano2013practical, parno2013pinocchio, abadi2016vd, costello2015geppetto, fiore2016hash, pepper, ginger, goldwasser2015delegating, ben2016interactive, wahby2017full}, 
trusted hardware \cite{brandenburger2018blockchain, xiao2019enforcing, cheng2019ekiden, tramer2018slalom},
redundant system (that includes at least one trusted server) \cite{canetti2011practical,avizheh2019verifiable}, 
game theory
\cite{nix2012contractual, optimal, belenkiy2008incentivizing, incentive, kupccu2015incentivized, dong2017betrayal, liu2018new, liu2020game},
or combinations of the above.
As briefly surveyed in Section~\ref{sec::related-works},
the approaches purely relying on cryptography or trusted hardware 
usually have high costs and/or low performance/scalability,
while the game-based approaches have gained more popularity
for their lower costs due to 
the practical assumption of economically-rational participants.
Hence, we also apply game theory in our study.

We adopt the basic model that 
each client outsources her tasks to only one server (without redundancy)
but with probabilistic auditing. 
Additionally, for each task,
the server is required to make a deposit,
which can be taken by the client when the server is found misbehaving; 
each client should prepare a {\em budget} that includes
the wage paid to the server (if the server is not found dishonest)
and the cost for hiring a trusted third party (TTP) 
to check the result returned by the server (i.e., auditing).
A relation among the deposit, wage and the auditing probability 
can be found such that,
the server's most beneficial strategy is to act honestly 
as long as the condition is satisfied. 

It is natural to assume the cloud has a certain fund to spend as deposits
for the tasks it take. However, the fund is limited at a time and
should be spent smartly so that the server 
can maximize its benefit, which we measure as the wage it can earn; 
it becomes more challenging when the server services a dynamic sequence of tasks, 
as it is unknown when new tasks will arrive and the sizes of the future tasks.
For a client, it is also natural to assume she has some fund 
to spend on the tasks she outsources. The client's fund is 
limited too, and thus should be smartly spent as well to maximize
her benefit, which we measure as the overall delay 
that she has to experience when waiting for her tasks to complete. 
Here, the client's spending strategy includes: first, 
how to distribute a given amount of fund 
to the tasks that are submitted simultaneously or within the same time window;
second, for each of the tasks, how to further divide the assigned budget
for paying the server's wage and for hiring a TTP respectively. 
{\em How can we smartly allocate the server and the clients' funds to maximize
their profits?} 
To the best of our knowledge, 
this is a question that has not been raised or answered 
in the literature. 
The focus of this paper is to formulate and solve this problem. 

We formulate the problem in two steps.
First, we formulate a per-task game-based outsourcing model.
Specifically, the model enforces a security relation among three components,
the server's deposit, the server's wage and the client's auditing probability, 
where the latter two determines the client's budget,
to ensure the server's best choice is to compute the task honestly.
In addition, the model has the attractive property that,
the wage and the auditing probability are not fixed but 
functions of the server's deposit and the client's budget;
the larger is the deposit and/or the budget,
the larger is the wage and the smaller is the auditing probability. 
Note that, larger wage and smaller auditing probability (and thus shorter delay)
are desired by the server and the client, respectively. 
In the second step, we formulate 
the interactions between the server and the clients
into an infinite extensive game with perfect information.
Within this game, 
the server and the clients are the parties;
the different ways to dividing the server's fund into the tasks' deposits
and to dividing the clients' funds into the tasks' budgets
are the parties' actions;
and the parties' utilities are defined as functions of the actions. 

We solve the problem in three steps. 
First,
we develop an algorithm that finds the Nash equilibria of the game, 
which is also the optimal solution that 
maximizes the server's wage meanwhile minimizes the client's delay, 
for the special setting where there is only one client who submits a static set of tasks.
Second,
we develop an algorithm that finds the Nash equilibria and also the optimal solution
for the more general setting where there are multiple clients
each submitting a static set of tasks. 
Finally, we develop heuristic algorithms,
which call the solution developed in the second step,
to solve the problem when there are multiple clients each submitting a dynamic sequence of tasks. 
Rigorous proofs have been developed to show the optimality 
of the solutions developed in the first two steps.
Extensive simulations have been conducted to evaluate the performance
of the solutions developed in the third step.

In the rest of the paper,
Section 2 surveys the related works.
Section 3 introduces the system model and 
the per-task game-based outsourcing model. 
Section 4 defines the game between the server and the clients.
Section 5, 6, and 7 develop the solutions in three steps. 
Finally, Section 8 concludes the paper and discusses the future work.


\section{Related Works}
\label{sec::related-works}

There has been extensive research on verifying outsourced computation. 
We briefly summarize these efforts as follows.

Many schemes 
\cite{gennaro2010non, parno2012delegate, catalano2013practical, parno2013pinocchio, abadi2016vd, costello2015geppetto, fiore2016hash, pepper, ginger, goldwasser2015delegating, ben2016interactive, wahby2017full}
have been designed based on cryptographic primitives/algorithms. 
For example, Gennaro et al.~\cite{gennaro2010non} 
formalize the notion of verifiable computation;  
they utilize Yao's garbled circuits to represent an outsourced function 
and homomorphic encryption to hide the circuits.
Parno et al.~\cite{parno2013pinocchio} exploit quadratic programs to 
encode computation and 
generate a fixed-length proof independent of the input/output size. 
Geppetto~\cite{costello2015geppetto} 
adapts multi-quadratic arithmetic programs to encode the computation, and 
constructs a commit-and-prove scheme to share data 
and prove execution of function. 
Various interactive proofs~\cite{pepper, ginger, goldwasser2015delegating, ben2016interactive, wahby2017full} 
have also been proposed. 
%
In general, 
the computational cost incurred by these schemes is high, 
which hampers their application in practice. 
As pointed out by Dong et al.~\cite{dong2017betrayal} 
as well as Walfish and Blumberg~\cite{walfish2015verifying}, 
their computational overhead 
could be $10^3-10^9$ higher than the cost to compute the task.
Hence, we do not adopt this approach in this paper.

Trusted Execution Environment (TEE)~\cite{brandenburger2018blockchain, xiao2019enforcing, cheng2019ekiden, tramer2018slalom} could be applied for verifiable computation. 
However, 
currently prevalent TEE such as Intel SGX is inappropriate 
for multi-threading and tasks requiring high demand of memory~\cite{tramer2018slalom}, 
which could make it inappropriate for enormous computation and validation for 
applications such as machine learning tasks~\cite{lu2018enabling} 
or heavy duty smart contracts~\cite{liu2020game} on blockchain. 
In this paper,
we assume a TTP 
may use TEE; 
but we aim to minimize the employment of TTP
and thus makes TEE an infrequently-used deterrence.

Alternatively, verifiable computation may be implemented based on redundancy
and the assumption of at least one server being honest
\cite{canetti2011practical,avizheh2019verifiable}. 
%
Hence, when the servers return the same final result, 
the result can be immediately accepted; otherwise,
the servers may be asked to provide intermediate computation results
that can be compared to identify the correct computation. 
However, 
the assumption of existing at least one honest server could be impractical
especially when only a small number of servers are employed.

Without the guarantee of trusted server employed,
game theoretic approaches 
\cite{nix2012contractual, optimal, belenkiy2008incentivizing, incentive, kupccu2015incentivized, dong2017betrayal, liu2018new, liu2020game}
have been proposed to prevent rational servers from misbehaving. 
For example, Nix and Kantarcioglu \cite{nix2012contractual} 
design two contracts for employing two servers. 
The games induced from the contracts 
have a Nash Equilibrium, 
where the servers behave honestly, 
as long as they cannot share information or collude with each other.
Pham et al.~\cite{optimal} consider two settings, 
employing a single server or 
employing two servers that do not share information or collude with each other.
They study how to coordinate the rewards, punishments, and auditing 
to make honest behavior the optimal strategy of the server(s). 

Collusion among the servers has also been studied 
\cite{belenkiy2008incentivizing, incentive, kupccu2015incentivized, dong2017betrayal}.
In particular, 
Belenkiy et al.~\cite{belenkiy2008incentivizing} 
study the proper fine-to-reward ratio 
to sabotage rational or malicious misbehavior and 
the impact of the probability that a hired server is misbehaving. 
Kupcu~\cite{kupccu2015incentivized} generalizes the work
by systematically studying 
the settings of hiring multiple servers with multiple types. 
Dong et al.~\cite{dong2017betrayal}
study the collusion and betrayal resulting from the interactions among 
the client-server and server-server contracts,
and propose the design of smart contracts that make use of the blockchain 
for the participants to escort and distribute funds. 
In their recent works~\cite{liu2018new,liu2020game}, 
Liu and Zhang point out the necessity of auditing 
in hiring multiple possibly-colluding servers, 
even when an additional contract exists to encourage betrayal;
they also consider that all hired servers may fully collude with each other,
and design smart contracts for efficient execution of heavy-duty smart contracts.

In this paper, we also adopts game theory for verifiable computation.
While the afore-discussed works consider the setting of
one client outsourcing tasks to one server or multiple servers,
we study the setting of multiple (cloud) clients outsourcing 
a static set or a dynamic sequence of tasks to one (cloud) server. 
Similar to most of the afore-discussed works,
the client needs to hire TTPs for probabilistic auditing,
and we develop the security condition regarding the required relations among
deposit, wage and auditing probability for each outsourced task.
Significantly different from the state of the art, 
our work further studies the optimal distribution of
the limited funds held by the server and the clients 
to maximize the system performance while 
meeting the security condition.

Game theory has also been applied in allocating cloud computation resources~\cite{teng2010new, teng2010resource, wei2010game, kaewpuang2013framework, pillai2014resource, xu2014game}.
For example, Wei et al.~\cite{wei2010game} propose to use the evolutionary game for cloud resource allocation.
Kaewpuang et al.~\cite{kaewpuang2013framework} formalize an optimization problem for allocating radio and computing resources for mobile devices and applies the core and Shapley value to coordinate the revenue allocation in the cooperative game. 
Pillai and Rao~\cite{pillai2014resource} study the problem of on-demand virtual machines allocation.
Xu and Yu~\cite{xu2014game} propose a fairness-utilization game theoretical allocation for multiple resources such as memory, CPU and storage. 
In this paper, we also apply game theory to formalize and solve 
the problems of optimally allocating resources. 
While aforementioned research aims to optimize the allocation of physical resources, 
our scheme is the first to allocate limited funds held by 
the cloud server and every client, respectively, to the tasks 
as deposits or task budgets, with the three-fold goals of
(i) assuring honest computation of the tasks, 
(ii) maximizing the server's overall wages, and
(iii) minimizing the delay for verification.

\comment{
In our paper, 
we also apply game theory for verifiable computation and efficient resource allocation. 
We focus on the setting where a cloud server services 
a static set or a dynamic sequence of tasks submitted by multiple clients.
Every client wishes to assure honest execution of each of her tasks
by additionally employs a TTP to re-compute the task with a certain probability. 
To implement this strategy, 
the cloud server makes a deposit for each task it takes,
each client allocates a budget 
(including the wage paid to the server and the possible cost for hiring TTP) 
for each task it submits, and
every party has its limited fund that can be used for either deposits or task budgets.
We study how the funds should be optimally allocated to achieve the three-fold goals:
(i) rational cloud server should honestly compute each tasks it takes,
(ii) the server's wages earned from computing the tasks is maximized,
and (iii) the overall delay experienced by each task for verifying her tasks is minimized. 
}

\section{System Architecture}

In this section,
we propose an architecture for 
game-based computation outsourcing to cloud server,
which is facilitated by an underlying blockchain.   

\subsection{System Model}

We consider a system consisting of 
a cloud service provider (called cloud server or {\em server} hereafter),
$m$ {\em clients} that need to outsource computation tasks to the server,
and some trusted third parties (called {\em TTP}s hereafter) which 
the clients can resort to for verifying outsourced computation. 
Figure~\ref{fig:system} illustrates the system architecture.

\begin{figure}[htb]
\includegraphics[width=\columnwidth]{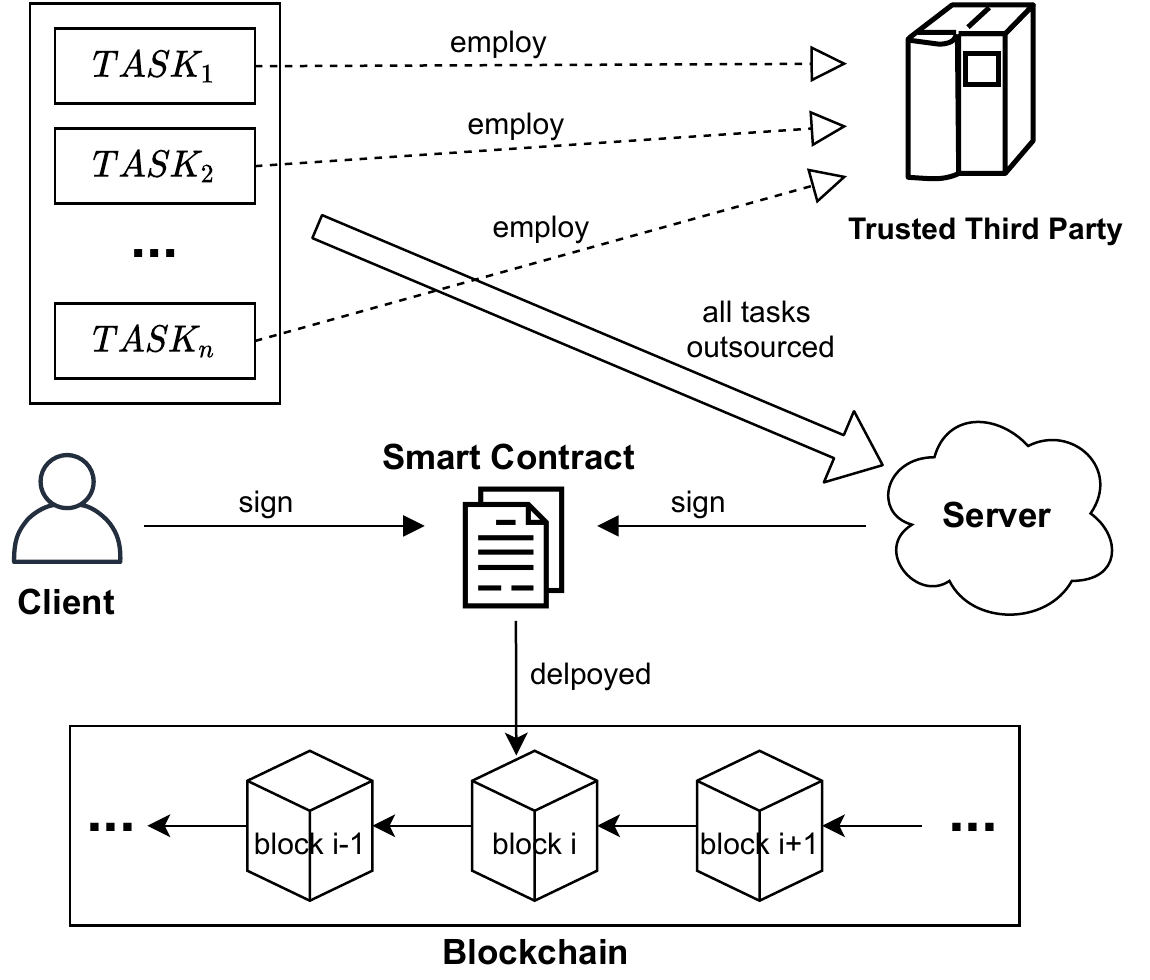}
\caption{System architecture}
\label{fig:system}
\end{figure}

The server, denoted as $S$, is not completely trusted and 
its execution of the tasks outsourced by the clients may not always be correct.
However, we assume the server is economically rational; that is,
it always aims to maximize its profit and 
will not misbehave if that would cause penalty. 
As to be elaborated in Section~\ref{sec:outsourcing-model},
we introduce a game-based approach to guarantee that 
the server honestly executes the outsourced tasks.    
We assume that the server is willing to use a certain amount of fund 
as deposit to assure its client of its honest behavior. 

We denote the $m$ clients as $C_1, \cdots, C_m$. 
The tasks outsourced by each client $C_i$ are denoted as 
$t_{i,j}$ for $j=1,\cdots,n_i$, where $n_i$ is the number of such tasks.
Each task $t_{i,j}$ is associated with two costs denoted as $c_{i,j}$ and $\hat{c}_{i,j}$,
where $c_{i,j}$ is the server's cost to execute the task and 
$\hat{c}_{i,j}$ is each TTP's cost to execute the task. 
To simplify the presentation, 
we assume the execution time is proportional to the costs; 
that is, assuming $k$ is a certain constant, 
the server's execution time of the task is $k\cdot c_{i,j}$
and each TTP's execution time of the task is $k\cdot\hat{c}_{i,j}$. 
Each client $C_i$ allocates a budget $b_{i,j}$ for each task $t_{i,j}$,
where $b_{i,j}\geq c_{i,j}$ so that the server is willing to take the task.
        
Each TTP can be hired at the price of $\hat{c}_{i,j}$ by a client 
to check if the server's execution is correct 
via re-execution.
A TTP can also be a cloud server that 
has a trusted execution environment (TEE) such as Intel SGX enclave.
  
Finally, we assume that the server, the clients and the TTPs 
can access a blockchain system so that no
any centralized trusted authority is required.

\subsection{Per-task Game-based Outsourcing Model}
\label{sec:outsourcing-model}

To ensure that the server honestly executes tasks,
we adopt a game theoretic approach as follows.
For each task $t_{i,j}$, 
the server should make a deposit of $d_{i,j}$ and
client $C_i$ should promise a budget with 
a certain expected value of $b_{i,j}$.

After the client outsources $t_{i,j}$ to the server,
with a probability denoted as $p_{i,j}$ 
it also hires a TTP to execute the task.
After the client has received a result of computation task from the server and/or the TTP,
funds are distributed between the client and the server as follows:
If no TTP is hired, or the results returned by the server and the hired TTP are the same, 
the client should pay a wage denoted as $w_{i,j}$, where $w_{i,j} \geq c_{i,j}$, 
to the server, and the server should also be returned with its deposit $d_{i,j}$.
If the results returned by the server and the TTP are different,
deposit $d_{i,j}$ should be given to the client.    
Hence, 
\begin{equation}\label{eq:budget}
b_{i,j}=w_{i,j}+p_{i,j}\cdot\hat{c}_{i,j}. 
\end{equation}
Also, as stated in the following theorem,
$p_{i,j}\geq\frac{c_{i,j}}{w_{i,j}+d_{i,j}}$
is the sufficient condition to
deter the server from misbehaving and 
ensure it honestly executes task $t_{i,j}$.

\begin{theorem}
\label{theo:pij}
As long as $w_{i,j}\geq c_{i,j}$ and $p_{i,j}\geq\frac{c_{i,j}}{w_{i,j}+d_{i,j}}$, 
an economically rational server must execute task $t_{i,j}$ honestly and 
submit a correct result to the client.
\end{theorem}
\begin{proof}
(sketch). 
%
If the server behaves dishonestly, 
it should lose its deposit with probability $p_{i,j}$ while 
still receive a wage of $w_{i,j}$ with probability $1-p_{i,j}$; 
hence, its expected payoff is $-p_{i,j}*d_{i,j}+(1-p_{i,j})*w_{i,j}$. 
If the server behaves honestly, 
it should receive a wage of $w_{i,j}$ while pay the honest execution cost of $c_{i,j}$;
hence, its expected payoff is $w_{i,j}-c_{i,j}$.
For a rational server to behave honestly,   
it must hold that $-p_{i,j}*d_{i,j}+(1-p_{i,j})*w_{i,j}\leq w_{i,j}-c_{i,j}$. 
Therefore, 
$p_{i,j}\geq\frac{c_{i,j}}{w_{i,j}+d_{i,j}}$. 
\end{proof}

\section{Optimization Problem}

To efficiently implement the proposed architecture,
it is desired to optimize the allocation of
the cloud server's fund for deposits and the clients' funds for tasks,
to achieve the following dual goals:
the server can {\em maximize} its wages earned from the clients; 
each clients can {\em minimize} the total time to 
verify the results of its tasks outsourced to the server.

\subsection{Game between The Server and The Clients}

We model the interactions between the server and the clients
as an infinite extensive game with perfect information,
%
denoted as $G=(P, A, U)$.
\begin{itemize}
    \item $P=\{S, C_1, \cdots, C_m\}$:
        the set of players. 
    \item $A$:
        the set of actions taken by the players, including
                (i) all possibilities that each $C_i$ can split its budget $b_i$ to $n_i$ tasks
                and (ii) all possibilities that $S$ can split its deposit fund $d$ to the 
                    $n=\sum_{i=1}^{m}n_i$
                tasks.
                Hence, the action set each $C_i$ can take is denoted as
                $A_{c,i} = \{(b_{i,1},\cdots,b_{i,n_i})~|~\sum_{j=1}^{n_i}b_{i,j}=b_i\}$,
                where $b_i$ is $C_i$'s total budget for its tasks, and 
                each action $(b_{i,1},\cdots,b_{i,n_i})$ is one possible division of $b_i$ to $n_i$ tasks;
                the action set $S$ can take is denoted as
                $A_s = \{(d_{1,1},\cdots,d_{m,n_m})\nonumber
                    |\sum_{i=1}^{m}\sum_{j=1}^{n_i}d_{i,j}=d\}$,
                where $d$ is the server's fund for deposits, and
                each action $(d_{1,1},\cdots,d_{m,n_m})$ 
                is one possible division of $d$ to $n$ tasks.
    \item $U=\{U_s,U_{c,1},\cdots,U_{c,m}\}$:
        the players' utility functions. 
        %
\end{itemize}


\subsection{Constraints on Budgets and Deposit}

%
According to the above definitions of the clients' and the server's actions,
the following constraints are obvious:
\begin{equation}\label{eq:sum_bij}
\sum_{j=1}^{n_i}b_{i,j} = b_i,~\forall~i\in\{1,\cdots,m\},
\end{equation}
and
\begin{equation}\label{eq:sum_d}
\sum_{i=1}^{m}\sum_{j=1}^{n_i}d_{i,j} = d.
\end{equation}
For each task $t_{i,j}$,
the server's deposit for it should be at least $\hat{c}_{i,j}$,
to compensate client $C_i$'s cost for hiring a TTP if the server is found dishonest. 
Hence, we have the following constraint:  
\begin{equation}\label{eq:deposit-constraint}
d_{i,j}\geq\hat{c}_{i,j}.
\end{equation}

Regarding budget $b_{i,j}$ for $t_{i,j}$, 
according to Eq.~(\ref{eq:budget}), it includes
wage $w_{i,j}$ paid to the server for honest computation and
the expected cost to hire TTP. 
%
First, based on Theorem~\ref{theo:pij} and that TTP should be hired as infrequently as possible, 
we set
\begin{equation}\label{eq:pij} 
p_{i,j}=\frac{c_{i,j}}{w_{i,j}+d_{i,j}}.
\end{equation} 
%
%
Second, $w_{i,j} \geq c_{i,j}$ must hold to incentive the server.
Because $b_{i,j}=w_{i,j} + \frac{c_{i,j}\hat{c}_{i,j}}{w_{i,j}+d_{i,j}}$, 
which is from Equations (\ref{eq:budget}) and (\ref{eq:pij}),
is an increasing function of $w_{i,j}$, 
it holds that $w_{i,j} \geq c_{i,j}$ is equivalent to 
    $b_{i,j} \geq c_{i,j} + \frac{c_{i,j}\hat{c}_{i,j}}{c_{i,j} + d_{i,j}}$.
Further due to $d_{i,j}\geq \hat{c}_{i,j}$, we set
\begin{equation}\label{eq:budget-constraint}
b_{i,j}\geq c_{i,j} + \frac{c_{i,j}\hat{c}_{i,j}}{c_{i,j}+\hat{c}_{i,j}},
\end{equation}
which implies $b_{i,j} \geq c_{i,j} + \frac{c_{i,j}\hat{c}_{i,j}}{c_{i,j} + d_{i,j}}$ 
and $w_{i,j} \geq c_{i,j}$.


\subsection{Utility Functions}

In the game,
server $S$ aims to maximize its total wage $\sum_{i=1}^{m}\sum_{j=1}^{n_i}w_{i,j}$ 
under the constraints of
(\ref{eq:budget}), (\ref{eq:pij}), (\ref{eq:sum_bij}), (\ref{eq:sum_d}), 
(\ref{eq:deposit-constraint}) and (\ref{eq:budget-constraint}).
From (\ref{eq:budget}) and (\ref{eq:pij}), it holds that $b_{i,j}=w_{i,j}+\frac{c_{j,j}\hat{c}_{i,j}}{w_{i,j}+d_{i,j}}$ which can be written as 
a quadratic equation for variable $w_{i,j}$ as  $w_{i,j}^2+w_{i,j}(d_{i,j}-b_{i,j})+c_{j,j}\hat{c}_{i,j}-b_{i,j}d_{i,j}=0$. 
Then, we have 
\begin{equation}\label{eq:wage}
    w(b_{i,j},d_{i,j}) 
    = \frac{b_{i,j}-d_{i,j}+\sqrt{(b_{i,j}+d_{i,j})^2-4c_{i,j}\hat{c}_{i,j}}}{2}.
\end{equation}
%
%
Therefore, 
%
the utility function of server $S$ is 
\begin{eqnarray}
    U_s(A_s, A_{c,1}, \cdots, A_{c,m}) 
    =\sum_{i=1}^{m}\sum_{j=1}^{n_i}w(b_{i,j},d_{i,j}).
\end{eqnarray}
Each $C_i$ aims to minimize the expected time for verifying its $n_i$ tasks. 
For each task $t_{i,j}$, the expected verification time, denoted as $T_{i,j}$, is
\begin{eqnarray}
    &&T_{i,j}(b_{i,j},d_{i,j}) 
    = k\cdot (b_{i,j} - w(b_{i,j},d_{i,j}))\nonumber\\
    &=& k\cdot\frac{b_{i,j}+d_{i,j}-\sqrt{(b_{i,j}+d_{i,j})^2-4c_{i,j}\hat{c}_{i,j}}}{2}.
\end{eqnarray}
Then, the utility function of client $C_i$ is defined as
\begin{eqnarray}
    U_{c,i}(A_s,A_{c,i}) 
    =\sum_{j=1}^{n_i} [T_{i,j}(b_{i,j},d_{i,j})].
\end{eqnarray}

\subsection{Nash Equilibrium of the Game}


A Nash equilibrium of the game is a combination of action, denoted as 
$(A*_s,A*_{c,1},\cdots,A*_{c,m})$, 
taken by the server and the clients respectively, such that:
for the server and any $A_s\not= A*_s$,
$U_{s}(A_s,A*_{c,1},\cdots,A*_{c,m})\leq U_{s}(A*_s,A*_{c,1},\cdots,A*_{c,m})$;
for each client $i\in\{1,\cdots,m\}$ and any $A_{c,i}\not= A*_{c,i}$, 
$U_{c,i}(A*_s,A_{c,i})\leq U_{c,i}(A*_s,A*_{c,i})$.

%

\section{Single Client with Static Set of Tasks}
\label{sec:single_clt}

We first study the optimization problems in
the context that server $S$ interacts with only one client, 
denoted as $C_i$.


\subsection{Client's Optimization Problem}

The client's purpose is to minimize her utility, 
i.e., the expected time for verifying her tasks.
Hence, the client's optimization problem is as follows.
(Note: parameter $k$ is ignored for the simplicity of exposition.)

\begin{eqnarray}\label{eq:min_delay}
        &\min&\sum_{j=1}^{n_i}\frac{b_{i,j}+d_{i,j}-\sqrt{(b_{i,j}+d_{i,j})^2-4c_{i,j}\hat{c}_{i,j}}}{2}\label{eq:min_delay-objective}\\
        &s.t.& \nonumber\\
        &&\sum_{j=1}^{n_i} d_{i,j}=d;~\sum_{j=1}^{n_i} b_{i,j}=b_i; \label{eq:min_delay-c1-c2}\\
        &&d_{i,j}\geq \hat{c}_{i,j};~b_{i,j}\geq c_{i,j}+\frac{c_{i,j}\hat{c}_{i,j}}{c_{i,j}+\hat{c}_{i,j}}. \label{eq:min_delay-c3-c4}
\end{eqnarray}

\subsection{Server's Optimization Problem}

The server's purpose is also to maximize its utility, 
i.e., the total wage earned from the client.
Hence, its optimization problem is as follows.

\begin{eqnarray}\label{eq:max_wage}
        &\max&\sum_{j=1}^{n_i}\frac{b_{i,j}-d_{i,j}+\sqrt{(b_{i,j}+d_{i,j})^2-4c_{i,j}\hat{c}_{i,j}}}{2}\label{eq:max_wage-objective}\\
        &s.t.&~\mbox{constraints}~(\ref{eq:min_delay-c1-c2}),~(\ref{eq:min_delay-c3-c4}).\nonumber
\end{eqnarray}
Note that,
the sum of the objective functions of the above two optimization problems is
\[
(\ref{eq:min_delay-objective}) + (\ref{eq:max_wage-objective}) = \sum_{j=1}^{n_i}b_{i,j} = b_i.
\]
Hence,
the objective function of the server's optimization problem can be re-written to
\[
\max~b_i-\sum_{j=1}^{n_i}\frac{b_{i,j}+d_{i,j}-\sqrt{(b_{i,j}+d_{i,j})^2-4c_{i,j}\hat{c}_{i,j}}}{2},
\]
which is further equivalent to
\[
\min~\sum_{j=1}^{n_i}\frac{b_{i,j}+d_{i,j}-\sqrt{(b_{i,j}+d_{i,j})^2-4c_{i,j}\hat{c}_{i,j}}}{2}.
\]
Therefore,
the above two optimization problems are equivalent.
That is, a solution to the client's optimization problem is
also a solution to the server's optimization problem,
and thus is also the Nash equilibrium of the game 
between the server and the client.

\subsection{Proposed Algorithm}

Due to the equivalence of the above two optimization problems,
we only need to solve one of them. Next,
we develop the algorithm, formally presented in Algorithm~\ref{alg:SingleClient},
to find the solution to the client's optimization problem.

The core of the algorithm is to solve the following optimization problem,
which is re-written from the afore-presented client's optimization problem.
\begin{eqnarray}\label{eq:min_delay-revised}
        &\min&\sum_{j=1}^{n_i}f(s_{i,j},i,j) \label{eq:fx}\\
        &where&~f(x,i,j)= \frac{x-\sqrt{x^2-4c_{i,j}\hat{c}_{i,j}}}{2} \nonumber\\
        &s.t.&
        s_{i,j} = b_{i,j}+d_{i,j} \nonumber\\
        &&\mbox{constraints}~(\ref{eq:min_delay-c1-c2})~and~(\ref{eq:min_delay-c3-c4}). \nonumber
\end{eqnarray}
Note that, $f(x,i,j)$ is the client's utility associated with each task $t_{i,j}$,
when the task is assigned with $x$ as the sum of $b_{i,j}$ and $d_{i,j}$.
In the algorithm, we also use a partial derivative function of $f(x,i,j)$, 
which is defined as
\begin{equation}\label{eq:fpartial}
f'(x,i,j) = \frac{\partial f(x,i,j)}{\partial x}.
\end{equation}

After the client and server exchange with each other their budget and deposit (i.e., $b_i$ and $d$),
they each run Algorithm~\ref{alg:SingleClient} to 
optimally allocate $b_i+d$ to the $n_i$ tasks,
i.e., each task $t_{i,j}$ is assigned with budget $b_{i,j}$ and deposit $d_{i,j}$
where $\sum_{j=1}^{n_i}b_{i,j}=b_i$ and $\sum_{j=1}^{n_i}d_{i,j}=d$,
with the goal of maximizing the client's utility.
Intuitively, 
the algorithm runs in the following three phases: 

In the first phase, 
each task $t_{i,j}$ is assigned an initial value for $s_{i,j}$,
which denotes the sum of $b_{i,j}$ and $d_{i,j}$.
Here, the initial value is set to $\hat{c}_{i,j} + c_{i,j} + \frac{c_{i,j}\hat{c}_{i,j}}{c_{i,j}+\hat{c}_{i,j}}$
in order to satisfy constraints (\ref{eq:min_delay-c3-c4}).
After this phase completes,
$s=b_i+d - \sum_{j=1}^{n_i}(\hat{c}_{i,j}+c_{i,j}+\frac{c_{i,j}\hat{c}_{i,j}}{c_{i,j}+\hat{c}_{i,j}})$
remains to be allocated in the second phase.

In the second phase, 
$s$ is split into units each of size $\delta$ and 
the units are further assigned to the tasks step by step.
Specifically, with each step, 
one remaining unit is assigned to task $t_{i,j}$ whose
$f'(s_{i,j},i,j)$ is the minimal among all the tasks;
this way, the units are assigned in a greedy manner 
to maximize the total utility of all the $n_i$ tasks. 

After the $b_i+d$ have been greedily assigned to all the tasks,
in the third phase, 
$s_{i,j}$ is further split into $b_{i,j}$ and $d_{i,j}$
such that, the shorter verification time a task has,
the larger deposit is assigned to it.
This way, the server's deposit can be reclaimed as soon as possible from the tasks.

\begin{algorithm}[htb]
\caption{Optimizing Resource Allocation (Server $S$ v.s. Client $C_i$ with Static Task Set)}\label{alg:SingleClient}

{\bf Input:}
\begin{itemize}
\item
$b_i$: total budget of client $C_i$;
\item
$d$: total deposit of server $S$;
\item
$n_i$: total number of tasks;
\item
task set $\{t_{i,1},\cdots,t_{i,n_i}\}$ and
associated costs $\{c_{i,1},\cdots,c_{i,n_i}\}$ and $\{\hat{c}_{i,1},\cdots,\hat{c}_{i,n_i}\}$.
\end{itemize}

{\bf Output:} $\{b_{i,1},\cdots,b_{i,n_i}\}$ and $\{d_{i,1},\cdots,d_{i,n_i}\}$. \\

{\bf Phase I:} Initialization.

\begin{algorithmic}[1]
\For{$j\in\{1,\cdots,n_i\}$}
    \State
    $s_{i,j}\leftarrow (\hat{c}_{i,j}+c_{i,j}+\frac{c_{i,j}\hat{c}_{i,j}}{c_{i,j}+\hat{c}_{i,j}})$
    \Comment{meet constraints~(\ref{eq:min_delay-c3-c4})}
\EndFor
\end{algorithmic}

~~\\
{\bf Phase II:} Greedy Allocation of the Remaining Fund.
\begin{algorithmic}[1]
\State $s\leftarrow [b_i+d - \sum_{j=1}^{n_i}(\hat{c}_{i,j}+c_{i,j}+\frac{c_{i,j}\hat{c}_{i,j}}{c_{i,j}+\hat{c}_{i,j}})]$
\Comment{remaining fund to distribute}
\While{$s\geq\delta$} \Comment{distribute remaining fund in unit $\delta$}
    \State $j*\leftarrow \operatorname*{arg\,min}_{j\in\{1,\cdots,n_i\}} f'(s_{i,j},i,j)$
    \State $s_{i,j*}\leftarrow (s_{i,j*}+\delta)$;~$s\leftarrow (s-\delta)$
\EndWhile
\end{algorithmic}

~~\\
{\bf Phase III:} Splitting Sum to Budget/Deposit.
\begin{algorithmic}[1]
\State $d'\leftarrow d-\sum_{j=1}^{n_i}\hat{c}_{i,j}$
\State $tempSet=\{1,\cdots,n_i\}$
\While{$tempSet\not=\emptyset$}
    \State $j*=\operatorname*{arg\,min}_{j\in\{1,\cdots,n_i\}} f(s_{i,j},i,j)$
        \Comment{find the task with the shortest verification time}
    \State $x\leftarrow\min\{d',s_{i,j*}-\hat{c}_{i,j*}-(c_{i,j}+\frac{c_{i,j}\hat{c}_{i,j}}{c_{i,j}+\hat{c}_{i,j}})\}$
    \State $d_{i,j*}\leftarrow (\hat{c}_{i,j*}+x)$ \Comment{assign as much deposit to task with the shortest verification time}
    \State $b_{i,j*}\leftarrow (s_{i,j*}-d_{i,j*})$
    \State $d'\leftarrow (d'-x)$
    \State $tempSet\leftarrow (tempSet-\{j*\})$
\EndWhile
\end{algorithmic}
\end{algorithm}

\subsection{Analysis}

We now prove that Algorithm~\ref{alg:SingleClient} 
finds an optimal solution for 
the client's optimization problem 
(which is also a solution for the server's optimization problem), 
and the solution is a Nash equilibrium of 
the game between the client and the server. 
We develop the proof in the following steps.
First, we introduce an optimization problem, as follows, 
which is relaxed from (\ref{eq:min_delay-revised}):

\begin{eqnarray}\label{eq:min_delay-revised2}
        &\min&\sum_{j=1}^{n_i}f(s_{i,j},i,j) \label{eq:fx2}\\
        &where&~f(x,i,j)= \frac{x-\sqrt{x^2-4c_{i,j}\hat{c}_{i,j}}}{2} \nonumber\\
        &s.t.& 
        -s_{i,j} + c_{i,j} + \frac{c_{i,j}\hat{c}_{i,j}}{c_{i,j}+\hat{c}_{i,j}} + \hat{c}_{i,j}\leq 0.
        \label{eq:constraint-revised2}
\end{eqnarray}
Second, we derive the following Lemmas:
\begin{lemma}
\label{lem:singleClient-1}
The optimization problem defined in (\ref{eq:min_delay-revised2}) has a unique solution.
\end{lemma}
\begin{proof}
As elaborated in Appendix~1,
we show that 
the objective function is strictly convex and 
all the constraints are convex. 
\end{proof}
%
\begin{lemma}
\label{lem:singleClient-2}
Phases I and II of Algorithm~\ref{alg:SingleClient} finds the unique solution to 
the optimization problem defined in (\ref{eq:min_delay-revised2}).
\end{lemma}
\begin{proof}
See Appendix~2.
\end{proof}
%
\begin{lemma}
\label{lem:singleClient-3}
Phase III of Algorithm~\ref{alg:SingleClient} converts a solution of 
the optimization problem defined in (\ref{eq:min_delay-revised2})
into a solution of the optimization problem defined in (\ref{eq:min_delay-revised}).
\end{lemma}
\begin{proof}
See Appendix~3.
\end{proof}
Based on the above lemmas, we therefore have the following theorem:
\begin{theorem}\label{theo:singleClientOpt}
Algorithm~\ref{alg:SingleClient} finds a solution of 
the optimization problem defined in (\ref{eq:min_delay-revised}).
\end{theorem}
Finally, we also prove in Appendix~4 the following theorem: 
\begin{theorem}\label{theo:singleClientNE}
Algorithm~\ref{alg:SingleClient} finds a Nash equilibrium of
the game between server $S$ and client $C_i$.
\end{theorem}

\comment{
\begin{lemma}
\label{lem:alg1-proof}
Algorithm~\ref{alg:SingleClient} solves the optimization problem defined in (\ref{eq:min_delay-revised}).
\end{lemma}
\begin{proof}
Based on Appendix~\ref{appendix:single_client_unique_point}, $f(x, i, j)$ has unique optimum point. 
In order to prove Algorithm~\ref{alg:SingleClient} solves the optimization problem that is to prove the Phase II greedy algorithm leads to the unique optimum point. 
We prove it by contradiction.

Assume the optimal allocation strategy for the remaining fund is $O=\{o_{i,1},o_{i,2},\cdots,o_{i,n_i}\}$, the greedy algorithm in Phase II leads to strategy $G=\{g_{i,1},g_{i,2},\cdots,g_{i,n_i}\}$. 
For contradiction, we assume that $G \neq O$. Then there exist some tasks (at least one) such that they receive more/less budget and deposit allocation than the optimal solution. 
Let $M$/$L$ denote the sets of such tasks respectively. Let $C= \{c_{i,1}, c_{i,2},\dots,c_{i,n_i}\}$ be the common parts of $G$ and $O$. 
That is $c_j = min(g_{i,j}, o_{i,j})$ for $j\in\{1,\cdots,n_i\}$. 
Let $V_C$, $V_G$ and $V_O$ be the values of the objective function of \eqref{eq:min_delay-revised} given by allocation $C$, $G$ and $O$ respectively. 
Notice that $O$ is the unique optimal solution of \eqref{alg:SplitDeposit}. This implies that
\begin{equation}
\label{eq:V_O and V_G}
    V_G > V_O. 
\end{equation}
We define $diff_{C,O}$ and $diff_{C,G}$ as follows:
\begin{equation}
    diff_{C,O} = V_O - V_C = \sum _{t_{i,j} \in L} \int_{c_{i,j}}^{o_{i,j}} f'(x,i,j) \,dx, \nonumber
\end{equation}
\begin{equation}
    diff_{C,G} = V_G - V_C = \sum _{t_{i,j} \in M} \int_{c_{i,j}}^{g_{i,j}} f'(x,i,j) \,dx. \nonumber 
\end{equation}

And $f'(x,i,j)$ is the partial derivative function of $f(x,i,j)$ which has
\begin{equation}
\label{eq:f(x)_first_derivative}
   f'(x,i,j)=\frac{1}{2}(1-\frac{x}{\sqrt{x^2-4c_{i,j}\hat{c}_{i,j}}})<0 
\end{equation}
It indicates that $f(x,i,j)$ is decreasing monotonously. 
Let $f''(x,i,j)$ be the partial derivative function of $f'(x,i,j)$ and it holds
 \begin{equation}
 \label{eq:f(x)_second_derivative}
     f''(x,i,j)=\frac{2c_{i,j}\hat{c}_{i,j}}{{\sqrt{x^2-4c_{i,j}\hat{c}_{i,j}}}^3}>0.
 \end{equation}
  In other words, $f(x,i,j)$ decreases monotonically but slower and slower as $x$ increases (see Fig \ref{fig:f(x,i,j)}). 
\begin{figure}[htb!]
    \centering
    \includegraphics[width=\columnwidth]{cpt-outsourcing/f(x,i,j).png}
    \caption{Illustration of $f(x, i, j)$.}
    \label{fig:f(x,i,j)}
\end{figure}

According to \eqref{eq:f(x)_first_derivative} and \eqref{eq:f(x)_second_derivative}, $f'(x,i,j)$ increases as $x_{i,j}$ increases. Therefore, for each task $t_{i,j} \in L$, $f'(x,i,j) > f'(c_{i,j},i,j)$ for $c_{i,j} < x \leq o_{i,j}$. Let $f'_{omin} = min(f'(c_{i,j},i,j))$ for all tasks $t_{i,j} \in L$, then we have 
\begin{align}
\label{eq:diff_{C,O}}
    diff_{C,O} &= \sum _{t_{i,j} \in L} \int_{c_{i,j}}^{o_{i,j}} f'(x,i,j) \,dx \nonumber\\
    & > \sum_{t_{i,j} \in L} \int_{c_{i,j}}^{o_{i,j}} f'(c_{i,j},i,j) \,dx \nonumber\\
    & = \sum_{t_{i,j} \in L} ((o_{i,j} - c_{i,j}) \times f'(c_{i,j},i,j)) \nonumber\\
    & \geq  \sum_{t_{i,j} \in L} ((o_{i,j} - c_{i,j}) \times f'_{omin} \nonumber\\
    &= \Delta \times f'_{omin},
\end{align}
where $\Delta = \sum_{t_{i,j} \in L} ((o_{i,j} - c_{i,j})$. 
Notice that the total amount of budget and deposit is fixed, we have $\Delta = \sum_{t_{i,j} \in L} (o_{i,j} - c_{i,j}) = \sum_{t_{i,j} \in M} (g_{i,j} - c_{i,j})$. 
Meanwhile, let $f'_{gmax} = max(f'(x,i,j))$ for all tasks $t_{i,j} \in M$ and $c_{i,j} \leq x \leq g_{i,j}$. 
Since our greedy algorithm will pick the task with the minimal first derivative at each step and it chooses tasks in $M$ over tasks in $L$, we have $f'_{gmax} \leq f'_{omin}$. Therefore we have 
\begin{align}
\label{eq:diff_{C,G}}
    diff_{C,G} &= \sum _{t_{i,j} \in M} \int_{c_{i,j}}^{g_{i,j}} f'(x,i,j) \,dx \nonumber \\
    & < \sum_{t_{i,j} \in M} \int_{c_{i,j}}^{g_{i,j}} f'_{gmax} \,dx \nonumber\\
    & = \sum_{t_{i,j} \in L} ((o_{i,j} - c_{i,j}) \times f'_{gmax} \nonumber\\
    & = \Delta \times f'_{gmax} \nonumber \\
    & \leq \Delta \times f'_{omin}\nonumber\\
    & < diff_{C,O},
\end{align}
which means
\begin{equation}
    V_G < V_O,
\end{equation}
since $V_G = V_C + diff_{C, G}$ and $V_O = V_C + diff_{C,O}$ . However, this result contradicts with \eqref{eq:V_O and V_G}. Therefore, our assumption is incorrect.  

That is, the $s_{i,j}$ derived from Phase II in algorithm~\ref{alg:SingleClient} produces minimal total for client. In Phase III, the strategy to split budget and deposit meets the constraints $s_{i,j}=b_{i,j}+d_{i,j}$, ~(\ref{eq:min_delay-c1}),~(\ref{eq:min_delay-c2}),~(\ref{eq:min_delay-c3}),~and~(\ref{eq:min_delay-c4}). Therefore, Algorithm~\ref{alg:SingleClient} solves the problem (\ref{eq:min_delay-revised}).

\end{proof}
\begin{corollary}
\label{cly:singleClient}
Algorithm~\ref{alg:SingleClient} solves the optimization problem defined in (\ref{eq:max_wage}).
\end{corollary}

\begin{theorem}
Algorithm~\ref{alg:SingleClient} finds a Nash Equilibrium of the game between server $S$ and one client.
\end{theorem}
\begin{proof}
In order to prove Algorithm~\ref{alg:SingleClient} finds a Nash Equilibrium, we need to prove that neither client nor server has incentive to change the strategy $A_{c,i}=\{b_{i,1},\cdots,b_{i,n_i}\}$ and $A_{s,i}=\{d_{i,1},\cdots,d_{i,n_i}\}$ generated by algorithm~\ref{alg:SingleClient} if the other player stays the same. 
As proved in Lemma~\ref{alg:SingleClient} and Appendix~\ref{appendix:single_client_unique_point}, algorithm~\ref{alg:SingleClient} outputs optimum point $\{s_{i,1},\cdots,s_{i,n_1}\}$ for the optimization problem (\ref{eq:min_delay-revised}) where $s_{i,j}=b_{i,j}+d_{i,j}$ for $j\in\{1,\cdots,n_i\}$ and the optimum point is unique.

If server keeps the strategy $A_{s,i}=\{d_{i,1},\cdots,d_{i,n_i}\}$, client changes its strategy from $A_{c,i}=\{b_{i,1},\cdots,b_{i,n_i}\}$ to $A'_{c,i}=\{b'_{i,1},\cdots,b'_{i,n_i}\}$ where $A_{c,i}\neq A'_{c,i}$. 
Assume $\{s'_{i,1},\cdots,s'_{i,n_i}\}$ is the point generated under the strategy $A_{s,i}$ and $A'_{c,i}$ where $s'_{i,j}=b'_{i,j}+d_{i,j}$ for $j\in\{1,\cdots,n_i\}$.
That is, there exists at least one task $t_{i,j}$ which has $b'_{i,j}\neq b_{i,j}$ and $s_{i,j}\neq s'_{i,j}$.
Then the utility for client under strategy set $(A_{s,i}, A'_{c,i})$ is
$U'_{c,i}=\sum_{j=1}^{n_i}\frac{s'_{i,j}-\sqrt{{s'_{i,j}}^2-4c_{i,j}\hat{c}_{i,j}}}{2}$ which is greater than $U_{c,i}=\sum_{j=1}^{n_i}\frac{s_{i,j}-\sqrt{s_{i,j}^2-4c_{i,j}\hat{c}_{i,j}}}{2}$, the unique minimal value produced by algorithm~\ref{alg:SingleClient}.
That is, deviate from $\{d_{i,j},\cdots,d_{i,n_i}\}$ will not produce lower total checking time for client.

If server deviates its strategy from $A_{s,i}$ to $A'_{s,i}=\{d'_{i,1},\cdots,d'_{i,n_i}\}$ where $A_{s,i}\neq A'_{s,i}$ and client maintains strategy $A_{c,i}=\{b_{i,1},\cdots,b_{i,n_i}\}$, 
then the total wage for server on actions set $(A'_{s,i}, A_{c,i})$ is $U'_s=\sum_{j=1}^{n_i}b_{i,j}-\sum_{j=1}^{n_i}\frac{s'_{i,j}-\sqrt{{s'_{i,j}}^2-4c_{i,j}\hat{c}_{i,j}}}{2}$ where $s'_{i,j}=d'_{i,j}+b_{i,j}$ for $j\in\{1,\cdots,n_i\}$. 
The total wage for server under strategy $(A_{s,i}, A_{c,i})$ is $U_{s} = \sum_{j=1}^{n_i}b_{i,j}-\sum_{j=1}^{n_i}\frac{s_{i,j}-\sqrt{{s_{i,j}}^2-4c_{i,j}\hat{c}_{i,j}}}{2}$.
Since $\{s_{i,1},\cdots,s_{i,n_1}\}$ are unique optimum points for the optimization problem (\ref{eq:min_delay-revised}), then $U'_s<U_s$ which indicates that server has no incentive to deviate from $A_{s,i}$.

Therefore, the strategy generated by algorithm~\ref{alg:SingleClient} is a Nash Equilibrium. 

\end{proof}
}


\section{Multiple Clients with Static Task Sets}

In this section, 
we study the optimization problem in the context that
server $S$ services $m$ clients $C_1,\cdots,C_m$.
Different from the previous context of single client,
optimizing for the server's utility and 
for each client's utility are not equivalent.
So we cannot solve it in one step.
Instead, we tackle the problem in two steps:
we first optimize for the server's utility,
which produces an allocation of the server's deposits to the clients; 
then, we optimize for each client's utility
based on the client's budget and 
the deposit allocated by the server. 
%

\subsection{Algorithm}


We propose an algorithm, 
formally presented in Algorithm~\ref{alg:MultiClientStatic},
which runs in the following two steps.

First,
we solve the server's optimization problem,
which produces the optimal allocation of 
the server's deposits to the clients
that maximizes the server's wages.
Thus, the optimization problem can be defined as follows:
\begin{eqnarray}\label{eq:max_wage_multiple}
        &\max&\sum_{i=1}^{m}\sum_{j=1}^{n_i}w(b_{i,j},d_{i,j})\label{eq:max_wage_multiple-objective}\\
        &where&~w(x,y)~\mbox{is defined as in (\ref{eq:wage})}\nonumber\\
        &s.t.&~\mbox{constraints}~(\ref{eq:min_delay-c1-c2})~\mbox{and}~ (\ref{eq:min_delay-c3-c4}).\nonumber
\end{eqnarray}
Because
\begin{eqnarray}
    &&\sum_{i=1}^{m}b_i - \sum_{i=1}^{m}\sum_{j=1}^{n_i}w(b_{i,j},d_{i,j}) \nonumber\\
    &=&\sum_{i=1}^{m}\sum_{j=1}^{n_i} [b_{i,j} - w(b_{i,j},d_{i,j})] \nonumber\\
    &=&\sum_{i=1}^{m}\sum_{j=1}^{n_i} \frac{b_{i,j}+d_{i,j}-\sqrt{(b_{i,j}+d_{i,j})^2-4c_{i,j}\hat{c}_{i,j}}}{2} \nonumber\\
    &=&\sum_{i=1}^{m}\sum_{j=1}^{n_i} f(b_{i,j}+d_{i,j},i,j), \nonumber
\end{eqnarray}
the objective function of the above optimization problem, i.e., (\ref{eq:max_wage_multiple-objective}),
is equivalent to
$\min~\sum_{i=1}^{m}\sum_{j=1}^{n_i} f(b_{i,j}+d_{i,j},i,j)$.
Furthermore, let $s_{i,j}=b_{i,j}+d_{i,j}$, and then
the above optimization problem can be converted to: 
\begin{eqnarray}\label{eq:max_wage_multiple-revised}
        &\min&\sum_{i=1}^{m}\sum_{j=1}^{n_i} f(s_{i,j},i,j)\label{eq:max_wage_multiple-objective-revised}\\
        &where&~f(x,i,j)~\mbox{is defined as in (\ref{eq:fx})},\nonumber\\
        &s.t.& \nonumber\\
        &&s_{i,j}\geq c_{i,j}+\frac{c_{i,j}\hat{c}_{i,j}}{c_{i,j}+\hat{c}_{i,j}}+\hat{c}_{i,j},
            \label{eq:max_wage_multiple-revised-c1}\\
        &&\sum_{j=1}^{n_i}s_{i,j}\geq b_i +\sum_{j=1}^{n_i}\hat{c}_{i,j},
            \label{eq:max_wage_multiple-revised-c2}\\
        &&\sum_{i=1}^{m}\sum_{j=1}^{n_i}s_{i,j}=\sum_{i=1}^{m}b_i+d.
            \label{eq:max_wage_multiple-revised-c3}
\end{eqnarray}
Here, constraints
\eqref{eq:max_wage_multiple-revised-c1},
\eqref{eq:max_wage_multiple-revised-c2}
and \eqref{eq:max_wage_multiple-revised-c3}
are derived from constraints
\eqref{eq:min_delay-c1-c2} and
\eqref{eq:min_delay-c3-c4}.
%
This optimization problem can be solved in three phases,
as formally presented in Algorithm~\ref{alg:SplitDeposit} 
and intuitively explained as follows:

    {\em Phase I: Initial allocation of clients' budgets and the server's deposits}. 
    The same as in Phase I of Algorithm~\ref{alg:SingleClient},
    each task $t_{i,j}$ is initially allocated with 
    a deposit of $\hat{c}_{i,j}$ and 
    a budget of $c_{i,j}+\frac{c_{i,j}\hat{c}_{i,j}}{c_{i,j}+\hat{c}_{i,j}}$,
    in order to satisfy constraints 
    \eqref{eq:min_delay-c3-c4}.
    
    {\em Phase II: Greedy allocation of each client's remaining budget}.
    After Phase I, each client $C_i$ may have remaining budget, denoted as
    $b'_i = b_i -\sum_{j=1}^{n_i}(c_{i,j}+\frac{c_{i,j}\hat{c}_{i,j}}{c_{i,j}+\hat{c}_{i,j}})$. 
    This remaining budget should be allocated step-by-step to 
    the tasks of this client greedily. That is, 
    similar to Phase II of Algorithm~\ref{alg:SingleClient},
    $b'_i$ is divided into small units. With each step,
    the client whose utility function (a monotonically-decreasing function) 
    has the minimal first-derivative value
    (i.e., the utility will drop the fastest), 
    is allocated with one more unit from $b'_i$.
    This procedure continues until all units are allocated.
    
    {\em Phase III: Greedy allocation of the server's remaining deposits}.
    In this phase, the remaining deposits at the server 
    is allocated to all the clients' tasks step-by-step 
    in the greedy manner, similar to Phase II. 
    
Note that, 
Algorithm~\ref{alg:SplitDeposit} also includes Phase IV,
which computes each $d_i$, 
the total deposit allocated to the tasks of $C_i$.
Each $d_i$ should be shared with $C_i$ 
for the client to find out the optimal allocation of its budget
in the next step. 

\begin{algorithm}[htb]
\caption{Optimal Splitting of Deposit (Server $S$ v.s. Client $C_i$, $i=1,\cdots,m$, with Static Task Set)}\label{alg:SplitDeposit}

{\bf Input:}
\begin{itemize}
\item
$b_i$: total budget of each client $C_i$;
\item
$d$: total deposit of server $S$;
\item
$n_i$: total number of tasks from each $C_i$;
\item
task set $\{t_{1,1},\cdots,t_{1,n_i},\cdots,t_{m,1},\cdots,t_{m,n_m}\}$ and
associated costs
$\{c_{1,1},\cdots,c_{1,n_1},\cdots,c_{m,1},\cdots,c_{m,n_m}\}$ and
$\{\hat{c}_{1,1},\cdots,\hat{c}_{1,n_1},\cdots,\hat{c}_{m,1},\cdots,\hat{c}_{m,n_m}\}$.
\end{itemize}

{\bf Output:} deposit $d_i$ allocated to each client $C_i$. \\

{\bf Phase I:} Initialization.

\begin{algorithmic}[1]
\For{$i\in\{1,\cdots,m\}$}
    \For{$j\in\{1,\cdots,n_i\}$}
        \State
        $s_{i,j}\leftarrow (\hat{c}_{i,j}+c_{i,j}+\frac{c_{i,j}\hat{c}_{i,j}}{c_{i,j}+\hat{c}_{i,j}})$
    \EndFor
\EndFor
\end{algorithmic}

~~\\
{\bf Phase II:} Greedy Allocation of Clients' Remaining Budgets.
\begin{algorithmic}[1]
\For{$i\in\{1,\cdots,m\}$}
    \State $b'_i\leftarrow b_i-\sum_{j=1}^{n_i}(c_{i,j}+\frac{c_{i,j}\hat{c}_{i,j}}{c_{i,j}+\hat{c}_{i,j}})$
    \While{$b'_i\geq\delta$} 
    \State $j*=\operatorname*{arg\,min}_{j\in\{1,\cdots,n_i\}} f'(s_{i,j},i,j)$.
    \State $s_{i,j*}\leftarrow (s_{i,j*}+\delta)$;~$b'_i\leftarrow (b'_i-\delta)$
    \EndWhile
\EndFor
\end{algorithmic}

~~\\
{\bf Phase III:} Greedy Allocation of Remaining Deposit to tasks.
\begin{algorithmic}[1]
\State $d'\leftarrow d-\sum_{i=1}^{m}\sum_{j=1}^{n_i}d_{i,j}$
\State $TS=\{(1,1),\cdots,(1,n_1),\cdots,(m,1),\cdots,(m,n_m)\}$
\While{$d'\geq\delta$}
    \State $(i*,j*)=\operatorname*{arg\,min}_{(i,j)\in TS} f'(s_{i,j},i,j)$
    \State $s_{i*,j*}\leftarrow (s_{i*,j*}+\delta)$;~$d'\leftarrow (d'-\delta)$
\EndWhile
\end{algorithmic}

~~~\\
{\bf Phase IV:} Preparing the Output.
\begin{algorithmic}[1]
\For{$i\in\{1,\cdots,m\}$}
    \State $d_i\leftarrow \sum_{j=1}^{n_i}s_{i*,j*} - b_i$ 
\EndFor
\end{algorithmic}
\end{algorithm}

In the second step,
it is already known the server's deposits allocated to the clients.
Because the budget of each client is also known,
each server-client pair can run Algorithm~\ref{alg:SingleClient}, 
presented in the previous section,
to find out the optimal allocation of budget/deposit to the client's tasks
to minimize the client's utility.

\begin{algorithm}[htb]
\caption{Optimal Resource Allocation (Server $S$ v.s. Client $C_i$, $i=1,\cdots,m$, with Static Task Set)}\label{alg:MultiClientStatic}

{\bf Input:}
\begin{itemize}
\item
$b_i$: total budget of each client $C_i$;
\item
$d$: total deposit of server $S$;
\item
$\{n_i\}$: total number of tasks from each $C_i$ for $i=1,\cdots,m$;
\item
task set $\{t_{1,1},\cdots,t_{1,n_i},\cdots,t_{m,1},\cdots,t_{m,n_m}\}$ and
associated costs
$\{c_{1,1},\cdots,c_{1,n_1},\cdots,c_{m,1},\cdots,c_{m,n_m}\}$ and
$\{\hat{c}_{1,1},\cdots,\hat{c}_{1,n_1},\cdots,\hat{c}_{m,1},\cdots,\hat{c}_{m,n_m}\}$.
\end{itemize}

{\bf Output:}
$\{b_{1,1},\cdots,b_{1,n_1},\cdots,b_{m,1},\cdots,b_{m,n_m}\}$ and
$\{d_{1,1},\cdots,d_{1,n_1},\cdots,d_{m,1},\cdots,d_{m,n_m}\}$. \\

\begin{algorithmic}[1]
\State $\{d_1,\cdots,d_m\} \leftarrow$ Algorithm~\ref{alg:SplitDeposit}

\For{$i\in\{1,\cdots,m\}$} \Comment{for each client}
    \State $(\{b_{i,1},\cdots,b_{i,n_i}\},\{d_{i,1},\cdots,d_{i,n_i}\})\leftarrow$ Algorithm~\ref{alg:SingleClient}.
\EndFor
\end{algorithmic}

\end{algorithm}

\subsection{Analysis}

To analyze our proposed solution, 
we prove the following:
First, the optimization problem defined in (\ref{eq:max_wage_multiple-revised}) has only one unique solution.
Second, Algorithm~\ref{alg:SplitDeposit} solves the optimization problem defined in (\ref{eq:max_wage_multiple-revised}).
Third, the optimization problem defined in \eqref{eq:max_wage_multiple-revised} is equivalent to 
the one defined in \eqref{eq:max_wage_multiple}.
Finally, the budget and deposit allocation strategy produced by Algorithm~\ref{alg:MultiClientStatic} is a Nash equilibrium. 

\begin{lemma}
\label{lem:multi_clients-1}
The optimization problem defined in (\ref{eq:max_wage_multiple-revised}) has only one unique solution.
\end{lemma}
\begin{proof}
(sketch).
Similar to the proof of Lemma~\ref{lem:singleClient-1},
we can show that the objective function is strictly convex and every constraint is convex. 
\end{proof}

\begin{theorem}
\label{theo:multi_clients-1}
Phases I, II and III of Algorithm~\ref{alg:SplitDeposit} 
solves the optimization problem defined in (\ref{eq:max_wage_multiple-revised}).
\end{theorem}
\begin{proof}
As elaborated in Appendix~5,
we prove by induction on the amount of fund to be allocated in Phase III. 
\comment{
Let $Opt(t,b_1,\cdots,b_m)$ denote the optimization problem
defined in (\ref{eq:max_wage_multiple-revised}) where
$t$ is an integer, 
$d=\sum_{i=1}^{m}\sum_{j=1}^{n_i}\hat{c}_{i,j} + t\cdot\delta$ and 
$b_i\geq\sum_{j=1}^{n_i}(c_{i,j}+\frac{c_{i,j}\hat{c}_{i,j}}{c_{i,j}+\hat{c}_{i,j}})$
for each $i\in\{1,\cdots,m\}$.
Note that, $c_{i,j}$ and $\hat{c}_{i,j}$ are constants.

Let ${\cal P}(t,b_1,\cdots,b_m)$ denote the following predicate: 
Phases I, II and III of Algorithm~\ref{alg:SplitDeposit} 
solves $Opt(t,b_1,\cdots,b_m)$. 
If ${\cal P}(t,b_1,\cdots,b_m)\equiv TRUE$, 
let $f^{*}(t,b_1,\cdots,b_m)$ denote 
the found minimal value of the objective function for $Opt(t,b_1,\cdots,b_m)$,
and each $s_{i,j}(t,b_1,\cdots,b_m)$ denote the value assigned to $s_{i,j}$ in the solution.  

The theorem claims that ${\cal P}(t,b_1,\cdots,b_m)\equiv TRUE$ 
for every $t\geq 0$ and every $\{b_i|i=1,\cdots,m\}$ satisfying the above relevant constraints.
This can be proved by induction on $t$.

Base case: 
When $t=0$ (i.e., $d=\sum_{i=1}^{m}\sum_{j=1}^{n_i}\hat{c}_{i,j}$), 
${\cal P}(t,b_1,\cdots,b_m)\equiv TRUE$, i.e., 
Phases I-III of Algorithm~\ref{alg:SplitDeposit} 
solves $Opt(t,b_1,\cdots,b_m)$, because:
Phase I simply initializes every $s_{i,j}$ 
to satisfy constraint \eqref{eq:max_wage_multiple-revised-c1};
Phase II minimizes $\sum_{j=1}^{n_i}f(s_{i,j},i,j)$ 
with the remaining budget of each client $C_i$, 
independently, based on the arguments similar to 
the proof of Lemma~\ref{lem:singleClient-2};
Phase III does nothing as $d'=0$.

Induction step: 
Assuming ${\cal P}(t,b_1,\cdots,b_m)\equiv TRUE$ 
for every integer $0\leq t\leq t_0$ and 
every $\{b_i|i=1,\cdots,m\}$ satisfying the above relevant constraints,
next we prove 
${\cal P}(t_0+1,\hat{b}_1,\cdots,\hat{b}_m)\equiv TRUE$
for every $\{\hat{b}_i|i=1,\cdots,m\}$ satisfying the relevant constraints.

First, let us consider optimization problem 
    \[
    OPT(0,\hat{b}_1,\cdots,\hat{b}_m).
    \]
According to the base case, 
${\cal P}(0, \hat{b}_1,\cdots,\hat{b}_m)\equiv TRUE$, and
each $s_{i,j}(0,\hat{b}_1,\cdots,\hat{b}_m)$ denotes 
the assignment of $s_{i,j}$ in the optimal solution.

Second, let 
\[
I=\{(i*,j*)\} = 
\operatorname*{arg\,min}_{\forall i\in\{1,\cdots,m\}\forall j\in\{1,\cdots,n_i\}}f'(s_{i,j},i,j)).
\]
Then, there must exist at least one $(i*,j*)\in I$ such that,
in the optimal solution to $OPT(t_0+1,\hat{b}_1,\cdots,\hat{b}_m)$, 
the assignment of $s_{i*,j*}$ is greater than 
the assignment of $s_{i*,j*}$ in the optimal solution to $OPT(0,\hat{b}_1,\cdots,\hat{b}_m)$; i.e.,
$s_{i*,j*}(t_0+1,\hat{b}_1,\cdots,\hat{b}_m)>s_{i*,j*}(0,\hat{b}_1,\cdots,\hat{b}_m)$. 
This can be proved by contradiction. 
If this is not the case, as $t_0+1\geq 1$, there should be at least $(i',j')\not\in I$ such that
$s_{i',j'}(t_0+1,\hat{b}_1,\cdots,\hat{b}_m)>s_{i',j'}(0,\hat{b}_1,\cdots,\hat{b}_m)$;
then, if one unit assigned to $s_{i',j'}$ instead is assigned to $s_{i*,j*}$, 
the server can earn higher wage. 

Third, let us consider optimization problem 
    \[
    OPT(t_0,\hat{b}'_1,\cdots,\hat{b}'_m),
    \]
where $\hat{b}'_{i*} = \hat{b}_{i*}+1$ while $\hat{b}'_i=\hat{b}_i$ for every $i\not= i*$.
We can prove the following:
\begin{itemize}
    \item 
    The optimal solution to $OPT(t_0,\hat{b}'_1,\cdots,\hat{b}'_m)$ is a feasible solution 
    to $OPT(t_0+1,\hat{b}_1,\cdots,\hat{b}_m)$.
    This is because: 
    $\hat{b}'_i\geq\hat{b}_i$ for every $i$, thus 
    satisfying constraint \eqref{eq:max_wage_multiple-revised-c2} in $OPT(t_0,\hat{b}'_1,\cdots,\hat{b}'_m)$
    implies satisfying constraint \eqref{eq:max_wage_multiple-revised-c2} in $OPT(t_0+1,\hat{b}_1,\cdots,\hat{b}_m)$; 
    $\sum_{i=1}^{m}\hat{b}_i+\delta = \sum_{i=1}^{m}\hat{b}'_i$ and 
    the value of $d$ in $OPT(t_0+1,\hat{b}_1,\cdots,\hat{b}_m)$ is greater than the $d$
    in $OPT(t_0,\hat{b}'_1,\cdots,\hat{b}'_m)$ by $\delta$, thus
    satisfying \eqref{eq:max_wage_multiple-revised-c3} in $OPT(t_0,\hat{b}'_1,\cdots,\hat{b}'_m)$
    implies satisfying constraint \eqref{eq:max_wage_multiple-revised-c3} in $OPT(t_0+1,\hat{b}_1,\cdots,\hat{b}_m)$.
    \item 
    Similarly, the optimal solution to $OPT(t_0+1,\hat{b}_1,\cdots,\hat{b}_m)$ is a feasible solution
    to $OPT(t_0,\hat{b}'_1,\cdots,\hat{b}'_m)$. 
\end{itemize}
Further due to the uniqueness of optimal solution to these optimization problems 
(based on Lemma~\ref{lem:multi_clients-1}),
$OPT(t_0+1,\hat{b}_1,\cdots,\hat{b}_m)$ and $OPT(t_0,\hat{b}'_1,\cdots,\hat{b}'_m)$ 
share the same optimal solution. 

Finally, based on the induction assumption, 
${\cal P}(t_0,\hat{b}'_1,\cdots,\hat{b}'_m)\equiv TRUE$.
That is, Algorithm~\ref{alg:SplitDeposit} solves $OPT(t_0,\hat{b}'_1,\cdots,\hat{b}'_m)$;
thus, it also solves $OPT(t_0+1,\hat{b}_1,\cdots,\hat{b}_m)$,
except that $\{\hat{b}'_i\}$ and $d=\sum_{i=1}^{m}\sum_{j=1}^{n_i}\hat{c}_{i,j}+t\cdot\delta$ 
(not $\{\hat{b}_i\}$ and $d=\sum_{i=1}^{m}\sum_{j=1}^{n_i}\hat{c}_{i,j}+(t_0+1)\cdot\delta$) 
are the inputs to the algorithm.
Hence, in the execution of the algorithm,
let us move the last assignment to $s_{i*,j*}$ in Phase II 
to be the first step in Phase III; this way,  
the algorithm works exactly as it takes $\{\hat{b}_i\}$ and $d=\sum_{i=1}^{m}\sum_{j=1}^{n_i}\hat{c}_{i,j}+(t_0+1)\cdot\delta$ 
as inputs to solve $OPT(t_0+1,\hat{b}_1,\cdots,\hat{b}_m)$.
That is, ${\cal P}(t_0+1,\hat{b}_1,\cdots,\hat{b}_m)\equiv TRUE$.
}
\end{proof}

\begin{lemma}
Optimization problem defined in \eqref{eq:max_wage_multiple} is equivalent to that defined in \eqref{eq:max_wage_multiple-revised}.
\end{lemma}

\begin{proof}
(sketch)
On one hand, 
as the constraints of the optimization problem defined in \eqref{eq:max_wage_multiple-revised} are derived from 
(and necessary conditions) of those of \eqref{eq:max_wage_multiple}, 
each optimal solution to \eqref{eq:max_wage_multiple} 
is a feasible solution to \eqref{eq:max_wage_multiple-revised}.
On the other hand, 
each optimal solution to \eqref{eq:max_wage_multiple-revised} 
can be converted to 
a feasible solution to \eqref{eq:max_wage_multiple}, 
with Part IV of Algorithm~\ref{alg:SplitDeposit} and
steps 2-3 in Algorithm~\ref{alg:MultiClientStatic}.
\end{proof}

\begin{theorem}
Algorithm~\ref{alg:MultiClientStatic} finds a Nash Equilibrium for the game between server $S$ and $m$ clients $C_1,\cdots,C_m$.
\end{theorem}
\begin{proof}
(sketch).
For server $S$, as the algorithm produces a unique optimal solution that maximizes its wage,
there is no incentive to deviate from the solution. 
For each client, as proved in Theorem~\ref{theo:singleClientNE}, 
there is no incentive, either, to deviate from the solution. 
\end{proof}

\comment{
\begin{lemma}
\label{lem:multi_clients}
Algorithm~\ref{alg:SplitDeposit} solves the optimization problem defined in (\ref{eq:max_wage_multiple-revised}).
\end{lemma}
\begin{proof}
Let $f(x)=\sum_{i=1}^{m}\sum_{j=1}^{n_i} f(x_{i,j},i,j)$, where $x_{i,j}=b_{i,j}+d_{i,j}, i \in \{1,2,\dots,m\}, j \in \{1,2,\dots,n_i\}$ for each $i$, and $f(x,i,j)$ defined as in \eqref{eq:fx}. We define the following optimization problem.
\begin{align}
\label{eq: multi_clients_f(x)_formal}
        & \min f(x) \nonumber\\
        s.t. \ & g_{i,j}(x) \equiv -x_{i,j} + c_{i,j}+\frac{c_{i,j}\hat{c}_{i,j}}{c_{i,j}+\hat{c}_{i,j}} \leq 0, \nonumber\\
        & h(x)\equiv\sum_{i=1}^m \sum_{j=1}^{n_i} x_{i,j} - \sum_{i=1}^m b_i - d = 0, \nonumber\\
        & i \in \{1,\ldots,m\}, j \in \{1,2,\dots,n_i\} ~\mbox{for each} ~i.
\end{align}
It is easy to see that \eqref{eq: multi_clients_f(x)_formal} is in the same form of \eqref{eq:f(x)formal} in Appendix \ref{appendix:single_client_unique_point}. Therefore, \eqref{eq: multi_clients_f(x)_formal} also has a unique optimal point, which we denote as $s^* = \{s^*_{i,j}\}$. A valid allocation of \eqref{eq:max_wage_multiple-revised} should satisfies constraints \eqref{eq:min_delay-c1} - \eqref{eq:min_delay-c4} as well as $\eqref{eq:s=b+d}$ for all $i \in \{1,2,\dots,m\}$.

Phase I of Algorithm \eqref{alg:SplitDeposit} guarantees that constraints \eqref{eq:min_delay-c3} and \eqref{eq:min_delay-c4} are satisfied. Phase II ensures that \eqref{eq:min_delay-c2} is fulfilled while Phase III satisfies \eqref{eq:min_delay-c1}. Since in Algorithm \ref{alg:SplitDeposit}, $s_{i*,j*}$ is indeed the summation of $b_{i,j}$ and $d_{i,j}$, $\eqref{eq:s=b+d}$ holds naturally. A similar analysis to the proof of Lemma \ref{lem:alg1-proof} will show that $\{s_{i*,j*}\}$ generated by the algorithm is indeed the optimal solution of \eqref{eq: multi_clients_f(x)_formal}. Therefore, algorithm~\ref{alg:SplitDeposit} solves the optimization problem defined in (\ref{eq:max_wage_multiple-revised}).
\end{proof}
}

\section{Multiple Clients with Dynamic Tasks}

Tasks can be submitted to the cloud server dynamically.
We assume that, each task is submitted with a budget promised by a client,
and the server can start executing the task immediately
due to the typically rich resource available at the cloud server. 

Though the server can also 
immediately commit a deposit for a task at its arrival,
this is not desired:
First, the arrival of tasks is not predictable, 
making it difficult to optimize the distribution of the fund for deposit. 
Second, there could be a large number of tasks executed in parallel, 
which could divide the fund into very small pieces;
the smaller are deposits, 
as we learn from the previous sections,
the more frequently TTPs would be hired, 
which would decrease the amount of wage for the server 
and increase the latency for computation verification.

To address this issue,  
the server should distribute its fund to only a small number of tasks at a time.
Also, the deposit made to tasks should be reclaimed as soon as possible
and thus can be reused quickly; 
hence, it is more beneficial to make deposit to a task 
when it is completed and ready to be released
than when it is just submitted.
Moreover, as there could be a large number of tasks completing 
during a short period of time, the server should 
control the pace at which 
the completed tasks are released,
and thus it can make deposits to only 
a selected subset of completed tasks at a time. 

Based on the above ideas, we propose two algorithms.
We first propose a baseline scheme named {\em sequential releasing},
with which the server releases only one completed task at a time
and uses all of its fund as deposit for the task; this way, it can earn the most from 
each individual task, however, at the expense of slowed pace to release completed tasks.
To address the limitation, we further propose a more generic scheme, 
{\em parallel releasing}, with which the server releases a subset of completed tasks
at a time, based on certain criteria adjustable with some system parameters,
to balance the trade off between the wage earned from each task
and the pace of releasing completed tasks. 

\subsection{Sequential Releasing}

As formally presented in Algorithm~\ref{alg:sequential},
the sequential releasing algorithm works as follows. 
%
    When a task arrives,
    the server immediately start executing it.
    When the computation of a task finishes,
    it is released if currently the server has full fund available;
    otherwise, it is put into the queue waiting for its turn to be released.
    When a released task has been finally accepted by the client who submitted it, 
    the server reclaims the deposit assigned to the task.
    Then, the waiting queue is checked;
    if there is one or more tasks there,
    one of them is picked to be released with the server's fund as deposit. 

\begin{algorithm}[htb!]
\caption{Sequential Releasing}
\label{alg:sequential}

{\bf Variables:}
\begin{itemize}
\item
$d$: server's total fund for deposit;
\item 
$d_{avail}$: available fund for deposit, initialized to $d$;
\item
$\{t_{i,j}\}$: a dynamic sequence of tasks;
\item 
${\cal Q}$: queue buffering completed but unreleased tasks.
\end{itemize}

{\bf Upon receiving a new task:}
start executing the task. 

{\bf Upon completing execution of task $t_{i,j}$:}
\begin{algorithmic}[1]
\If{$d_{avail}>0$}
\State release $t_{i,j}$ with deposit: $d_{i*,j*}=d_{avail}$
\State $d_{avail}\leftarrow 0$
\Else
\State ${\cal Q}.append(t_{i,j})$
\EndIf
\end{algorithmic}

{\bf Upon task $t_{i,j}$ being accepted by its owner (a client):}
\begin{algorithmic}[1]
\State $d_{avail}\leftarrow d_{i,j}$
\If{${\cal Q}\neq\emptyset$}
\State $t_{i*,j*}\leftarrow{\cal Q}.dequeue$
\State release the result of $t_{i*,j*}$ with deposit: $d_{i*,j*}=d_{avail}$
\State $d_{avail}\leftarrow 0$
\EndIf
\end{algorithmic}
\end{algorithm}

\subsection{Parallel Releasing}

As formally presented in Algorithm~\ref{alg:parallel},
at the core of the parallel releasing algorithm is a function named {\em AssignDeposit}.
Every time when the function is called,
it works in the following three steps.

First, the server checks the following condition to determine if
it is time to find another set of tasks to release:
{\em The server's fund that has been currently locked due to being assigned to 
    the tasks that are released but yet finalized, should be less than
    percentage $\alpha$ of the server's whole fund for deposit}.

Second, if it is time to start a new round of task releasing, 
the server selects a number of tasks waiting in the front of 
the queue of completed yet unreleased tasks,
such that the sum of the required minimal deposits for these tasks
is no more than percentage $\beta$ of the currently available fund.
Note that, parameters $\alpha$ and $\beta$ are used to 
control the trade off between the pace of task releasing and
the amount of deposit that a task can be assigned:
the smaller are $\alpha$ and $\beta$,
the slower is the pace of task releasing and
the larger is the deposit that a task can be assigned 
(thus the larger is the wage that the server can earn from the task);
and vice versa. 

Third, once the set of tasks to be released has been selected,
Algorithm~\ref{alg:MultiClientStatic} is called
to find out the optimal strategy to deposit the currently available fund 
to the tasks, to attain the dual goals of maximizing the server's wage
and minimizing the delay for verification.

With the {\em AssignDeposit} function in place,
the parallel releasing algorithm 
runs as follows:
When a new task arrives,
the server immediately starts executing the task;
when a task is completed or finalized,
the server calls the {\em AssignDeposit} function 
to make deposit assignments as long as the above-discussed conditions are satisfied.

\begin{algorithm}[htb!]
\caption{Parallel Releasing}
\label{alg:parallel}

{\bf Variables:}
\begin{itemize}
\item
$d$: server's total fund for deposit;
\item 
$d_{locked}$: fund for deposit that have been locked by clients, initialized to $0$;
\item
$\{t_{i,j}\}$: a dynamic sequence of tasks;
\item 
${\cal Q}$: queue buffering completed but unreleased tasks.
\end{itemize}

{\bf Upon receiving a new task:}
start executing the task. 

{\bf Upon completing execution of task $t_{i,j}$:}
\begin{algorithmic}[1]
\State ${\cal Q}.append(t_{i,j})$
\State call {\em AssignDeposit}
\end{algorithmic}

{\bf Upon task $t_{i,j}$ being accepted by its owner (a client):}
\begin{algorithmic}[1]
\State $d_{locked}\leftarrow d_{locked}-d_{i,j}$
\State call {\em AssignDeposit}
\end{algorithmic}

{\bf Function}~{\em AssignDeposit}:
\begin{algorithmic}[1]

\State $S_{task}\leftarrow\emptyset$ \Comment{temporary task set}
\State $\hat{c}_{temp}\leftarrow 0$

\If{$d_{locked}\leq\alpha\cdot d$}

    \While{${\cal Q}\neq\emptyset$}
        \State $t_{i,j}\leftarrow{\cal Q}.front$
        \State $\hat{c}_{temp}\leftarrow\hat{c}_{temp}+\hat{c}_{i,j}$
        \If{$\hat{c}_{temp}\geq\beta\cdot(d-d_{locked})$}
            \State break
        \Else
            \State $S_{task}\leftarrow S_{task}\bigcup\{t_{i,j}\}$
            \State ${\cal Q}.dequeue$
        \EndIf
    \EndWhile
    
    \State call Algorithm~\ref{alg:MultiClientStatic} to allocate remaining fund for deposit, i.e., $d-d_{locked}$, to tasks in $S_{task}$
    \State $d_{locked}\leftarrow d$

\EndIf

\end{algorithmic}

\end{algorithm}

\subsection{Simulations}

Due to dynamic and unpredictable nature of the incoming tasks,
it is hard to analyze our proposed algorithms in theory. 
Hence, we evaluate the algorithms through simulations. 

\comment{
We assume executor has enough resources for the computation. Then each task gets computed once received. Algorithm~\ref{alg:sequential} and ~\ref{alg:parallel} describe the workflow of sequentially release and parallel release. The life span of each task can be described as follows.
\begin{itemize}
    \item {$received$:} task received by the executor and starts computation
    \item {$computed$:} executor finishes computation and task waits to be allocated with deposit and budget. Assume the computation time take has the same value as cost.
    \item {$release$:} task is assigned with a deposit and budget and waiting for the action of client.
    \item {$rebuttal$:} client raises a trusted execution mechanism to verify the result based on the assigned deposit and budget
    \item {$confirm$:} waiting for the computation result to be confirmed
    \item {$finalized$:} redistribute deposit and reward
\end{itemize}
}


\subsubsection{Settings}

As inputs,
we simulate $50$ groups of task and each group has $100$ tasks.
Each task can be regular or heavy, different in the cost.
The arrivals of tasks follows 
a Poisson distribution with average interval $\tau$.
Specifically, the tasks are characterized by following parameters:
(i) $c$ is the average cost of each regular task,
quantified by the required execution time; 
we let $c$ range from $8\tau$ to $128\tau$, 
with $32\tau$ as the default value.
(ii) $T_{confirm}$ is the average time needed for the funds 
(including the client's budget and the server's deposit) to be 
distributed and confirmed on blockchain; 
we let $T_{confirm}$ range from $0.25\tau$ to $16\tau$, 
with $2\tau$ by default. 
(iii) $b$ is the client's average budget for a task;
we let $b\in\{2c,2.5c,3c,3.5c,4c\}$, 
with $2.5c$ by default. 
(iv) $p_{heavy}$ is the probability that a task is heavy (rather than regular),
with the purpose of studying the impact of task heterogeneity; 
we let $p_{heavy}$ 
take the default value of $0.1$, and
each heavy task has an average cost being $3$ times of a regular task.
(v) the default cost for hiring a TTP for verifying 
a task with cost $c$ is $\hat{c}=3c$.
(vi) $d$ is the server's total fund for deposit;
we let $d$ take the default value of $d_0=768\tau$ and
vary over the set of $\{d_0, 2d_0, 4d_0, 8d_0, 16d_0, 32d_0\}$.

We measure the following metrics: 
(i) $\frac{w}{b}$ is the average percentage 
of the client's budgets for tasks 
that are earned by the server as wages; 
(ii) $\frac{w}{\tau}$ is the average amount of wage that
the server can earn per time interval;
(iii) $\frac{w}{c}$ is the average amount of wage that
the server can earn for each unit of computational cost it pays;
(iv) average delay for a task is the time elapse (in the unit of $\tau$) from 
the task is submitted until it is finalized 
(i.e., it is completed and the funds associated with it have been distributed and confirmed on the blockchain). 
Specifically, the delay includes four parts: 
the computation delay $c$, 
the release delay that is the time elapse 
from the task being completed till its result being released,
the TTP delay which is time for a hired TTP to verify a task (if the client hires a TTP),
and the confirmation delay which is the time for the funds associated with the task 
to be distributed and confirmed on blockchain. 
For each task, the computation and confirmation delays are fixed based on the task and system property. 
The release and the TTP delays, however, 
are affected by the algorithm used and the parameters of the algorithm.  

In the simulation, 
we first evaluate how parameters $\alpha$ and $\beta$ affect 
the performance of the parallel releasing algorithm, 
which is followed by the comparison between 
the sequential releasing and the parallel releasing algorithms.

\subsubsection{Impacts of $\alpha$ and $\beta$ on Parallel Releasing}

In parallel releasing,
$\alpha$ controls when a new round of fund allocation starts;
the larger is $\alpha$, the shorter interval between the rounds.
$\beta$ controls how many tasks are released in parallel;
the larger is $\beta$, the more tasks released at the same time, and
the less deposit each of the tasks is assigned. 

Figure~\ref{fig:delay_beta} shows the impact of $\alpha$ and $\beta$ 
on the average delay of task.
When other parameters are the same,
the delay decreases as $\alpha$ increases; this is obvious
because increasing $\alpha$ shortens the interval 
between the rounds of fund allocation and task releasing. 
Given a fixed $\alpha$, 
it is interesting to observe that $\beta$ has different impacts on the delay
when the value of $\alpha$ is different. 
With a small $\alpha$ (e.g., $0$ and $0.3$),
the delay generally increases along with $\beta$.
This is because, as $\beta$ increases,
the deposit assigned to a task decreases while
the probability of hiring TTP increases for verifying the task;
when a task is checked with TTP, 
this not only increases the delay of the task,
but also slows down the returning of the deposit and 
thus postpones the starting of the next round of allocation.
The latter impact becomes even more significant with smaller $\alpha$,
where a new round of allocation can start only after
all or a large percentage of the fund has been returned. 
On the other hand, with a large $\alpha$ (e.g., $1$ and $0.6$),
a new round of allocation can start even when a large percentage of fund is locked;
thus, increasing $\beta$ generally allows faster releasing of tasks
and thus decreases the average delay of task.

\begin{figure}[htb!]
  \begin{subfigure}{0.48\columnwidth}
    \includegraphics[width=\columnwidth]{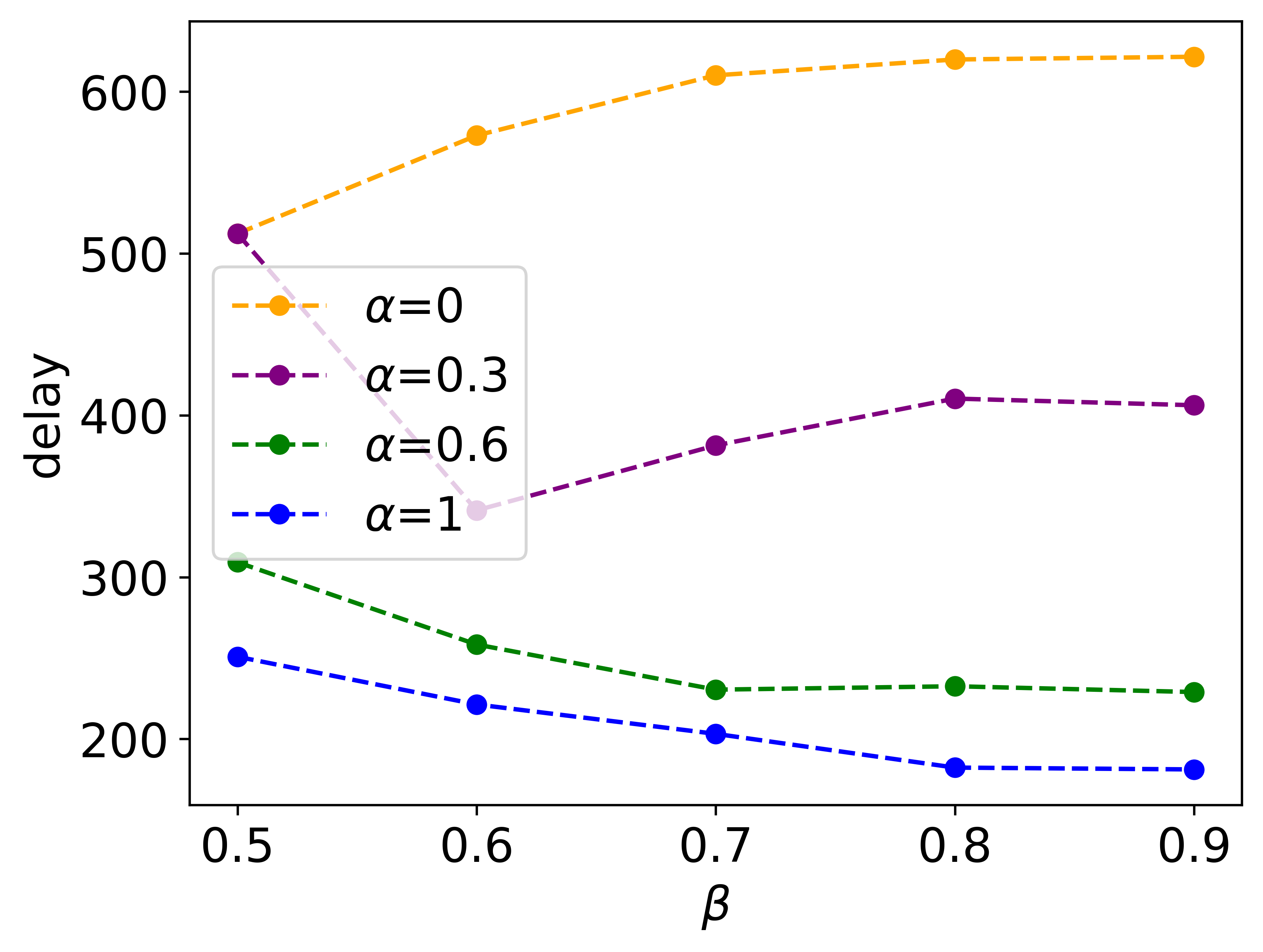}
    \caption{delay v.s. $\beta$} \label{fig:delay_beta}
  \end{subfigure}%
  \hspace*{\fill}   
  \begin{subfigure}{0.48\columnwidth}
    \includegraphics[width=\columnwidth]{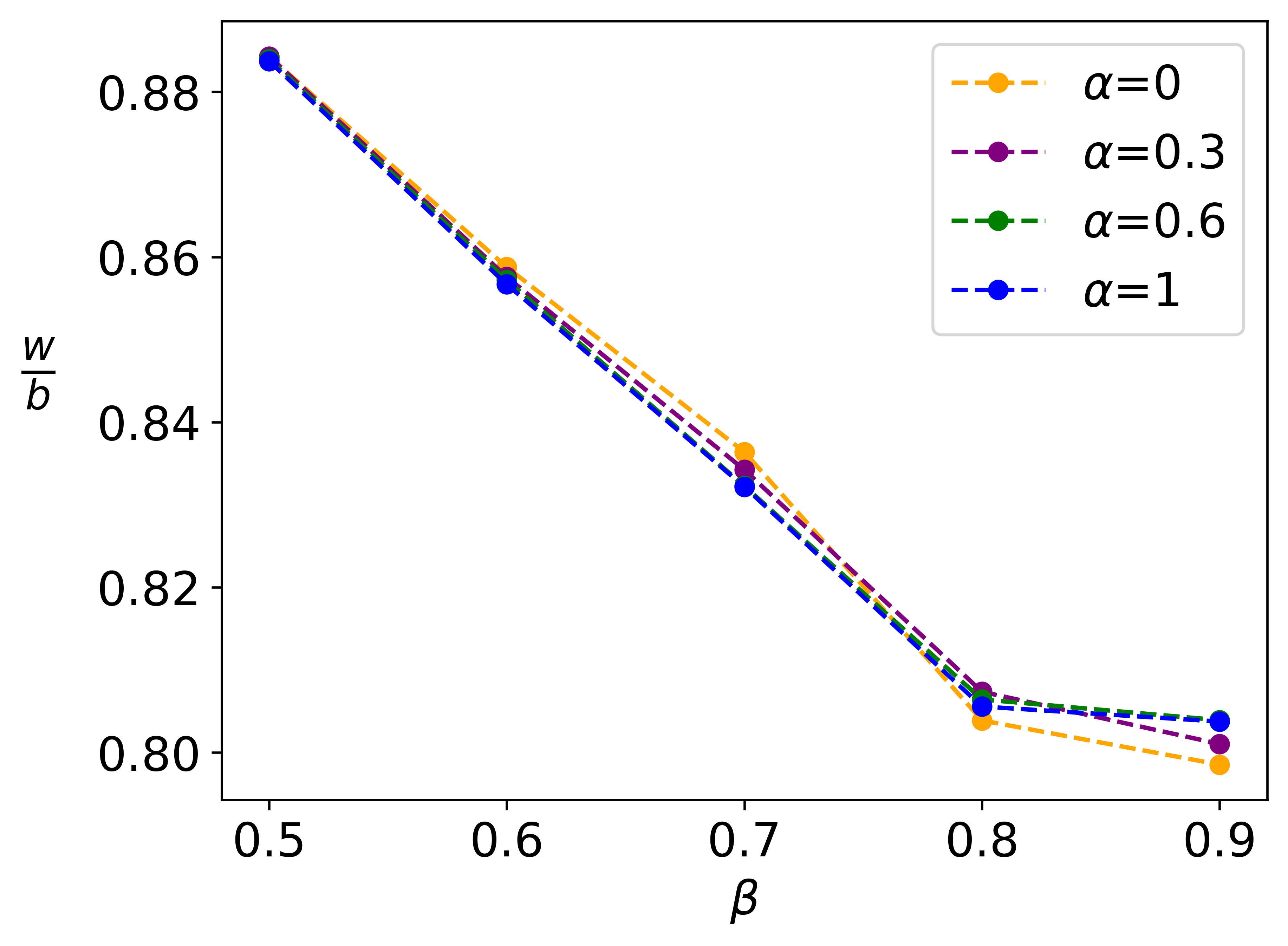}
    \caption{$\frac{w}{b}$ v.s. $\beta$} \label{fig:wage_budget_beta}
  \end{subfigure}%
  \hfill   
  \begin{subfigure}{0.48\columnwidth}
    \includegraphics[width=\columnwidth]{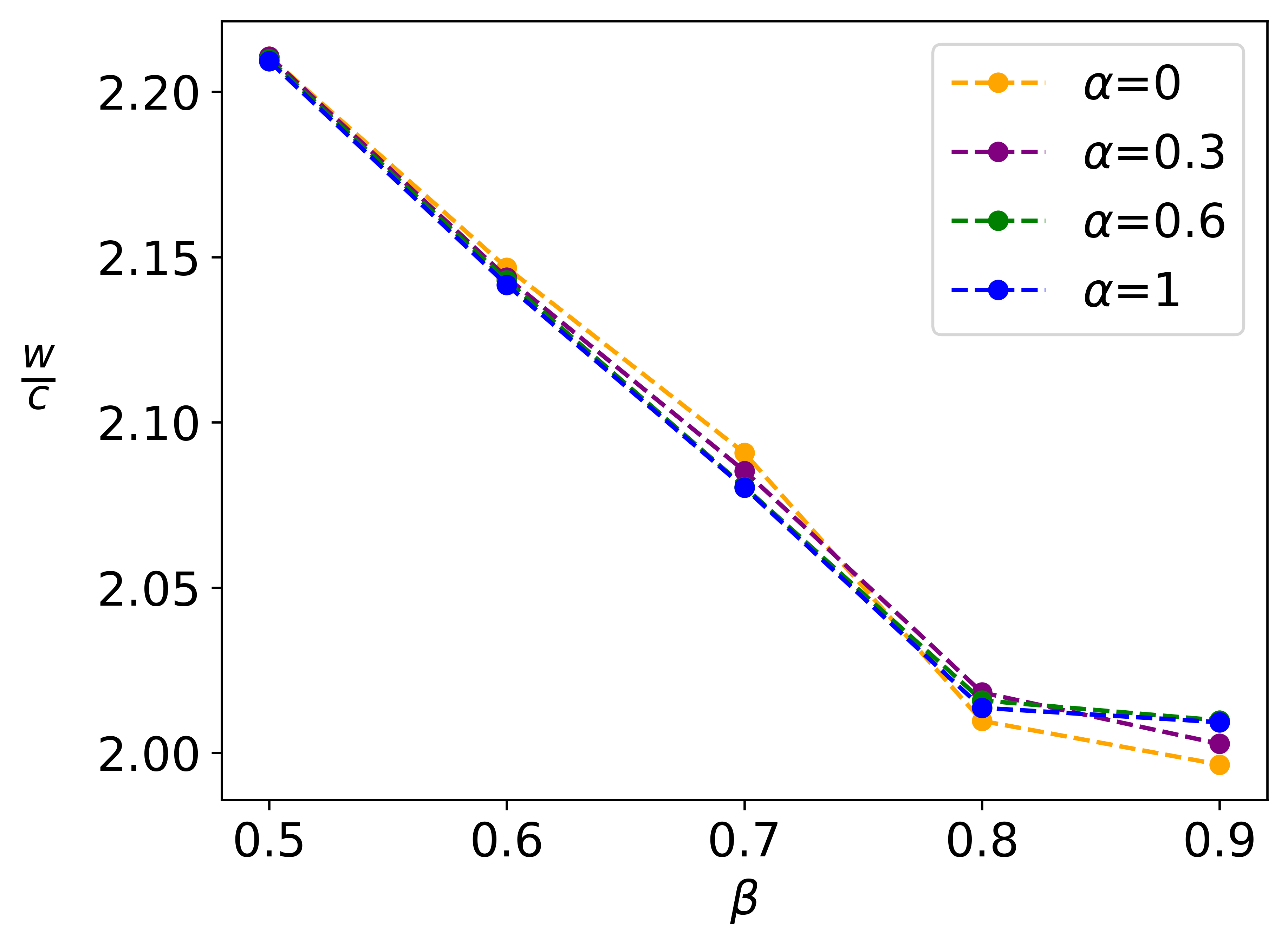}
    \caption{$\frac{w}{c}$ v.s. $\beta$} \label{fig:wage_cost_beta}
  \end{subfigure}%
  \hspace*{\fill}   
  \begin{subfigure}{0.48\columnwidth}
    \includegraphics[width=\columnwidth]{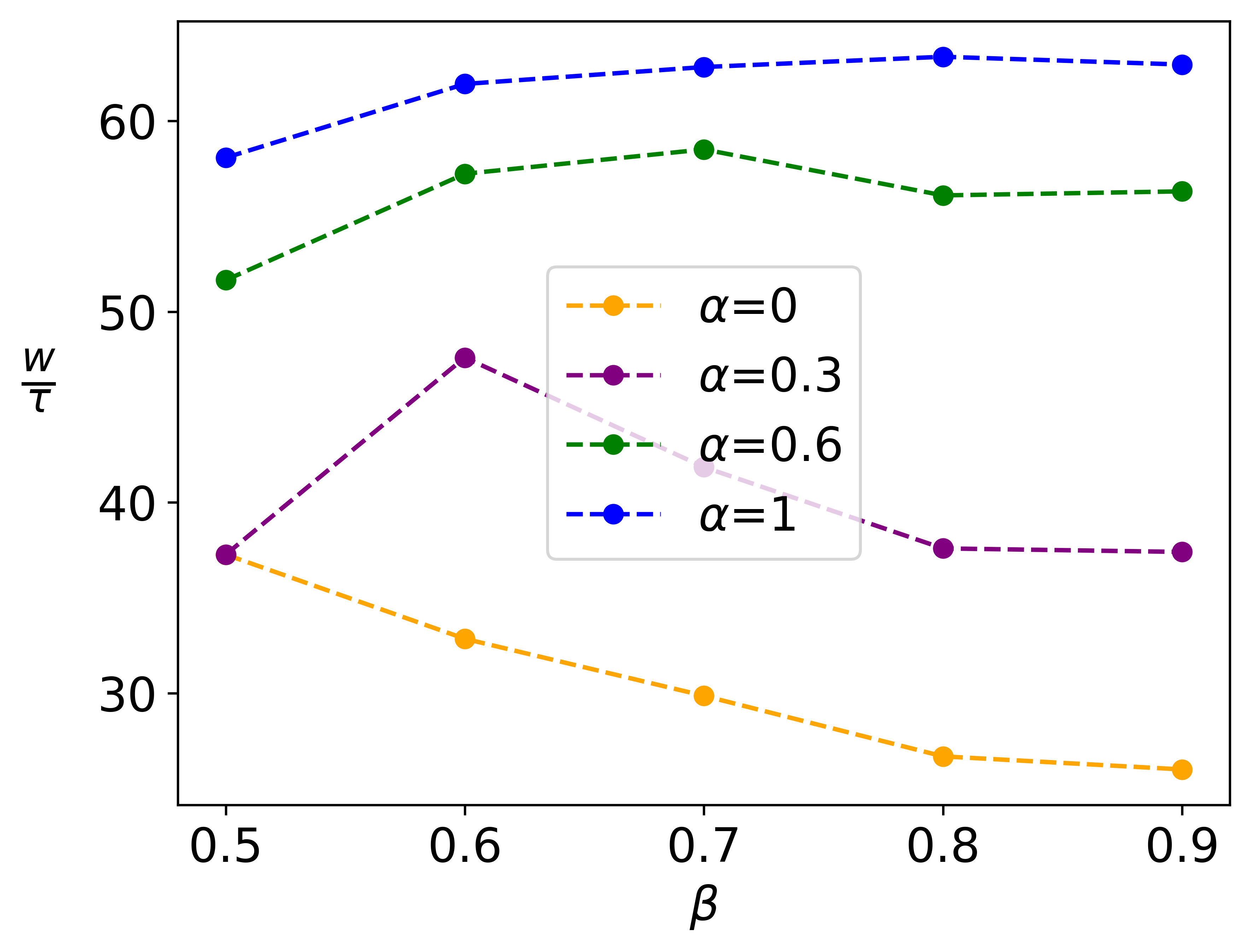}
    \caption{$\frac{w}{\tau}$ v.s. $\beta$} \label{fig:wage_time_beta}
  \end{subfigure}%
\caption{The performance under different $\alpha$ and $\beta$}
\label{fig:alpha_and_beta}
\end{figure}

From Figure~\ref{fig:wage_budget_beta} and Figure~\ref{fig:wage_cost_beta}, 
as expected, the greater is $\beta$, 
the smaller are $\frac{w}{b}$ and $\frac{w}{c}$;
and this is nearly independent of $\alpha$.
This is because, when $\beta$ increases, 
more tasks can be released at one round, 
which in turn leads to lower deposits allocated to each task.
As we learn from the above sections,
the server's wage is an increasing function of the deposit;
hence, lower deposits lead to lower wages earned by the server.

Fig~\ref{fig:wage_time_beta} shows the impacts of $\alpha$ and $\beta$ on 
metric $\frac{w}{\tau}$, which takes into account both 
the server's concern (i.e., high wage) and 
the client's concern (i.e., low delay);
intuitively, higher $\frac{w}{\tau}$ is desired.
As we can observe, 
a larger $\alpha$ generally leads to higher $\frac{w}{\tau}$ when other parameters are the same. 
When $\alpha$ is small (e.g., $0$ and $0.3$),
$\frac{w}{\tau}$ generally decreases as $\beta$ increases,
as it causes the wage to decrease and the delay to increase. 
When $\alpha$ is large (e.g., $1$ and $0.6$),
increasing $\beta$ causes both the wage and the delay to drop; as a result, 
$\frac{w}{\tau}$ increases with $\beta$ as long as $\beta$ is not too large 
(e.g., $\beta\leq 0.7$ for $\alpha=0.6$ and $\beta\geq 0.8$ for $\alpha=1$).


The above simulation results reveal that, 
parameter $\alpha$ only affects the average delay of task experienced by the client 
and $\frac{w}{\tau}$, the wage per time unit (concern of the server);
when other parameters are fixed, 
$\alpha=1$ leads the smallest delay and the largest $\frac{w}{\tau}$.
Hence, $\alpha=1$ is desired
and we will use this setting for the rest simulation. 


\subsubsection{Sequential v.s. Parallel Releasing}

Next, we compare the performance of 
sequential releasing and parallel releasing (with $\alpha=1$ and $\beta=0.6)$)
under various conditions. 

\begin{figure}[htb!]
  \begin{subfigure}{0.48\columnwidth}
    \includegraphics[width=\columnwidth]{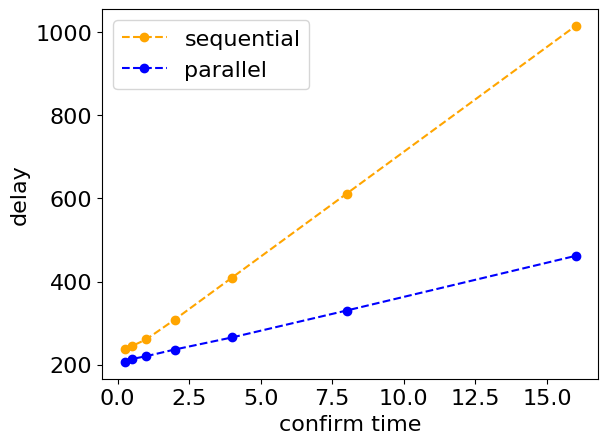}
    \caption{delay v.s. confirm time} \label{fig:delay_confirmtime}
  \end{subfigure}%
  \hspace*{\fill}   
  \begin{subfigure}{0.48\columnwidth}
    \includegraphics[width=\columnwidth]{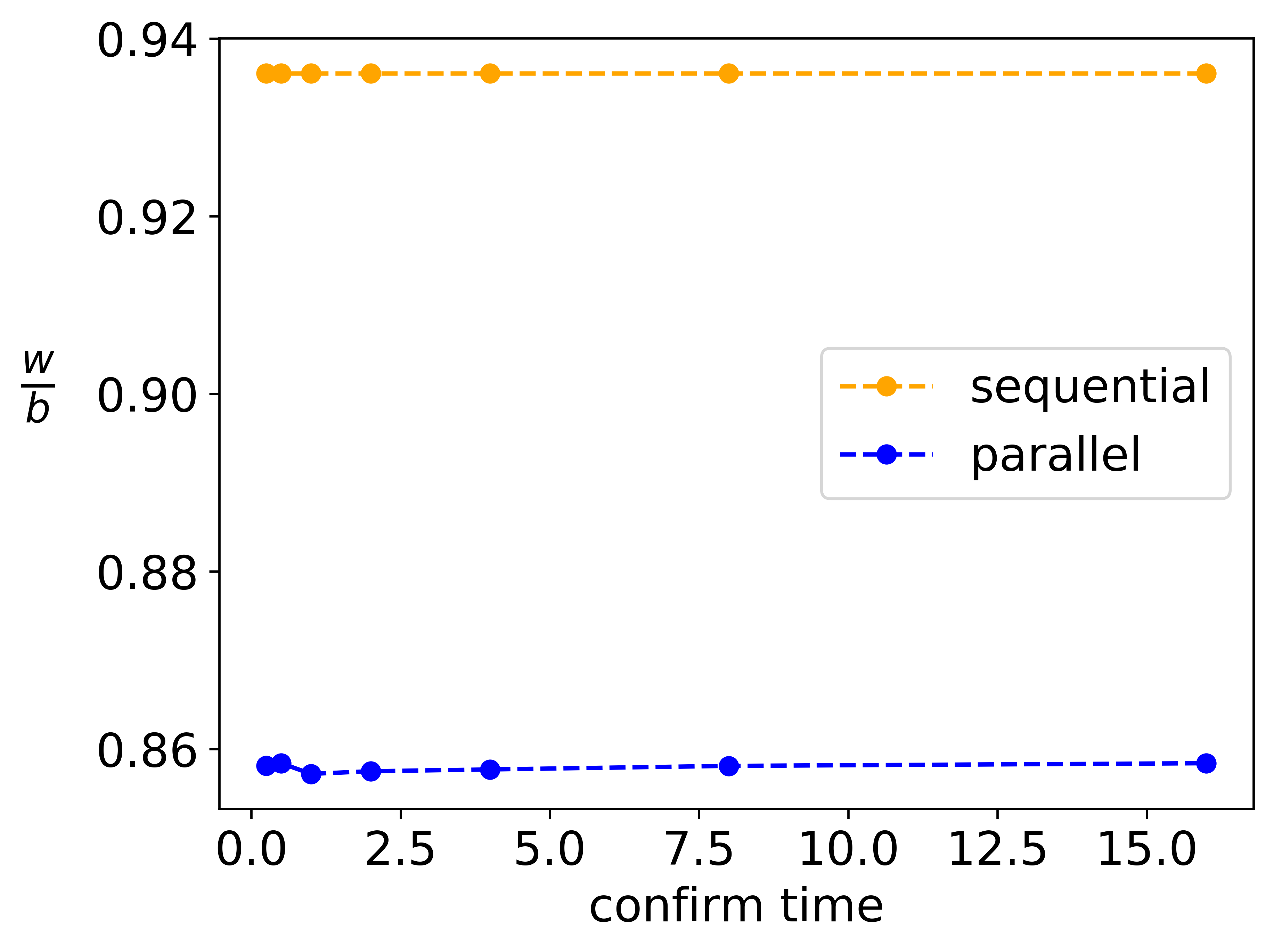}
    \caption{$\frac{w}{b}$ v.s. confirm time} \label{fig:wage_budget_confirmtime}
  \end{subfigure}%
  \hfill   
  \begin{subfigure}{0.48\columnwidth}
    \includegraphics[width=\columnwidth]{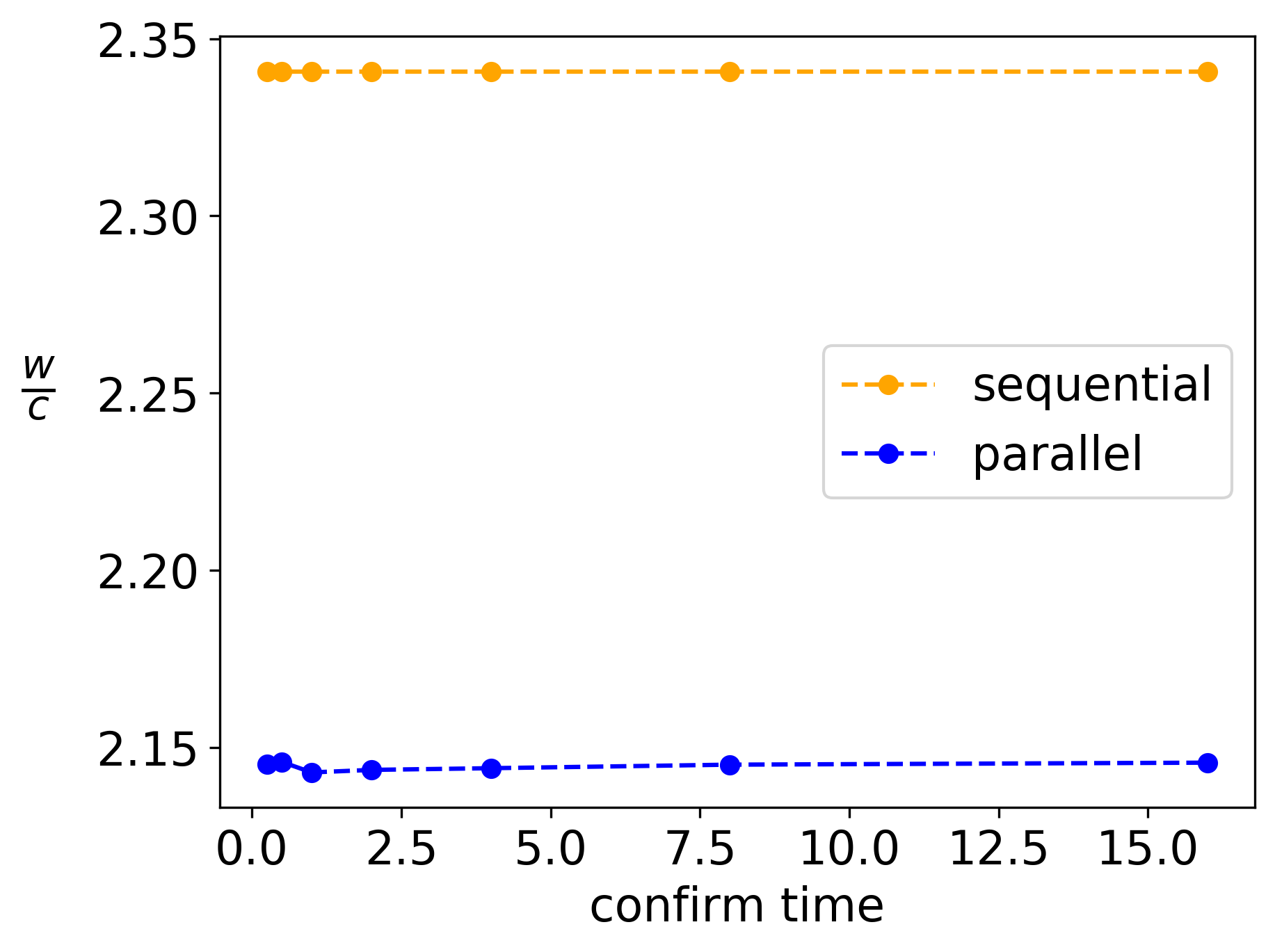}
    \caption{$\frac{w}{c}$ v.s. confirm time} \label{fig:wage_cost_confirmtime}
  \end{subfigure}%
  \hspace*{\fill}   
  \begin{subfigure}{0.48\columnwidth}
    \includegraphics[width=\columnwidth]{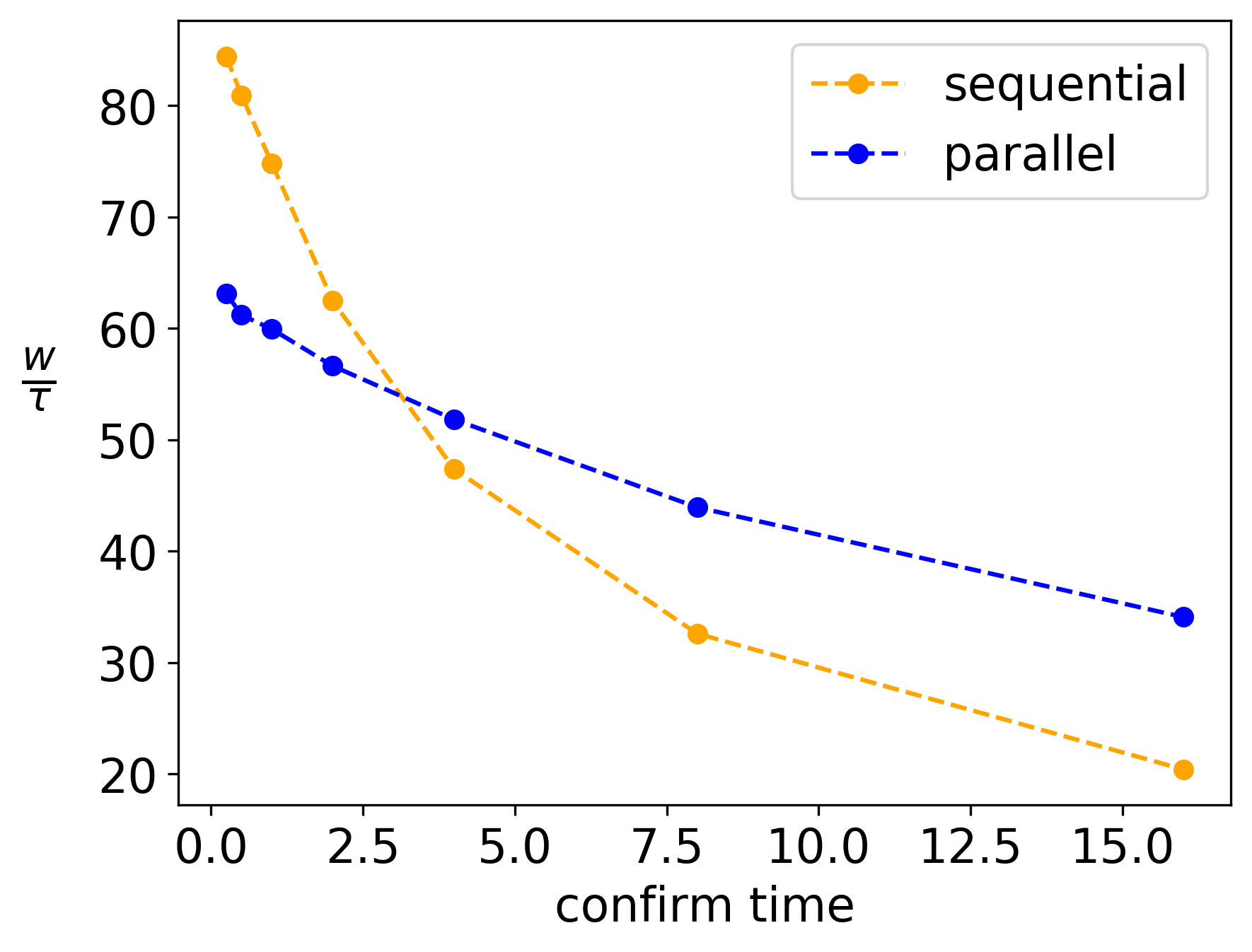}
    \caption{$\frac{w}{\tau}$ v.s. confirm time} \label{fig:wage_time_confirmtime}
  \end{subfigure}%
\caption{Impact of confirmation time}
\label{fig:confirmtime}
\end{figure}

Figure~\ref{fig:confirmtime} shows 
the comparisons as the confirmation time varies while other parameters take their default values.
As we can see, when the confirmation time is small (e.g., close to $0$), 
sequential releasing has only slightly higher delay than parallel releasing;
this is because, the fund for deposit can be returned and then reused quickly for both algorithms. 
However, as the confirmation time increases,
the returning of deposit funds gets slower for both, while
the parallel one can still release tasks when funds are only partly returned; 
hence, the gap between two algorithms gets wider and wider
though their delays both increase. 
%
%
%

Without surprise,
$\frac{w}{b}$ and $\frac{w}{c}$ 
have nearly no change when confirm time increases,
for both the sequential and the parallel releasing algorithms. 
But as the two algorithms have different speeds in increasing the delay for task,
they also decrease their $\frac{w}{\tau}$ at different rates.
$\frac{w}{\tau}$ of the sequential releasing algorithm
drops faster, and quickly becomes lower than that of its parallel counterpart.

\begin{figure}[htb!]
  \begin{subfigure}{0.48\columnwidth}
    \includegraphics[width=\columnwidth]{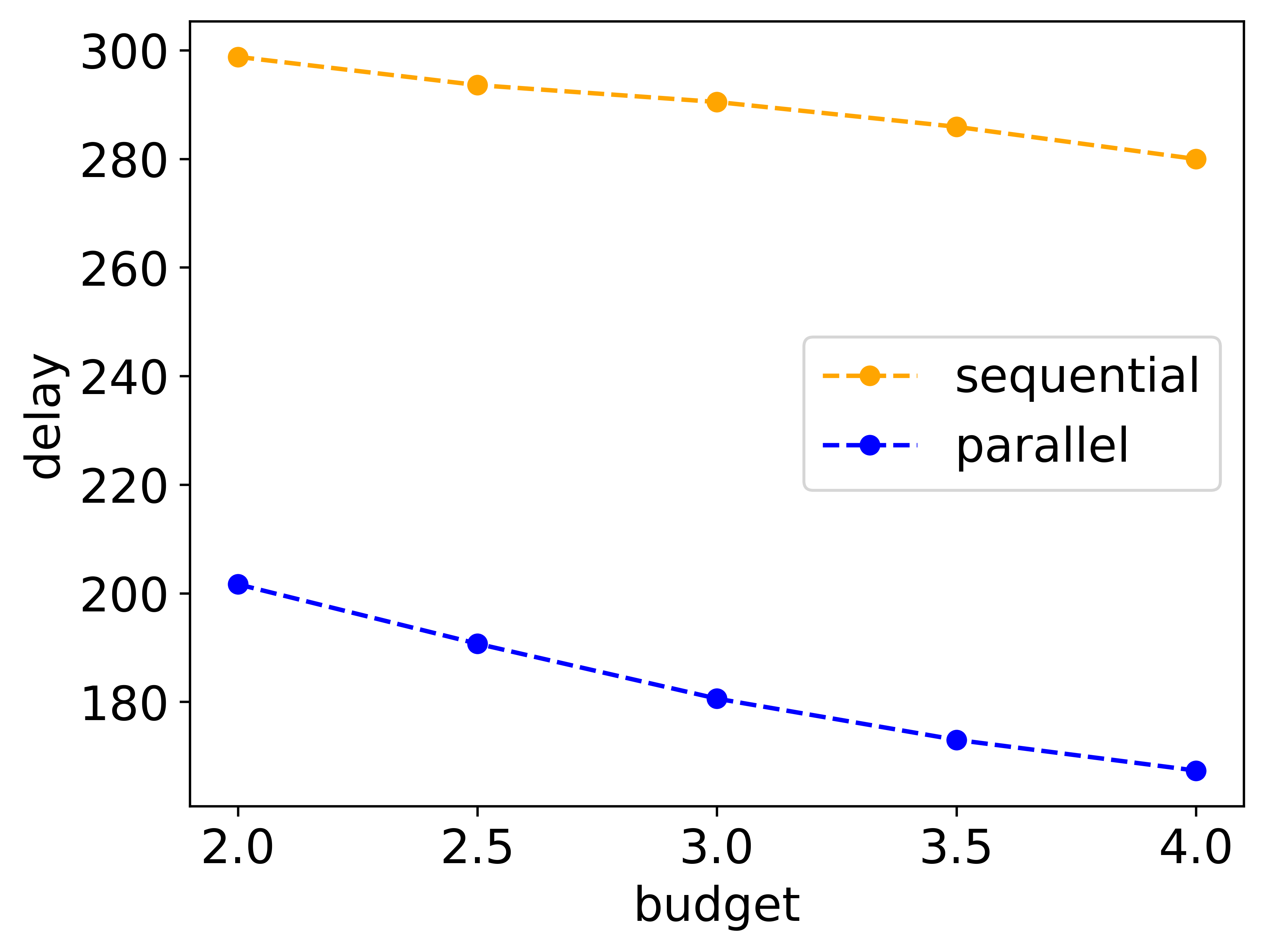}
    \caption{delay v.s. budget} \label{fig:delay_budget}
  \end{subfigure}%
  \hspace*{\fill}   
  \begin{subfigure}{0.48\columnwidth}
    \includegraphics[width=\columnwidth]{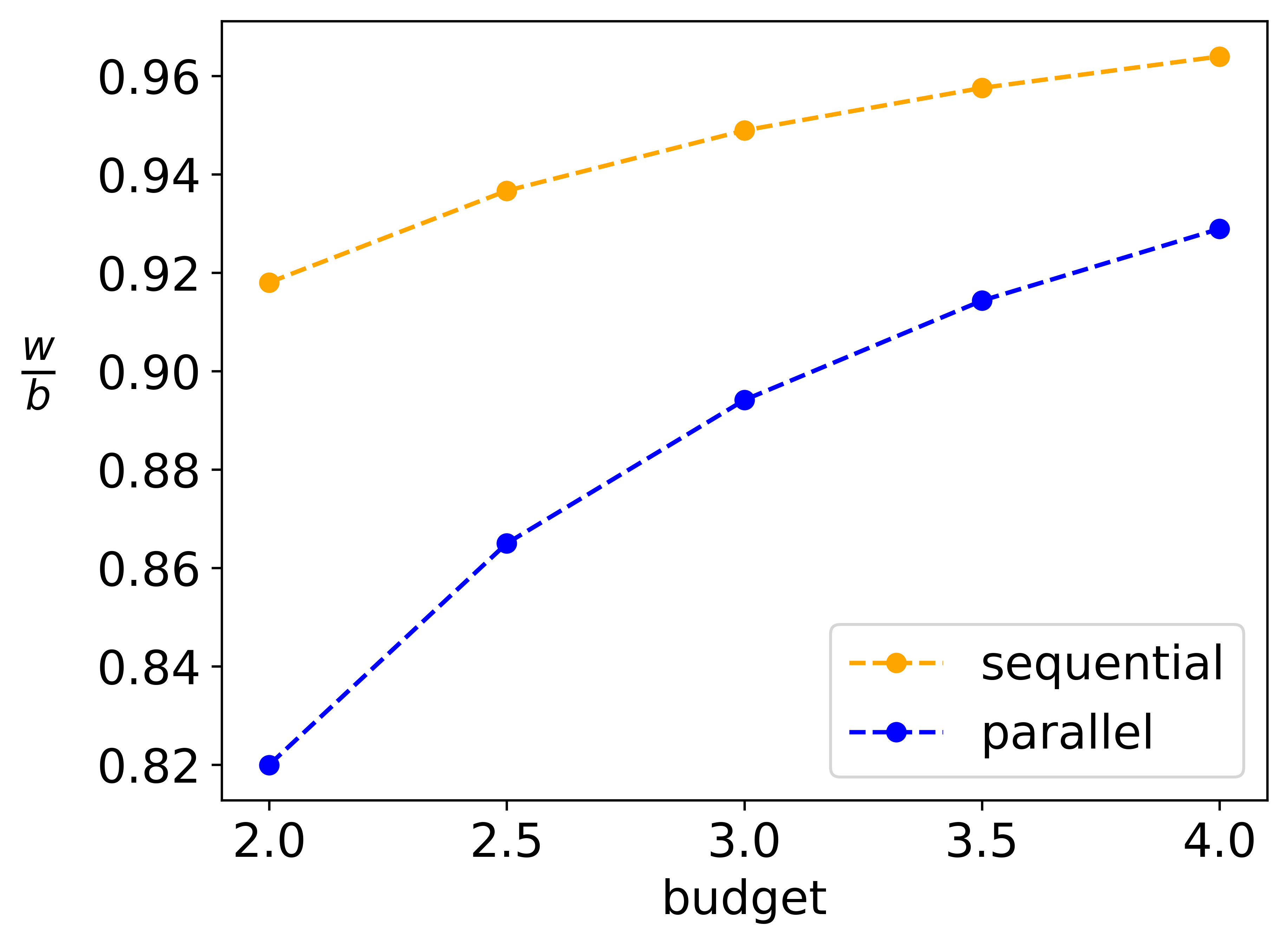}
    \caption{$\frac{w}{b}$ v.s. budget} \label{fig:wage_budget_budget}
  \end{subfigure}%
  \hfill   
  \begin{subfigure}{0.48\columnwidth}
    \includegraphics[width=\columnwidth]{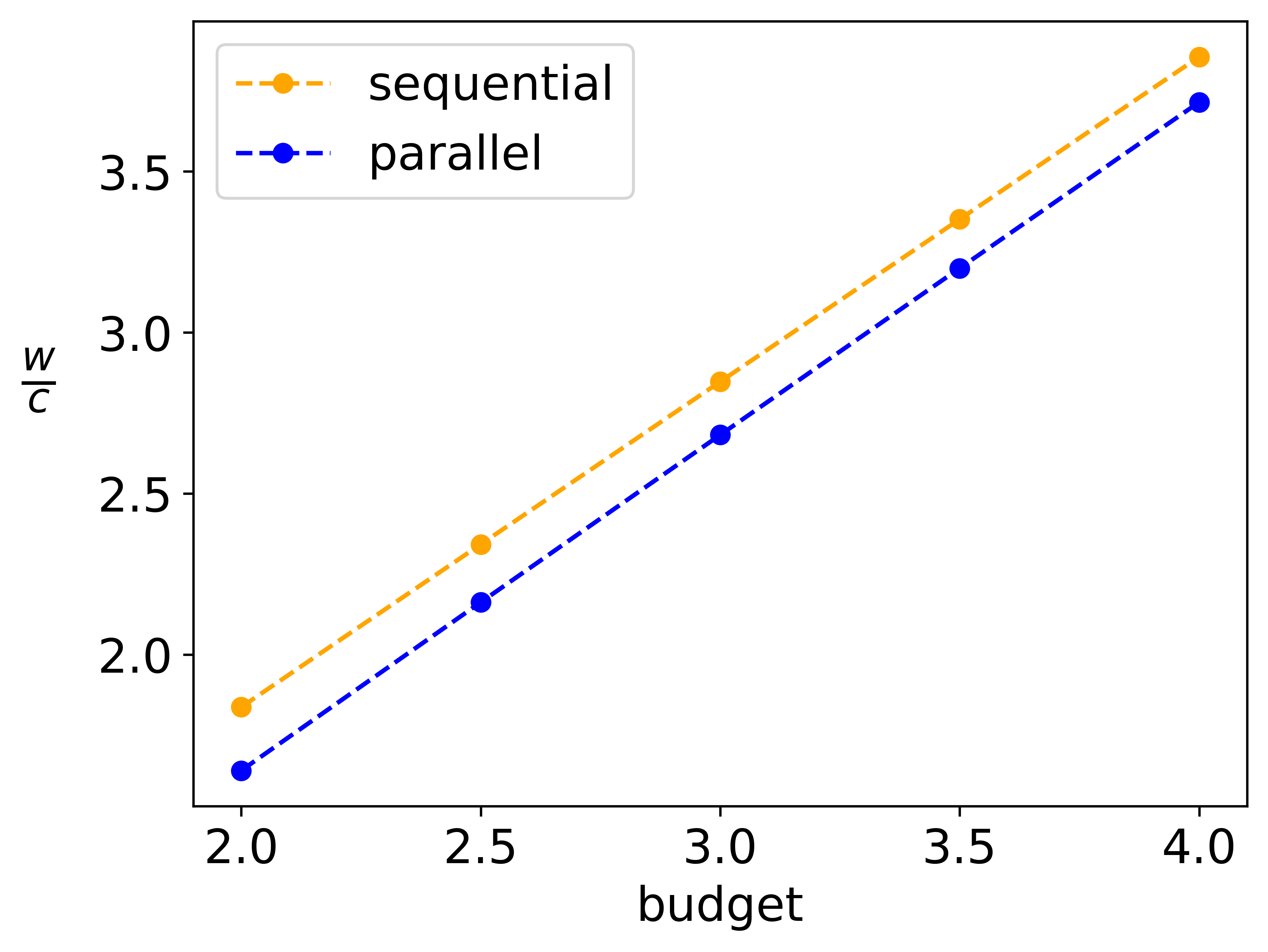}
    \caption{$\frac{w}{c}$ v.s. budget} \label{fig:wage_cost_budget}
  \end{subfigure}%
  \hspace*{\fill}   
  \begin{subfigure}{0.48\columnwidth}
    \includegraphics[width=\columnwidth]{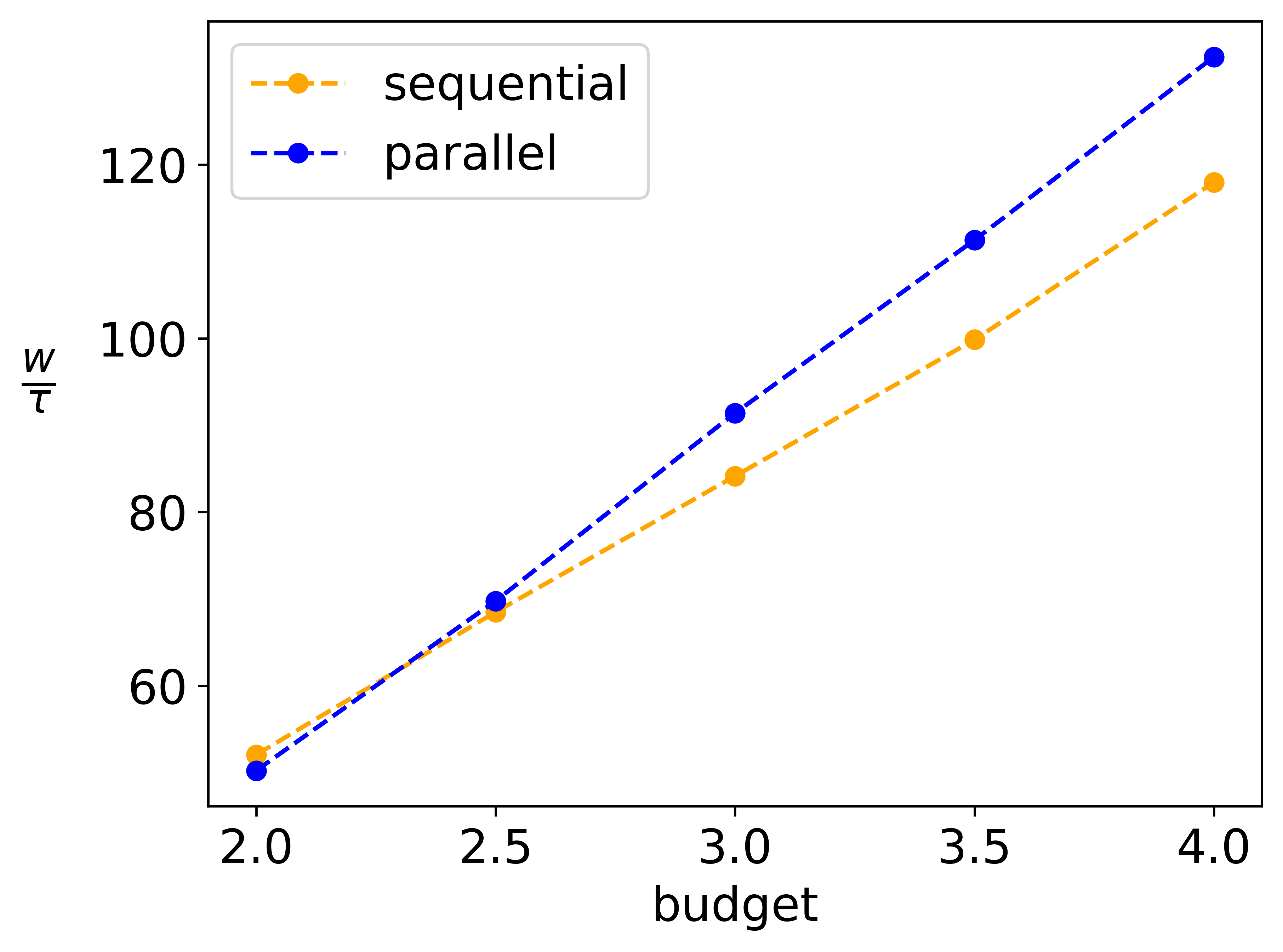}
    \caption{$\frac{w}{\tau}$ v.s. budget} \label{fig:wage_time_budget}
  \end{subfigure}%
\caption{Impact of clients' budgets}
\label{fig:budget}
\end{figure}

Figure~\ref{fig:budget} shows 
the comparisons as $b$, the clients' average budget per task, varies 
while other parameters take their default values.
As we can see from Figure~\ref{fig:delay_budget},
the average delay of task drops for both algorithms when the budget rises,
and the parallel releasing algorithm has faster decrease.
This is because: First, 
the delay as function $f(x,i,j)$ for each task $t_{i,j}$ as well as 
its totally assigned deposit and budget $x$, has the property that 
$f'(x,i,j)<0$ and $f''(x,i,j)>0$; hence,  
it decreases as each task is assigned with more budget 
while other conditions remain the same.
Second, the sequential releasing algorithm assigns more deposit to each released task
than the parallel one; hence, due to $f''(x,i,j)<0$,
a task released by the parallel releasing algorithm has smaller sum of deposit and budget,
and thus its delay decreases faster as its budget increases.
Similarly, the wage from a task as function $w(x,i,j)$ has the property that
$w'(x,i,j)>0$ and $w''(x,i,j)<0$. Hence, it explains that, 
as shown in Figure~\ref{fig:wage_budget_budget} and \ref{fig:wage_cost_budget}, 
both $\frac{w}{c}$ and $\frac{w}{b}$ increase with $b$ for both algorithms
while the parallel releasing algorithm has faster increase. 
Resulting from the above trends, 
as shown in Figure~\ref{fig:wage_time_budget},
$\frac{w}{\tau}$ increases with budget and the parallel releasing algorithm has faster increase.


\comment{
\begin{figure}[htb!]
  \begin{subfigure}{0.48\columnwidth}
    \includegraphics[width=\columnwidth]{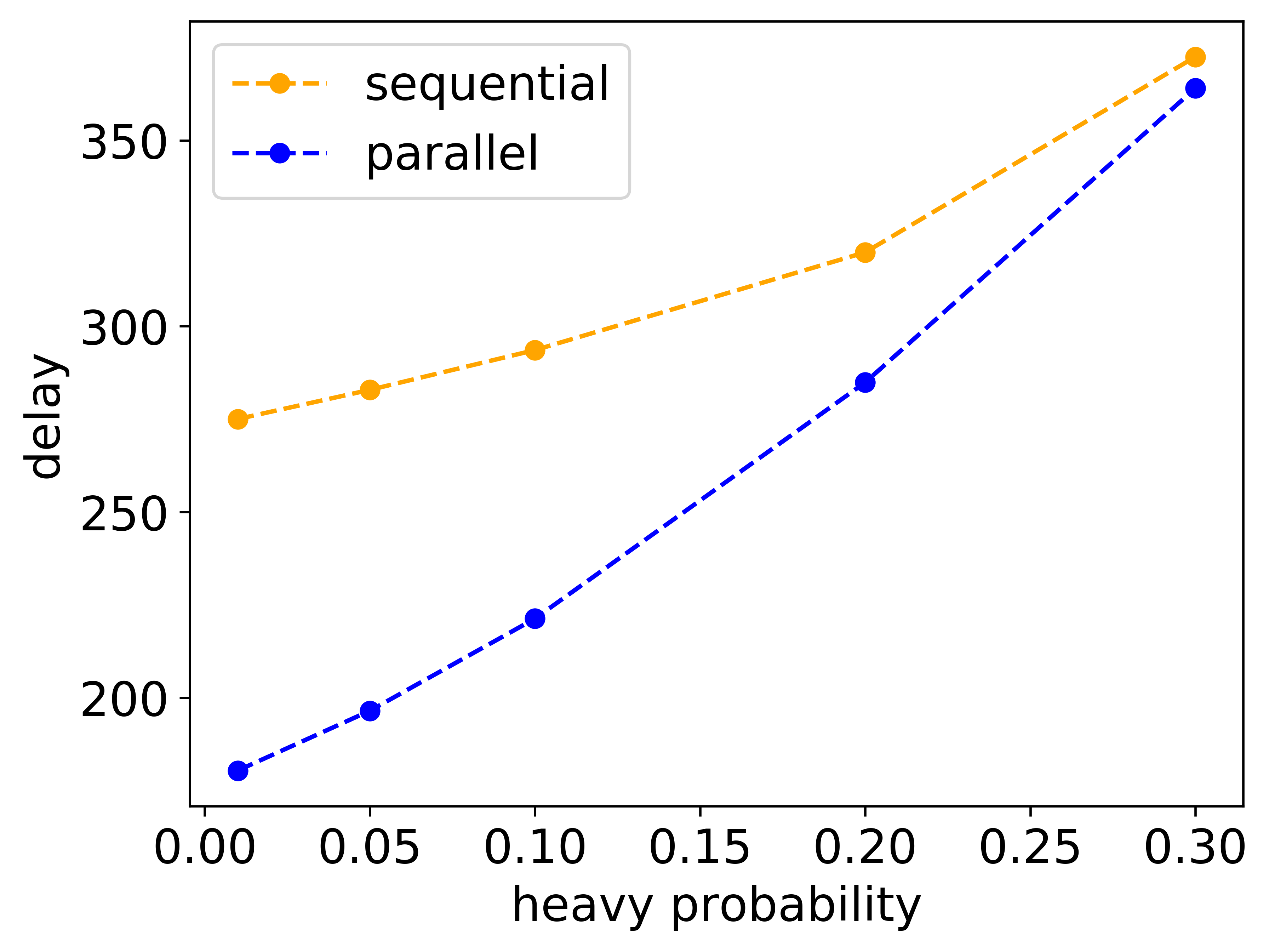}
    \caption{delay v.s. heavy probability} \label{fig:delay_heavyprob}
  \end{subfigure}%
  \hspace*{\fill}   
  \begin{subfigure}{0.48\columnwidth}
    \includegraphics[width=\columnwidth]{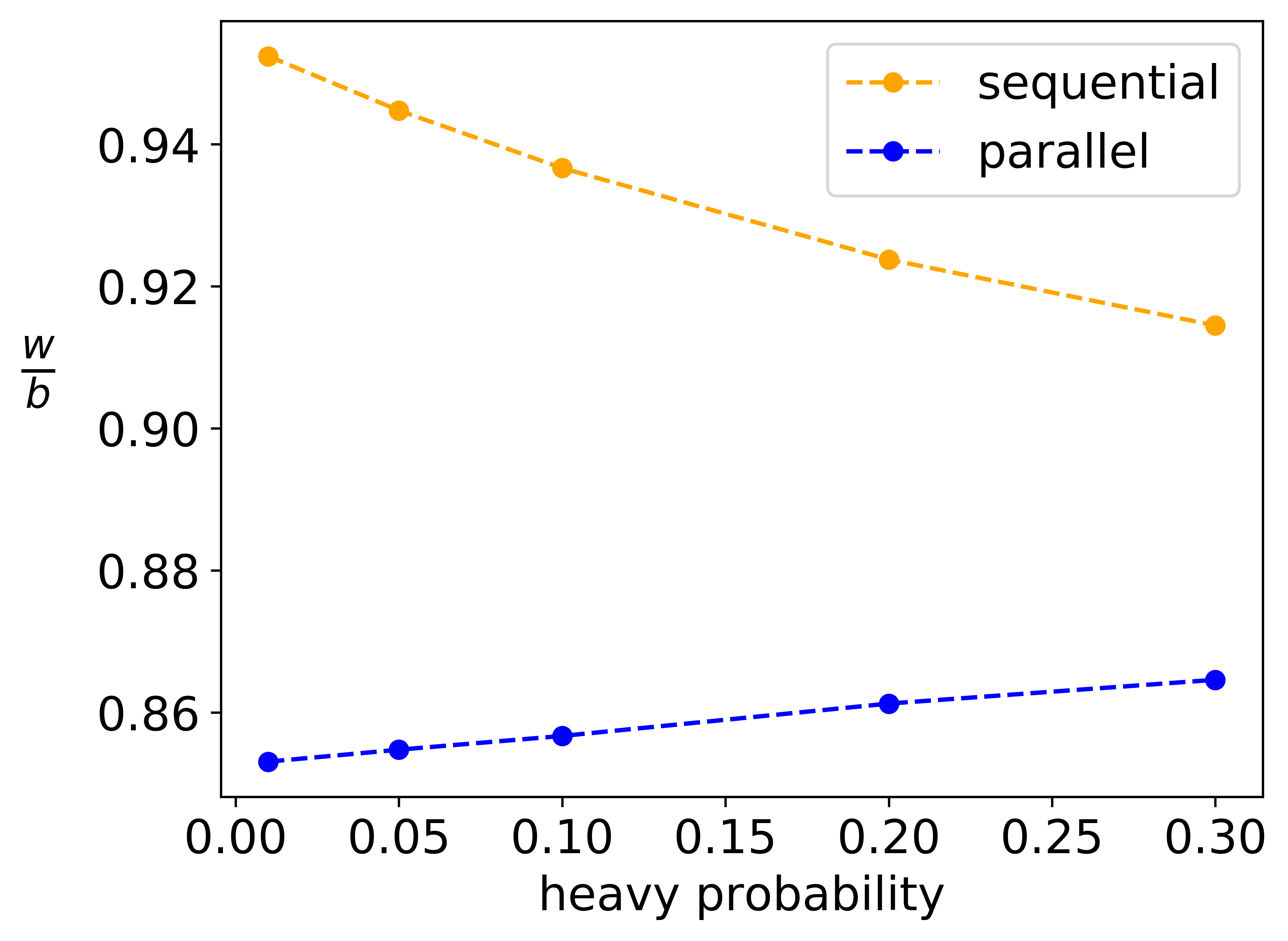}
    \caption{$\frac{w}{b}$ v.s. heavy probability} \label{fig:wage_budget_heavyprob}
  \end{subfigure}%
  \hfill   
  \begin{subfigure}{0.48\columnwidth}
    \includegraphics[width=\columnwidth]{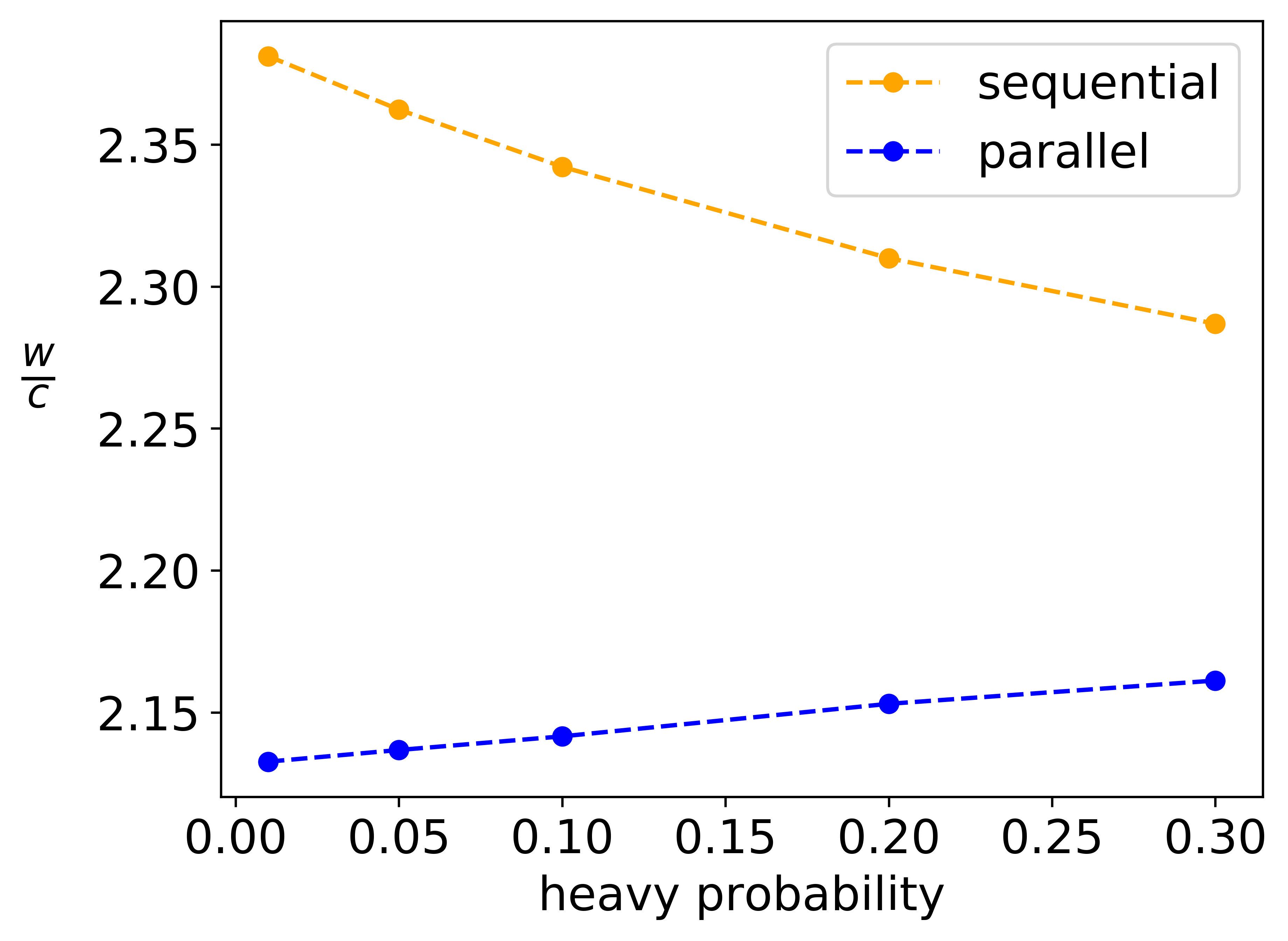}
    \caption{$\frac{w}{c}$ v.s. heavy probability} \label{fig:wage_cost_heavyprob}
  \end{subfigure}%
  \hspace*{\fill}   
  \begin{subfigure}{0.48\columnwidth}
    \includegraphics[width=\columnwidth]{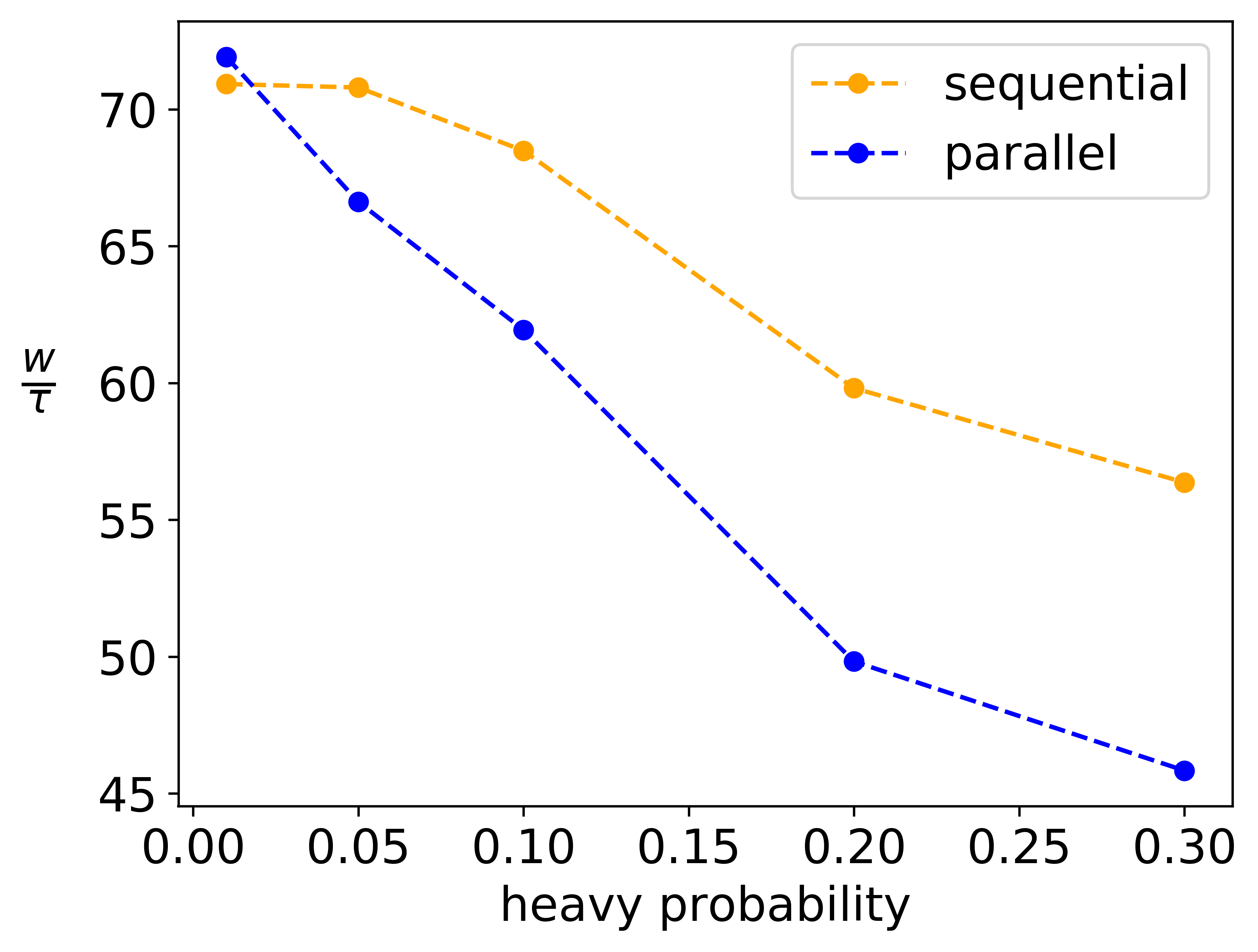}
    \caption{$\frac{w}{\tau}$ v.s. heavy probability} \label{fig:wage_time_heavyprob}
  \end{subfigure}%
\caption{Impact of heavy probability}
\label{fig:heavyprob}
\end{figure}

Figure~\ref{fig:heavyprob} shows how the delay and wage change with the probability of a task to be heavy when $T_{confirm}=2\tau$, $c=32\tau$, $b=2.5c$, $d=\frac{120c}{\tau}$ and $p_{heavy}\in\{0.01,0.05,0.1,0.2,0.3\}$. When heavy probability increases, the delay for both sequential and parallel increase (Fig~\ref{fig:delay_heavyprob}) which is reasonable as the average computation delay, the TTP delay as well as the release delay all increase. 

In sequential releasing, $\frac{w}{b}$ and $\frac{w}{c}$ decrease as $p_{heavy}$ increases. Because in the sequential releasing, each task is assigned with all available deposits $d=\frac{120c}{\tau}$. For a regular task with cost $c$ and a heavy task with cost $3c$, in order to make these two tasks have the same $\frac{w}{c}$, the budget and deposit assigned to heavy task should have 3 times to the value of a regular task. However, the deposit assigned to each task are the same, is $d$, which means heavy task has a smaller $\frac{w}{c}$. Therefore, when the heavy task probability increases, the average $\frac{w}{c}$ also increases which is verified in Fig~\ref{fig:wage_cost_heavyprob}. 
In parallel releasing, $\frac{w}{b}$ and $\frac{w}{c}$ increases with the increasing of $p_{heavy}$. As we take $\alpha=1$ and $\beta=0.6$ in the parallel releasing, each round there are at most four regular tasks released and only one heavy task can be released at a time. That is, when in the parallel release round with four regular tasks, each task only assigned with $\frac{d}{4}$. For each heavy task, it always assigned with $d$. Therefore, the higher percentage of heavy task leads to higher $\frac{w}{b}$ and $\frac{w}{c}$. 
In Fig~\ref{fig:wage_time_heavyprob}, $\frac{w}{\tau}$ decreases as $p_{heavy}$ increases. That is, when more heavy tasks streamed in, if the total deposits in the system is fixed, the wage earned per time interval decreases.
}

\begin{figure}[htb!]
  \begin{subfigure}{0.48\columnwidth}
    \includegraphics[width=\columnwidth]{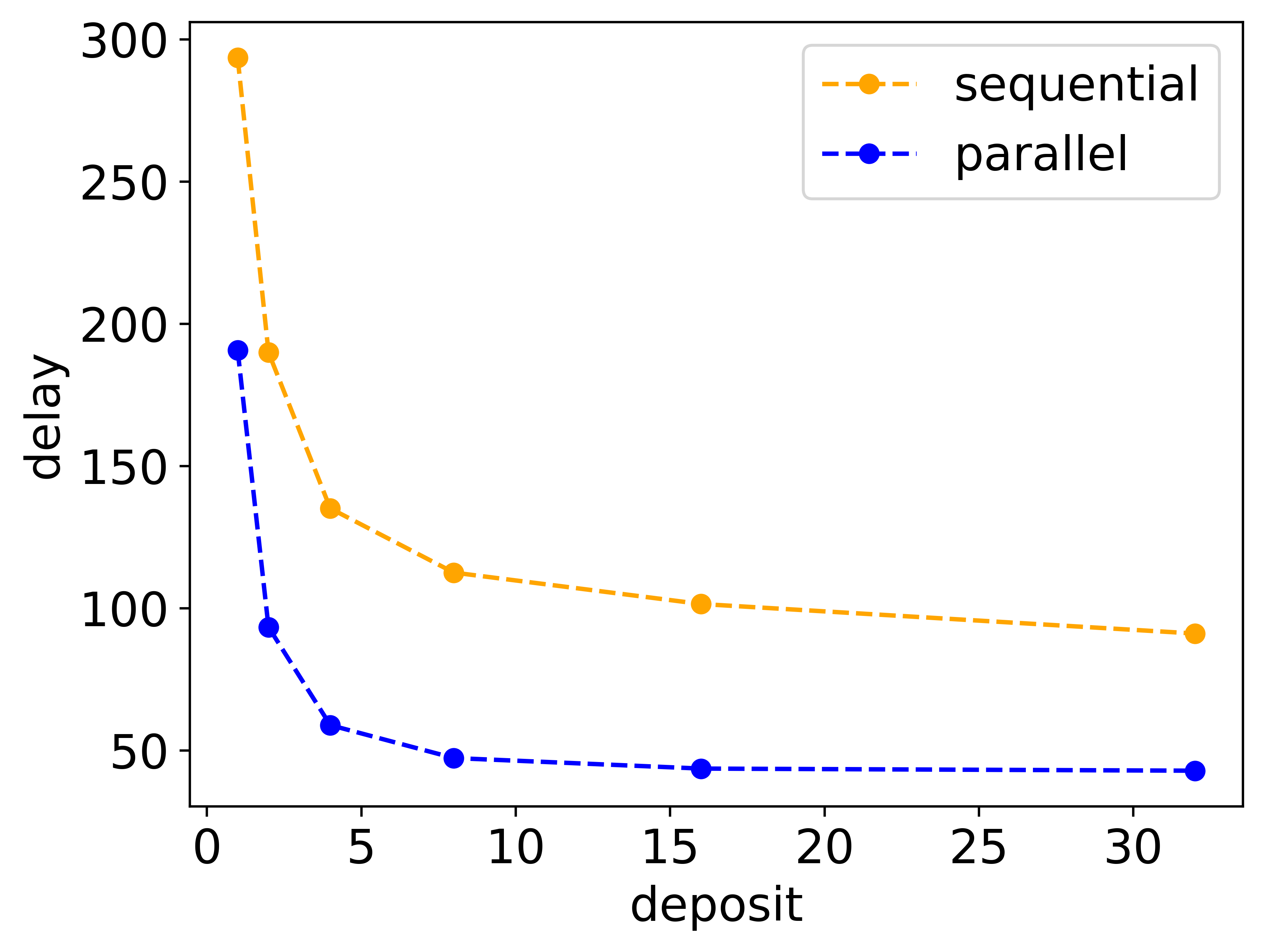}
    \caption{delay v.s. deposit} \label{fig:delay_deposit}
  \end{subfigure}%
  \hspace*{\fill}   
  \begin{subfigure}
  {0.48\columnwidth}
    \includegraphics[width=\columnwidth]{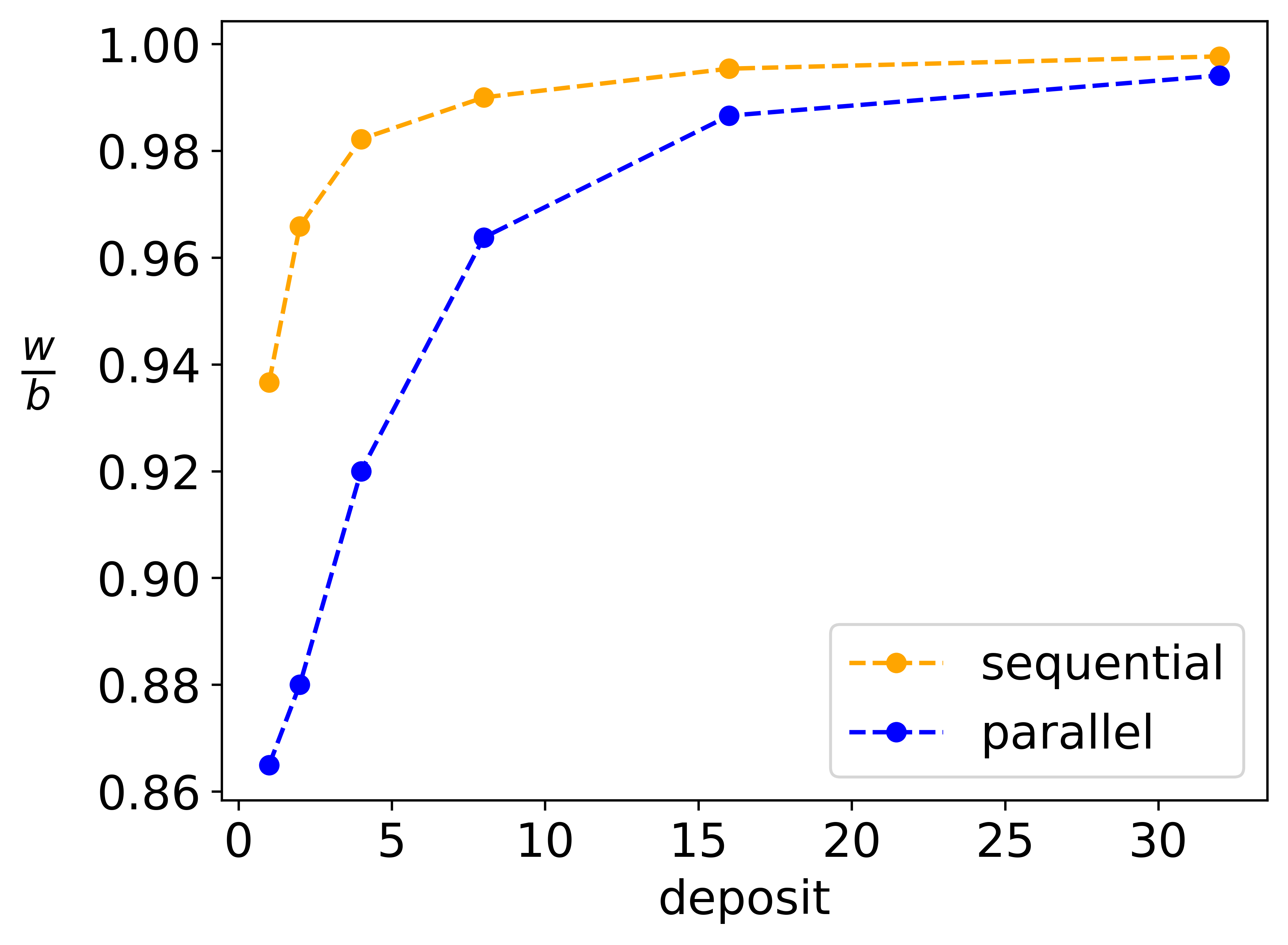}
    \caption{$\frac{w}{b}$ v.s. deposit} \label{fig:wage_budget_deposit}
  \end{subfigure}%
  \hfill   
  \begin{subfigure}{0.48\columnwidth}
    \includegraphics[width=\columnwidth]{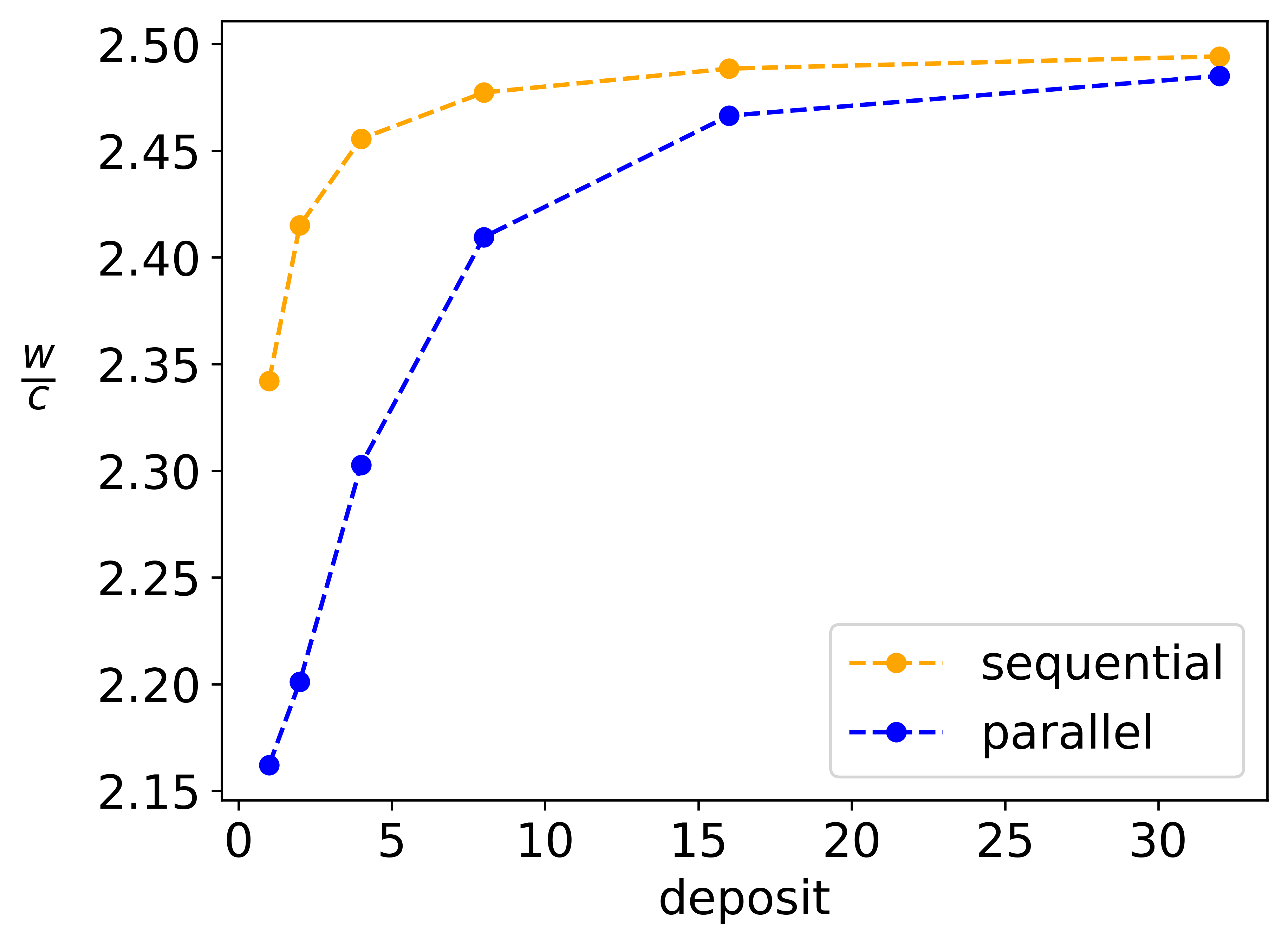}
    \caption{$\frac{w}{c}$ v.s. deposit} \label{fig:wage_cost_deposit}
  \end{subfigure}%
  \hspace*{\fill}   
  \begin{subfigure}{0.48\columnwidth}
    \includegraphics[width=\columnwidth]{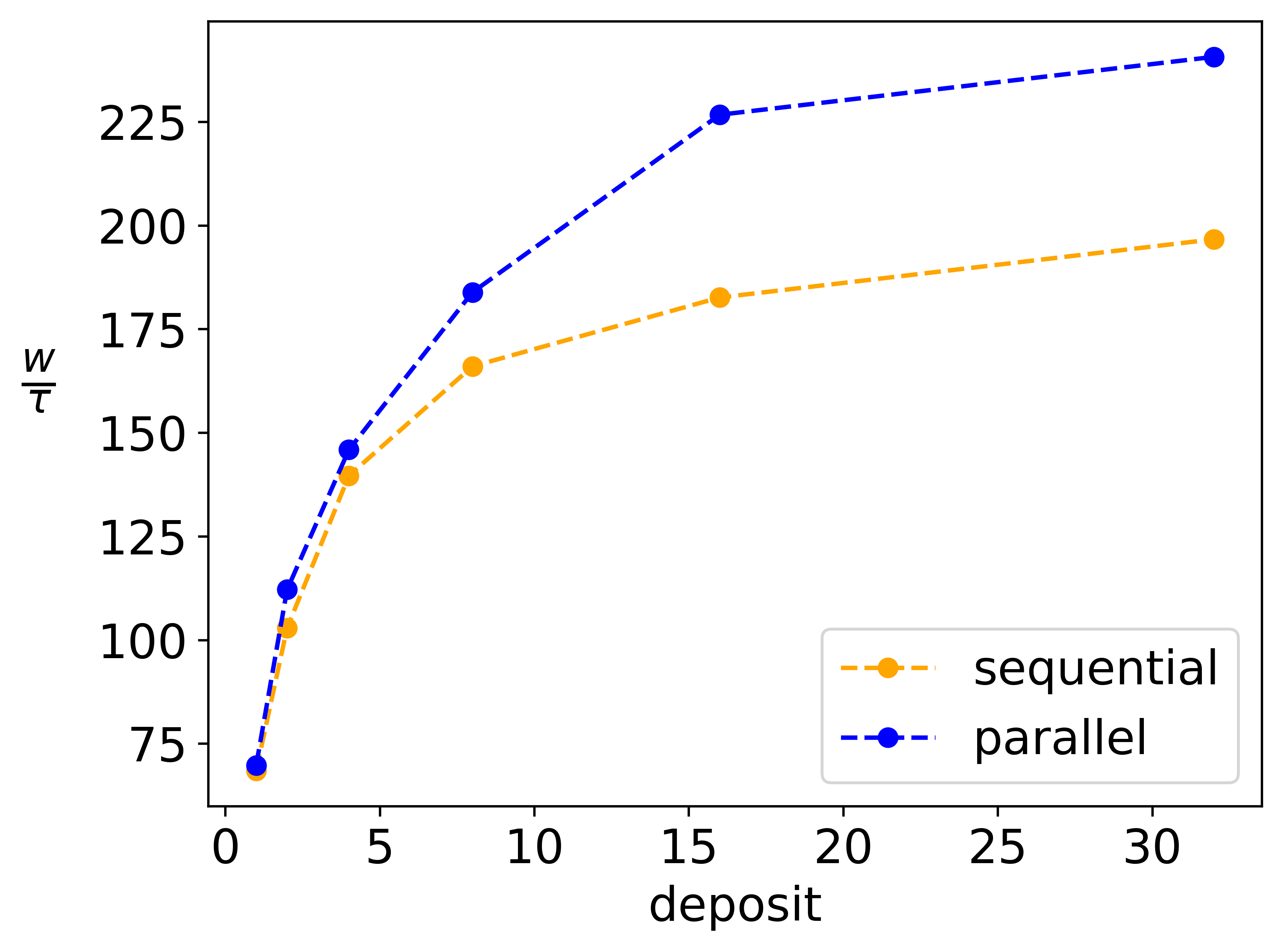}
    \caption{$\frac{w}{\tau}$ v.s. deposit} \label{fig:wage_time_deposit}
  \end{subfigure}%
\caption{Impact of server's deposits}
\label{fig:deposit}
\end{figure}

Figure~\ref{fig:deposit} shows the comparisons 
as the server's budget for deposit, i.e., $d$ varies from $d_0=768\tau$ to $32d_0$, 
while the other parameters take their default values.
As demonstrated by Figure~\ref{fig:delay_deposit},
the average delay of task decreases for both the sequential and parallel algorithms 
as the server's fund for budget increases.
This is because, as the fund increases, 
every released task is assigned with larger deposit,
which results in lower probability of hiring TTP and thus lower delay. 
The larger deposits also lead to larger wage that the server can earn from each released task,
as shown in Figures~\ref{fig:wage_budget_deposit} and \ref{fig:wage_cost_deposit}.
These figures also demonstrate that, 
the parallel algorithm increases the wage at higher speed.
This is because, as we know from the analysis in the previous sections, 
the wage is an increasing function of deposit and its second derivative is negative;
hence, the increase of wages earned from multiple concurrently-released tasks
is higher than the increase of wage earned from only one single task,
when the wage increase is due to the same amount of increased deposit. 
Resulting from the above, as shown in Figure~\ref{fig:wage_time_deposit},
$\frac{w}{\tau}$ is also increased with the fund for deposit, 
and the parallel algorithm is shown to have faster increase in $\frac{w}{\tau}$ 
than the sequential one. 

To summarize, 
we find that the following from the simulations:
\begin{itemize}
    \item 
    The sequential releasing algorithm earns higher wage per task than the parallel algorithm.
    \item 
    In generally, the parallel releasing algorithm incurs lower average delay of task than the sequential algorithm. 
    \item 
    When the server's deposit is much larger than the cost of each individual task 
    (which is common in practical cloud computing systems), the parallel releasing algorithm 
    can earn a similar level of wage per task as the parallel algorithm while incurring much lower delay of task;
    hence, the parallel algorithm is preferred in practice. 
\end{itemize}




\section{Conclusion and Future Works}
\label{sec::conclusion}

In this paper, 
we study the verifiable computation outsourcing problem in the setting where 
a cloud server services a static set or a dynamic sequence of tasks 
submitted by multiple clients.
We adopt a game-based model, where
the cloud server should make a deposit for each task it takes,
each client should allocate a budget that 
includes the wage paid to the server and the possible cost for hiring TTP 
for each task it submits, and
every party (i.e., each of the server and the clients) 
has its limited fund that can be used for either deposits or task budgets.
We study how the funds should be optimally allocated to achieve the three-fold goals:
a rational cloud server should honestly compute each task it takes;
the server's wages earned from computing the tasks are maximized;
and the overall delay experienced by each task for verifying her tasks is minimized.
Specifically, 
we apply game theory to formulate the optimization problems,
and develop the optimal or heuristic solutions 
for three application scenarios:
one client outsources a static set of tasks to the server;
multiple clients outsource a static set of tasks to the server;
multiple clients outsource a dynamic sequence of tasks to the server.
For each of the solutions,
we analyze the solutions through either
rigorous proofs or extensive simulations.

In the future, we will study in more depth the setting 
where there are multiple clients submitting dynamic sequences of tasks to the server. 
As it is challenging to develop optimal solution for the currently-defined general setting, 
we will explore to refine the problem with reasonable constraints
and then develop an optimal solution for it.

\ifCLASSOPTIONcompsoc
  \section*{Acknowledgments}
\else
  \section*{Acknowledgment}
\fi

The work is partly sponsored by NSF under grant CNS-1844591.

\ifCLASSOPTIONcaptionsoff
  \newpage
\fi




\bibliographystyle{IEEEtran}
\bibliography{reference}

\begin{thebibliography}{10}
\providecommand{\url}[1]{#1}
\csname url@samestyle\endcsname
\providecommand{\newblock}{\relax}
\providecommand{\bibinfo}[2]{#2}
\providecommand{\BIBentrySTDinterwordspacing}{\spaceskip=0pt\relax}
\providecommand{\BIBentryALTinterwordstretchfactor}{4}
\providecommand{\BIBentryALTinterwordspacing}{\spaceskip=\fontdimen2\font plus
\BIBentryALTinterwordstretchfactor\fontdimen3\font minus
  \fontdimen4\font\relax}
\providecommand{\BIBforeignlanguage}[2]{{%
\expandafter\ifx\csname l@#1\endcsname\relax
\typeout{** WARNING: IEEEtran.bst: No hyphenation pattern has been}%
\typeout{** loaded for the language `#1'. Using the pattern for}%
\typeout{** the default language instead.}%
\else
\language=\csname l@#1\endcsname
\fi
#2}}
\providecommand{\BIBdecl}{\relax}
\BIBdecl

\bibitem{gennaro2010non}
R.~Gennaro, C.~Gentry, and B.~Parno, ``Non-interactive verifiable computing:
  Outsourcing computation to untrusted workers,'' in \emph{Annual Cryptology
  Conference}.\hskip 1em plus 0.5em minus 0.4em\relax Springer, 2010, pp.
  465--482.

\bibitem{parno2012delegate}
B.~Parno, M.~Raykova, and V.~Vaikuntanathan, ``How to delegate and verify in
  public: Verifiable computation from attribute-based encryption,'' in
  \emph{Theory of Cryptography Conference}.\hskip 1em plus 0.5em minus
  0.4em\relax Springer, 2012, pp. 422--439.

\bibitem{catalano2013practical}
D.~Catalano and D.~Fiore, ``Practical homomorphic macs for arithmetic
  circuits,'' in \emph{Annual International Conference on the Theory and
  Applications of Cryptographic Techniques}.\hskip 1em plus 0.5em minus
  0.4em\relax Springer, 2013, pp. 336--352.

\bibitem{parno2013pinocchio}
B.~Parno, J.~Howell, C.~Gentry, and M.~Raykova, ``Pinocchio: Nearly practical
  verifiable computation,'' in \emph{2013 IEEE Symposium on Security and
  Privacy}.\hskip 1em plus 0.5em minus 0.4em\relax IEEE, 2013, pp. 238--252.

\bibitem{abadi2016vd}
A.~Abadi, S.~Terzis, and C.~Dong, ``Vd-psi: Verifiable delegated private set
  intersection on outsourced private datasets,'' in \emph{International
  Conference on Financial Cryptography and Data Security}.\hskip 1em plus 0.5em
  minus 0.4em\relax Springer, 2016, pp. 149--168.

\bibitem{costello2015geppetto}
C.~Costello, C.~Fournet, J.~Howell, M.~Kohlweiss, B.~Kreuter, M.~Naehrig,
  B.~Parno, and S.~Zahur, ``Geppetto: Versatile verifiable computation,'' in
  \emph{2015 IEEE Symposium on Security and Privacy}.\hskip 1em plus 0.5em
  minus 0.4em\relax IEEE, 2015, pp. 253--270.

\bibitem{fiore2016hash}
D.~Fiore, C.~Fournet, E.~Ghosh, M.~Kohlweiss, O.~Ohrimenko, and B.~Parno,
  ``Hash first, argue later: Adaptive verifiable computations on outsourced
  data,'' in \emph{Proceedings of the 2016 ACM SIGSAC Conference on Computer
  and Communications Security}, 2016, pp. 1304--1316.

\bibitem{pepper}
S.~T. Setty, R.~McPherson, A.~J. Blumberg, and M.~Walfish, ``Making argument
  systems for outsourced computation practical (sometimes).'' in \emph{NDSS},
  vol.~1, no.~9, 2012, p.~17.

\bibitem{ginger}
S.~Setty, V.~Vu, N.~Panpalia, B.~Braun, A.~J. Blumberg, and M.~Walfish,
  ``Taking proof-based verified computation a few steps closer to
  practicality,'' in \emph{Presented as part of the 21st $\{$USENIX$\}$
  Security Symposium ($\{$USENIX$\}$ Security 12)}, 2012, pp. 253--268.

\bibitem{goldwasser2015delegating}
S.~Goldwasser, Y.~T. Kalai, and G.~N. Rothblum, ``Delegating computation:
  interactive proofs for muggles,'' \emph{Journal of the ACM (JACM)}, vol.~62,
  no.~4, pp. 1--64, 2015.

\bibitem{ben2016interactive}
E.~Ben-Sasson, A.~Chiesa, and N.~Spooner, ``Interactive oracle proofs,'' in
  \emph{Theory of Cryptography Conference}.\hskip 1em plus 0.5em minus
  0.4em\relax Springer, 2016, pp. 31--60.

\bibitem{wahby2017full}
R.~S. Wahby, Y.~Ji, A.~J. Blumberg, A.~Shelat, J.~Thaler, M.~Walfish, and
  T.~Wies, ``Full accounting for verifiable outsourcing,'' in \emph{Proceedings
  of the 2017 ACM SIGSAC Conference on Computer and Communications Security},
  2017, pp. 2071--2086.

\bibitem{brandenburger2018blockchain}
M.~Brandenburger, C.~Cachin, R.~Kapitza, and A.~Sorniotti, ``Blockchain and
  trusted computing: Problems, pitfalls, and a solution for hyperledger
  fabric,'' \emph{arXiv preprint arXiv:1805.08541}, 2018.

\bibitem{xiao2019enforcing}
Y.~Xiao, N.~Zhang, W.~Lou, and Y.~T. Hou, ``Enforcing private data usage
  control with blockchain and attested off-chain contract execution,''
  \emph{arXiv preprint arXiv:1904.07275}, 2019.

\bibitem{cheng2019ekiden}
R.~Cheng, F.~Zhang, J.~Kos, W.~He, N.~Hynes, N.~Johnson, A.~Juels, A.~Miller,
  and D.~Song, ``Ekiden: A platform for confidentiality-preserving,
  trustworthy, and performant smart contracts,'' in \emph{2019 IEEE European
  Symposium on Security and Privacy (EuroS\&P)}.\hskip 1em plus 0.5em minus
  0.4em\relax IEEE, 2019, pp. 185--200.

\bibitem{tramer2018slalom}
F.~Tramer and D.~Boneh, ``Slalom: Fast, verifiable and private execution of
  neural networks in trusted hardware,'' \emph{arXiv preprint
  arXiv:1806.03287}, 2018.

\bibitem{canetti2011practical}
R.~Canetti, B.~Riva, and G.~N. Rothblum, ``Practical delegation of computation
  using multiple servers,'' in \emph{Proceedings of the 18th ACM conference on
  Computer and communications security}, 2011, pp. 445--454.

\bibitem{avizheh2019verifiable}
S.~Avizheh, M.~Nabi, R.~Safavi-Naini, and M.~Venkateswarlu~K, ``Verifiable
  computation using smart contracts,'' in \emph{Proceedings of the 2019 ACM
  SIGSAC Conference on Cloud Computing Security Workshop}, 2019, pp. 17--28.

\bibitem{nix2012contractual}
R.~Nix and M.~Kantarcioglu, ``Contractual agreement design for enforcing
  honesty in cloud outsourcing,'' in \emph{International Conference on Decision
  and Game Theory for Security}.\hskip 1em plus 0.5em minus 0.4em\relax
  Springer, 2012, pp. 296--308.

\bibitem{optimal}
V.~Pham, M.~Khouzani, and C.~Cid, ``Optimal contracts for outsourced
  computation,'' in \emph{International Conference on Decision and Game Theory
  for Security}.\hskip 1em plus 0.5em minus 0.4em\relax Springer, 2014, pp.
  79--98.

\bibitem{belenkiy2008incentivizing}
M.~Belenkiy, M.~Chase, C.~C. Erway, J.~Jannotti, A.~K{\"u}p{\c{c}}{\"u}, and
  A.~Lysyanskaya, ``Incentivizing outsourced computation,'' in
  \emph{Proceedings of the 3rd international workshop on Economics of networked
  systems}, 2008, pp. 85--90.

\bibitem{incentive}
M.~Khouzani, V.~Pham, and C.~Cid, ``Incentive engineering for outsourced
  computation in the face of collusion,'' in \emph{Proceedings of WEIS}, 2014.

\bibitem{kupccu2015incentivized}
A.~K{\"u}p{\c{c}}{\"u}, ``Incentivized outsourced computation resistant to
  malicious contractors,'' \emph{IEEE Transactions on Dependable and Secure
  Computing}, vol.~14, no.~6, pp. 633--649, 2015.

\bibitem{dong2017betrayal}
C.~Dong, Y.~Wang, A.~Aldweesh, P.~McCorry, and A.~van Moorsel, ``Betrayal,
  distrust, and rationality: Smart counter-collusion contracts for verifiable
  cloud computing,'' in \emph{Proceedings of the 2017 ACM SIGSAC Conference on
  Computer and Communications Security}, 2017, pp. 211--227.

\bibitem{liu2018new}
P.~Liu and W.~Zhang, ``A new game theoretic scheme for verifiable cloud
  computing,'' in \emph{2018 IEEE 37th International Performance Computing and
  Communications Conference (IPCCC)}.\hskip 1em plus 0.5em minus 0.4em\relax
  IEEE, 2018, pp. 1--8.

\bibitem{liu2020game}
------, ``Game theoretic approach for secure and efficient heavy-duty smart
  contracts,'' in \emph{2020 IEEE Conference on Communications and Network
  Security (CNS)}.\hskip 1em plus 0.5em minus 0.4em\relax IEEE, 2020, pp. 1--9.

\bibitem{walfish2015verifying}
M.~Walfish and A.~J. Blumberg, ``Verifying computations without reexecuting
  them,'' \emph{Communications of the ACM}, vol.~58, no.~2, pp. 74--84, 2015.

\bibitem{lu2018enabling}
Y.~Lu, Q.~Tang, and G.~Wang, ``On enabling machine learning tasks atop public
  blockchains: A crowdsourcing approach,'' in \emph{2018 IEEE International
  Conference on Data Mining Workshops (ICDMW)}.\hskip 1em plus 0.5em minus
  0.4em\relax IEEE, 2018, pp. 81--88.

\bibitem{teng2010new}
F.~Teng and F.~Magoul{\`e}s, ``A new game theoretical resource allocation
  algorithm for cloud computing,'' in \emph{International Conference on Grid
  and Pervasive Computing}.\hskip 1em plus 0.5em minus 0.4em\relax Springer,
  2010, pp. 321--330.

\bibitem{teng2010resource}
F.~Teng and F.~Magoules, ``Resource pricing and equilibrium allocation policy
  in cloud computing,'' in \emph{2010 10th IEEE International Conference on
  Computer and Information Technology}.\hskip 1em plus 0.5em minus 0.4em\relax
  IEEE, 2010, pp. 195--202.

\bibitem{wei2010game}
G.~Wei, A.~V. Vasilakos, Y.~Zheng, and N.~Xiong, ``A game-theoretic method of
  fair resource allocation for cloud computing services,'' \emph{The journal of
  supercomputing}, vol.~54, no.~2, pp. 252--269, 2010.

\bibitem{kaewpuang2013framework}
R.~Kaewpuang, D.~Niyato, P.~Wang, and E.~Hossain, ``A framework for cooperative
  resource management in mobile cloud computing,'' \emph{IEEE Journal on
  Selected Areas in Communications}, vol.~31, no.~12, pp. 2685--2700, 2013.

\bibitem{pillai2014resource}
P.~S. Pillai and S.~Rao, ``Resource allocation in cloud computing using the
  uncertainty principle of game theory,'' \emph{IEEE Systems Journal}, vol.~10,
  no.~2, pp. 637--648, 2014.

\bibitem{xu2014game}
X.~Xu and H.~Yu, ``A game theory approach to fair and efficient resource
  allocation in cloud computing,'' \emph{Mathematical Problems in Engineering},
  vol. 2014, 2014.

\end{thebibliography}

\section*{Appendix 1 Proof of Lemma~\ref{lem:singleClient-1}}

For the optimization problem defined in (\ref{eq:min_delay-revised2}),
it is obvious that every constraint is convex.
Next, we only need to show that the objective function is straightly convex.
To simplify presentation, we introduce vector $\vec{s}_i$ to denote
$\langle s_{i,1}, \cdots, s_{i,n_i}\rangle$, and $f(\vec{s}_i)$
to denote the objective function. 
We further use $K_{i,j}$ to denote $c_{i,j}\vec{c}_{i,j}$, which is a constant. 
The Hessian matrix of $f(\vec{s}_i)$ has only $n_i$ non-zero diagonal elements, i.e., 
    $\frac{\partial^2 f(\vec{s}_i)}{\partial s^2_{i,j}} 
    = \frac{2K_{i,j}}{(s_{i,j}^2-4K_{i,j})^{\frac{3}{2}}} > 0$, 
where $j \in \{1,2,\dots,n_i\}$. 
Hence, the Hessian matrix of $f(\vec{s}_i)$, denoted as $H_i$, has the form of
\begin{equation}\nonumber
    \begin{pmatrix}
\frac{2K_{i,1}}{(s_{i,1}^2-4K_{i,1})^{\frac{3}{2}}} & ... & 0 & ... & 0 \\
0 & \frac{2K_{i,2}}{(s_{i,2}^2-4K_{i,2})^{\frac{3}{2}}} &  0 &  ... & 0  \\
\vdots & \vdots & \vdots & \vdots & \vdots\\
0 & ... & ... & ... & 0 \\
0 & ... & ... & ... & \frac{2K_{i,n_i}}{(s_{i,n_i}^2-4K_{i,n_i})^{\frac{3}{2}}} \\
\end{pmatrix} \\
\end{equation} 

Letting $z=(z_1, ..., z_{n_i})^T$ be any non-zero column vector, 
where $T$ denotes transpose, we have
\begin{align}
\begin{pmatrix}
    z_1, \ldots, z_{n}
  \end{pmatrix}
  H_i
  \begin{pmatrix}
    z_1 \\
    \vdots\\
    z_{n}
  \end{pmatrix}
  \notag 
= 
\sum_{j = 1}^{n_i} z_j^2\frac{2K_{i,j}}{(s_{i,j}^2-4K_{i,j})^{\frac{3}{2}}} 
> 
0,\notag
\end{align}
where $H_i$ is positive definite and 
$f(\vec{x}_i)$ is strictly convex. 

\section*{Appendix 2 Proof of Lemma~\ref{lem:singleClient-2}}

Phase I of Algorithm~\ref{alg:SingleClient} is simply to satisfy constraint~(\ref{eq:constraint-revised2}),
Hence we focus on studying Phase II, which allocate the remaining fund greedily. 
In the following, we prove (by contradiction) that the strategy of greedy allocation in Phase II 
leads to a solution of the optimization problem~(\ref{eq:min_delay-revised}),
which is unique due to Lemma~\ref{lem:singleClient-1}.

Assume the optimal allocation strategy for the remaining fund is 
$\vec{o}=\langle o_{i,1},o_{i,2},\cdots,o_{i,n_i}\rangle$, 
where each $o_{i,j}$ is the sum of budget and deposit assigned to task $t_{i,j}$ from the remaining fund; 
but Phase II leads to a different allocation strategy 
$\vec{g}=\langle g_{i,1},g_{i,2},\cdots,g_{i,n_i}\rangle$.
Then, there are some tasks (at least one) which receive more or less allocation than the optimal solution;
let $M$ and $L$ denote the task sets respectively.
Let $\vec{c}= \langle c_{i,1}, c_{i,2},\dots,c_{i,n_i}\rangle$ 
denote the common parts of $\vec{g}$ and $\vec{o}$; 
that is, $c_j = min(g_{i,j}, o_{i,j})$ for $j\in\{1,\cdots,n_i\}$.

Let $v_c$, $v_g$ and $v_o$ be outcomes of the objective function of \eqref{eq:min_delay-revised} 
given by allocation $\vec{c}$, $\vec{g}$ and $\vec{o}$ respectively. 
As $\vec{o}$ leads to the unique optimal solution, we have:
    $v_g > v_o$. 
We define $diff_{c,o}$ and $diff_{c,g}$ as follows:
    $diff_{c,o} = v_o - v_c = \sum _{t_{i,j} \in L} \int_{c_{i,j}}^{o_{i,j}} f'(x,i,j) \,dx$, 
    $diff_{c,g} = v_g - v_c = \sum _{t_{i,j} \in M} \int_{c_{i,j}}^{g_{i,j}} f'(x,i,j) \,dx$. 
%
Here, $f'(x,i,j)$ is the partial derivative function of $f(x,i,j)$ which has
\begin{equation}
\label{eq:f(x)_first_derivative}
   f'(x,i,j)=\frac{1}{2}(1-\frac{x}{\sqrt{x^2-4c_{i,j}\hat{c}_{i,j}}})<0; 
\end{equation}
that is, $f(x,i,j)$ is decreasing monotonously. 
Let $f''(x,i,j)$ be the partial derivative function of $f'(x,i,j)$ and it holds
 \begin{equation}
 \label{eq:f(x)_second_derivative}
     f''(x,i,j)=\frac{2c_{i,j}\hat{c}_{i,j}}{{\sqrt{x^2-4c_{i,j}\hat{c}_{i,j}}}^3}>0;
 \end{equation}
i.e., $f(x,i,j)$ decreases monotonically but gets slower as $x$ increases 
(as shown by Fig.~\ref{fig:f(x,i,j)}). 
\begin{figure}
    \centering
    \includegraphics[width=0.7\columnwidth]{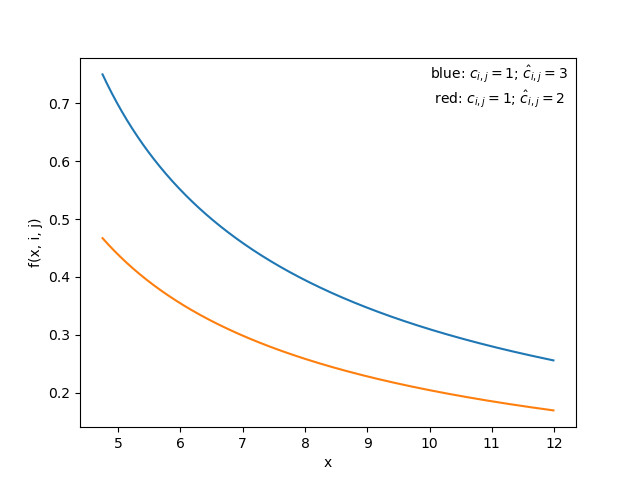}
    \caption{Illustration of $f(x, i, j)$.}
    \label{fig:f(x,i,j)}
\end{figure}

According to \eqref{eq:f(x)_first_derivative} and \eqref{eq:f(x)_second_derivative}, $f'(x,i,j)$ increases as $x_{i,j}$ increases. Therefore, for each task $t_{i,j} \in L$, $f'(x,i,j) > f'(c_{i,j},i,j)$ for $c_{i,j} < x \leq o_{i,j}$. Let $f'_{omin} = min(f'(c_{i,j},i,j))$ for all tasks $t_{i,j} \in L$, then we have 
\begin{align}
\label{eq:diff_{C,O}}
    & diff_{c,o} = \sum _{t_{i,j} \in L} \int_{c_{i,j}}^{o_{i,j}} f'(x,i,j) \,dx \nonumber\\
    & > \sum_{t_{i,j} \in L} \int_{c_{i,j}}^{o_{i,j}} f'(c_{i,j},i,j) \,dx \nonumber\\
    & 
    = \sum_{t_{i,j} \in L} ((o_{i,j} - c_{i,j}) \times f'(c_{i,j},i,j)) \nonumber\\
    & \geq  \sum_{t_{i,j} \in L} (o_{i,j} - c_{i,j}) \times f'_{omin} 
    = \Delta \times f'_{omin},
\end{align}
where $\Delta = \sum_{t_{i,j} \in L} ((o_{i,j} - c_{i,j})$. 
Since the total amount of budget and deposit is fixed, 
we have $\Delta = \sum_{t_{i,j} \in L} (o_{i,j} - c_{i,j}) = \sum_{t_{i,j} \in M} (g_{i,j} - c_{i,j})$. 
Meanwhile, let $f'_{gmax} = max(f'(x,i,j))$ for all tasks $t_{i,j} \in M$ and $c_{i,j} \leq x \leq g_{i,j}$. 
Since our greedy algorithm picks the task with the minimal first-derivative at each step and it chooses tasks in $M$ over tasks in $L$, we have $f'_{gmax} \leq f'_{omin}$. Therefore we have 
\begin{align}
\label{eq:diff_{C,G}}
    &diff_{c,g} = \sum _{t_{i,j} \in M} \int_{c_{i,j}}^{g_{i,j}} f'(x,i,j) \,dx \nonumber \\
    &< \sum_{t_{i,j} \in M} \int_{c_{i,j}}^{g_{i,j}} f'_{gmax} \,dx 
    = \sum_{t_{i,j} \in L} ((o_{i,j} - c_{i,j}) \times f'_{gmax} \nonumber\\
    & = \Delta \times f'_{gmax}  
    \leq \Delta \times f'_{omin} 
    < diff_{c,o},
\end{align}
which means
    $v_g < v_o$,
since $v_g = v_c + diff_{c, g}$ and $v_o = v_c + diff_{c,o}$.
This contradicts with $v_g > v_o$.


\section*{Appendix 3 Proof of Lemma~\ref{lem:singleClient-3}}

By applying Phase III of Algorithm~\ref{alg:SingleClient}, it is obvious that,
any allocation strategy $\langle s_{i,1}, \cdots, s_{i,n_i}\rangle$ produced by Phase I and II,
which is a solution to the optimization problem defined in \eqref{eq:min_delay-revised2},
can always be converted into a {\em detailed allocation strategy} 
$\langle (b_{i,1},d_{i,1}), \cdots, (b_{i,n_i},d_{i,n_i}) \rangle$,
as the execution of Phase III will not abort before it is completed. 

Because the above conversion assures  
$b_{i,j}+d_{i,j}=s_{i,j}$ for each $j\in\{1,\cdots,n_i\}$ and
the objective functions in both \eqref{eq:min_delay-revised} and \eqref{eq:min_delay-revised2}
are determined only by the sums of every task's budget and deposit,
the value of the objective function in \eqref{eq:min_delay-revised} 
by the detailed allocation strategy should be the same as
the minimal value of the objective function in \eqref{eq:min_delay-revised2}.

Moreover, as the optimization problems defined in \eqref{eq:min_delay-revised}
and \eqref{eq:min_delay-revised2} are the same except for that
the former has more constraints than the latter,
the solution to the former cannot be better than that for the latter;
i.e., the minimal value of objective function in \eqref{eq:min_delay-revised}
must be no less than the minimal value of objective function in \eqref{eq:min_delay-revised2}.

Based on the above reasoning, 
the detailed allocation strategy obtained from 
applying Phase III of Algorithm~\ref{alg:SingleClient}
to a solution to \eqref{eq:min_delay-revised2}
must also be a solution to \eqref{eq:min_delay-revised}.

\section*{Appendix 4 Proof of Theorem~\ref{theo:singleClientNE}}

Based on the definition of Nash equilibrium, 
we need to prove that neither the client nor the server has incentive to change strategy $A_{c,i}=(b_{i,1},\cdots,b_{i,n_i})$ or $A_{s,i}=(d_{i,1},\cdots,d_{i,n_i})$, 
which are produced by Algorithm~\ref{alg:SingleClient}, 
if the other player does not change its strategy. 

Let $\vec{s} = \langle s_{i,1},\cdots,s_{i,n_i}\rangle$ denote the unique solution to
the optimization problem defined by \eqref{eq:min_delay-revised2};
that is $s_{i,j}=b_{i,j}+d_{i,j}$ for each $j\in\{1,\cdots,n_i\}$.


First, let us consider the scenario that 
the server keeps strategy $A_{s,i}=(d_{i,1},\cdots,d_{i,n_i})$, 
while the client changes its strategy from $A_{c,i}=(b_{i,1},\cdots,b_{i,n_i})$ 
to a different strategy $A'_{c,i}=\{b'_{i,1},\cdots,b'_{i,n_i}\}$. 
Let $\vec{s}'=\langle s'_{i,1},\cdots,s'_{i,n_i}\rangle$ denote every task's sum of budget and deposit
according to strategy set $(A_{s,i},A'_{c,i})$; i.e.,  
$s'_{i,j}=b'_{i,j}+d_{i,j}$ for $j\in\{1,\cdots,n_i\}$.
Because $\vec{s}\neq\vec{s}'$ and 
$\vec{s}$ is the unique solution to \eqref{eq:min_delay-revised2},
it holds that the client's utility under strategy set $(A_{s,i},A'_{c,i})$,
i.e., $U'_{c,i}=\sum_{j=1}^{n_i}\frac{s'_{i,j}-\sqrt{{s'_{i,j}}^2-4c_{i,j}\hat{c}_{i,j}}}{2}$,
is greater than its utility under strategy set $(A_{s,i},A_{c,i})$,
i.e., $U_{c,i}=\sum_{j=1}^{n_i}\frac{s_{i,j}-\sqrt{s_{i,j}^2-4c_{i,j}\hat{c}_{i,j}}}{2}$.
So, the client should not have incentive to deviate from $A_{c,i}$.

Similarly, let us consider the scenario that
the client keeps strategy $A_{c,i}$ while the server changes from $A_{s,i}$ to 
a different strategy $A'_{s,i}=(d'_{i,1},\cdots,d'_{i,n_i})$.
Let $\vec{s}''=\langle s''_{i,1},\cdots,s''_{i,n_i}\rangle$, where  
$s''_{i,j}=b_{i,j}+d'_{i,j}$ for $j\in\{1,\cdots,n_i\}$.
The server's total wage on strategy set $(A'_{s,i}, A_{c,i})$ is $U'_s=\sum_{j=1}^{n_i}b_{i,j}-\sum_{j=1}^{n_i}\frac{s''_{i,j}-\sqrt{{s''_{i,j}}^2-4c_{i,j}\hat{c}_{i,j}}}{2}$, where $s''_{i,j}=d'_{i,j}+b_{i,j}$ for $j\in\{1,\cdots,n_i\}$. 
The server's total wage on strategy set $(A_{s,i}, A_{c,i})$ is $U_{s} = \sum_{j=1}^{n_i}b_{i,j}-\sum_{j=1}^{n_i}\frac{s_{i,j}-\sqrt{{s_{i,j}}^2-4c_{i,j}\hat{c}_{i,j}}}{2}$.
Since $\hat{s}=\langle s_{i,1},\cdots,s_{i,n_1}\rangle$ is the unique solution for the optimization problem \eqref{eq:min_delay-revised2}, $U'_s<U_s$.
Therefore, the server has no incentive to deviate from $A_{s,i}$ either.

\section*{Appendix 5 Proof of Theorem~\ref{theo:multi_clients-1}}
Let $Opt(t,b_1,\cdots,b_m)$ denote the optimization problem
defined in (\ref{eq:max_wage_multiple-revised}) where
$t$ is an integer, 
$d=\sum_{i=1}^{m}\sum_{j=1}^{n_i}\hat{c}_{i,j} + t\cdot\delta$ and 
$b_i\geq\sum_{j=1}^{n_i}(c_{i,j}+\frac{c_{i,j}\hat{c}_{i,j}}{c_{i,j}+\hat{c}_{i,j}})$
for each $i\in\{1,\cdots,m\}$.
Note that, $c_{i,j}$ and $\hat{c}_{i,j}$ are constants.
Let ${\cal P}(t,b_1,\cdots,b_m)$ denote the following predicate: 
Phases I, II and III of Algorithm~\ref{alg:SplitDeposit} 
solves $Opt(t,b_1,\cdots,b_m)$. 
If ${\cal P}(t,b_1,\cdots,b_m)\equiv TRUE$, 
let $f^{*}(t,b_1,\cdots,b_m)$ denote 
the found minimal value of the objective function for $Opt(t,b_1,\cdots,b_m)$,
and each $s_{i,j}(t,b_1,\cdots,b_m)$ denote the value assigned to $s_{i,j}$ in the solution.  

The theorem claims ${\cal P}(t,b_1,\cdots,b_m)\equiv TRUE$ 
for every $t\geq 0$ and every $\{b_i|i=1,\cdots,m\}$ satisfying the above relevant constraints.
We next prove it by induction on $t$.

Base case: 
When $t=0$ (i.e., $d=\sum_{i=1}^{m}\sum_{j=1}^{n_i}\hat{c}_{i,j}$), 
${\cal P}(t,b_1,\cdots,b_m)\equiv TRUE$, i.e., 
Phases I-III of Algorithm~\ref{alg:SplitDeposit} 
solves $Opt(t,b_1,\cdots,b_m)$, because:
Phase I simply initializes every $s_{i,j}$ 
to satisfy constraint \eqref{eq:max_wage_multiple-revised-c1};
Phase II minimizes $\sum_{j=1}^{n_i}f(s_{i,j},i,j)$ 
with the remaining budget of each client $C_i$, 
independently, based on the arguments similar to 
the proof of Lemma~\ref{lem:singleClient-2};
Phase III does nothing as $d'=0$.

Induction step: 
Assuming ${\cal P}(t,b_1,\cdots,b_m)\equiv TRUE$ 
for every integer $0\leq t\leq t_0$ and 
every $\{b_i|i=1,\cdots,m\}$ satisfying the above relevant constraints,
next we prove 
${\cal P}(t_0+1,\hat{b}_1,\cdots,\hat{b}_m)\equiv TRUE$
for every $\{\hat{b}_i|i=1,\cdots,m\}$ satisfying the relevant constraints.

First, let us consider optimization problem 
    $OPT(0,\hat{b}_1,\cdots,\hat{b}_m)$.
According to the base case, 
${\cal P}(0, \hat{b}_1,\cdots,\hat{b}_m)\equiv TRUE$, and
each $s_{i,j}(0,\hat{b}_1,\cdots,\hat{b}_m)$ denotes 
the assignment of $s_{i,j}$ in the optimal solution.

Second, let 
$I=\{(i*,j*)\} = 
\operatorname*{arg\,min}_{\forall i\in\{1,\cdots,m\}\forall j\in\{1,\cdots,n_i\}}f'(s_{i,j},i,j))$.
Then, there must exist at least one $(i*,j*)\in I$ such that,
in the optimal solution to $OPT(t_0+1,\hat{b}_1,\cdots,\hat{b}_m)$, 
the assignment of $s_{i*,j*}$ is greater than 
the assignment of $s_{i*,j*}$ in the optimal solution to $OPT(0,\hat{b}_1,\cdots,\hat{b}_m)$; i.e.,
$s_{i*,j*}(t_0+1,\hat{b}_1,\cdots,\hat{b}_m)>s_{i*,j*}(0,\hat{b}_1,\cdots,\hat{b}_m)$. 
This can be proved by contradiction. 
If this is not the case, as $t_0+1\geq 1$, there should be at least $(i',j')\not\in I$ such that
$s_{i',j'}(t_0+1,\hat{b}_1,\cdots,\hat{b}_m)>s_{i',j'}(0,\hat{b}_1,\cdots,\hat{b}_m)$;
then, if one unit assigned to $s_{i',j'}$ instead is assigned to $s_{i*,j*}$, 
the server can earn higher wage. 

Third, let us consider optimization problem 
    $OPT(t_0,\hat{b}'_1,\cdots,\hat{b}'_m)$,
where $\hat{b}'_{i*} = \hat{b}_{i*}+\delta$ while $\hat{b}'_i=\hat{b}_i$ for every $i\not= i*$.
We can prove the following:
    First, the optimal solution to $OPT(t_0,\hat{b}'_1,\cdots,\hat{b}'_m)$ is a feasible solution 
    to $OPT(t_0+1,\hat{b}_1,\cdots,\hat{b}_m)$.
    This is because: 
    $\hat{b}'_i\geq\hat{b}_i$ for every $i$, thus 
    satisfying constraint \eqref{eq:max_wage_multiple-revised-c2} in $OPT(t_0,\hat{b}'_1,\cdots,\hat{b}'_m)$
    implies satisfying constraint \eqref{eq:max_wage_multiple-revised-c2} in $OPT(t_0+1,\hat{b}_1,\cdots,\hat{b}_m)$; 
    $\sum_{i=1}^{m}\hat{b}_i+\delta = \sum_{i=1}^{m}\hat{b}'_i$ and 
    the value of $d$ in $OPT(t_0+1,\hat{b}_1,\cdots,\hat{b}_m)$ is greater than the $d$
    in $OPT(t_0,\hat{b}'_1,\cdots,\hat{b}'_m)$ by $\delta$, thus
    satisfying \eqref{eq:max_wage_multiple-revised-c3} in $OPT(t_0,\hat{b}'_1,\cdots,\hat{b}'_m)$
    implies satisfying constraint \eqref{eq:max_wage_multiple-revised-c3} in $OPT(t_0+1,\hat{b}_1,\cdots,\hat{b}_m)$.
%
    Similarly, the optimal solution to $OPT(t_0+1,\hat{b}_1,\cdots,\hat{b}_m)$ is a feasible solution
    to $OPT(t_0,\hat{b}'_1,\cdots,\hat{b}'_m)$. 

Further due to the uniqueness of optimal solution to these optimization problems 
(based on Lemma~\ref{lem:multi_clients-1}),
$OPT(t_0+1,\hat{b}_1,\cdots,\hat{b}_m)$ and $OPT(t_0,\hat{b}'_1,\cdots,\hat{b}'_m)$ 
share the same optimal solution. 

Finally, based on the induction assumption, 
${\cal P}(t_0,\hat{b}'_1,\cdots,\hat{b}'_m)\equiv TRUE$.
That is, Algorithm~\ref{alg:SplitDeposit} solves $OPT(t_0,\hat{b}'_1,\cdots,\hat{b}'_m)$;
thus, it also solves $OPT(t_0+1,\hat{b}_1,\cdots,\hat{b}_m)$,
except that $\{\hat{b}'_i\}$ and $d=\sum_{i=1}^{m}\sum_{j=1}^{n_i}\hat{c}_{i,j}+t_0\cdot\delta$ 
(not $\{\hat{b}_i\}$ and $d=\sum_{i=1}^{m}\sum_{j=1}^{n_i}\hat{c}_{i,j}+(t_0+1)\cdot\delta$) 
are the inputs to the algorithm.
Hence, in the execution of the algorithm,
let us move the last assignment to $s_{i*,j*}$ in Phase II 
to be the first step in Phase III; this way,  
the algorithm works exactly as it takes $\{\hat{b}_i\}$ and $d=\sum_{i=1}^{m}\sum_{j=1}^{n_i}\hat{c}_{i,j}+(t_0+1)\cdot\delta$ 
as inputs to solve $OPT(t_0+1,\hat{b}_1,\cdots,\hat{b}_m)$.
That is, ${\cal P}(t_0+1,\hat{b}_1,\cdots,\hat{b}_m)\equiv TRUE$.

\comment{
\appendices
\section{}
\label{appendix:single_client_unique_point}
Based on \eqref{eq:min_delay-revised}, for the single client problem (client $i$), we define an optimization problem in the following form:
\begin{align}
        & \min f_i(x_i) \label{eq:f(x)formal}\\
        s.t. \ & g_{i,j}(x_i) \equiv -x_{i,j} + c_{i,j}+\frac{c_{i,j}\hat{c}_{i,j}}{c_{i,j}+\hat{c}_{i,j}} + \hat{c}_{i,j} \leq 0, \label{eq:f-c1}\\
        & h_{i}(x_i)\equiv\sum_{j=1}^{n_i} x_{i,j} - b_i - d_i = 0, ~~j \in \{1,\ldots,n_i\}, \label{eq:f-c2}
\end{align}
where $x_i=\{x_i,j\}, x_{i,j} = b_{i,j} + d_{i,j}$, $f_i(x_i) = \sum_{j=1}^{n_i}f(x_i, i,j)$. $f(x_i, i,j)$ is defined in \eqref{eq:fx}. Constraint \eqref{eq:f-c1} is based on \eqref{eq:min_delay-c3} and \eqref{eq:min_delay-c4}; constraint \eqref{eq:f-c2} is based on \eqref{eq:min_delay-c1} and \eqref{eq:min_delay-c2}. We will show that the following lemma holds. 
\begin{lemma}
The optimization problem defined by \eqref{eq:f(x)formal} - \eqref{eq:f-c2} has a unique optimal point.
\end{lemma}
\begin{proof}
We prove this by showing that our problem is actually a convex optimization problem, which the standard form reads:
\begin{align} 
    &\min f(x)\nonumber \\
    s.t. \ & g_u(x)\leq 0, \ u = 1,2,...,m, \nonumber\\
    &h_v(x)=0, \ v = 1,2,...,p,
\end{align}
where $f,g_1,...,g_m$ are convex functions and $h_1,h_2,...,h_p$ are affine functions. It has been proven in \cite{beck2014} that if $f(x)$ is a strictly convex function, then there exists a unique global optimal point. We will prove that in our problem, $f_i$ is strictly convex; $g_{i,j}$'s are convex and $h_i$ is affine.

First, since $g_{i,j}$ is linear with $x_{i,j}$, the second derivative $\partial^2 g_{i,j}/\partial x_{i,j}\partial x_{i,k} = 0$ for any $j,k\in\{1,\ldots,n_i\}$. It is obvious that the Hessian matrix of any $g_{i,j}$ is positive semidefinite, which implies $g_{i,j}$ is convex \cite{boyd2004convex}. 

Second, it is obvious that $h_{i}$ is affine. 

Third we will show that $f_i$ is strictly convex. Let $K_{i,j} = c_{i,j}\hat{c}_{i,j}$. The Hessian matrix of $f_i(x_i)$ has only $n_i$ non-zero diagonal elements which are given by
\begin{equation}
    \frac{\partial^2 f_i(x_i)}{\partial x^2_{i,j}} = \frac{2K_{i,j}}{(x_{i,j}^2-4K_{i,j})^{\frac{3}{2}}} > 0, \nonumber
\end{equation}
where $j \in \{1,2,\dots,n_i\}$. That is, the Hessian matrix of $f(x_i)$ has the form of
\begin{equation}
H_i =
     \begin{pmatrix}
\frac{2K_{i,1}}{(x_{i,1}^2-4K_{i,1})^{\frac{3}{2}}} & ... & 0 & ... & 0 \\
0 & \frac{2K_{i,2}}{(x_{i,2}^2-4K_{i,2})^{\frac{3}{2}}} &  0 &  ... & 0  \\
\vdots & \vdots & \vdots & \vdots & \vdots\\
0 & ... & ... & ... & 0 \\
0 & ... & ... & ... & \frac{2K_{i,n_i}}{(x_{i,n_i}^2-4K_{i,n_i})^{\frac{3}{2}}} \\
\end{pmatrix} \\
\end{equation} 

Let $z=(z_1, ..., z_{n_i})^T$ be any non-zero column vector, where $T$ denotes transpose, then we have
\begin{align}
& \begin{pmatrix}
    z_1, \ldots, z_{n}
  \end{pmatrix}
  H_i
  \begin{pmatrix}
    z_1 \\
    \vdots\\
    z_{n}
  \end{pmatrix}
  \notag\\
= & \sum_{j = 1}^{n_i} z_j^2\frac{2K_{i,j}}{(x_{i,j}^2-4K_{i,j})^{\frac{3}{2}}} \nonumber\\
> & 0.\notag
\end{align}
That is, $z^T H_i z > 0$, which means that $H_i$ is positive definite, hence $f(x_i)$ is strictly convex \cite{boyd2004convex}.
\end{proof}
Let $s_i^*=\{s_{i,1}, s_{i,2},\dots,s_{i,n_i}\}$ be the unique optimal solution of \eqref{eq:f(x)formal}. Since $s_{i,j} = b_{i,j} + d_{i,j}$ for $j = \{1,2,\dots,n_i\}$, given $s_{i,j}$, there could be multiple ways to divide it into $b_{i,j}$ and $d_{i,j}$. Notice that problem \eqref{eq:min_delay-revised} and \eqref{eq:f(x)formal} are actually not exactly the same; the constraints of problem \eqref{eq:min_delay-revised} are stricter than the constraints of problem \eqref{eq:f(x)formal}. Besides the requirement that 
\begin{equation}
\label{eq:s=b+d}
    b_{i,j} + d_{i,j} = s_{i,j},
\end{equation}
where $j = \{1,2,\dots,n_i\}$, a valid allocation of the budget and deposit of \eqref{eq:min_delay-revised} must also satisfies constraints \eqref{eq:min_delay-c1} - \eqref{eq:min_delay-c4}. Hence, optimal value of \eqref{eq:f(x)formal} is no greater than that of \eqref{eq:min_delay-revised}. We will show that such a feasible allocation always exists. That is, \eqref{eq:f(x)formal} and \eqref{eq:min_delay-revised} have the same optimal value. Algorithm \ref{alg:s_to_b_and_d} provides a way to find an optimal valid allocation.
\begin{algorithm}[htb]
\caption{Find $\{b_{i,j}\}$ and $\{d_{i,j}\}$ given $\{s_{i,j}\}$.}\label{alg:s_to_b_and_d}
{\bf Input:}
\begin{itemize}
\item
$b_i$: total budget of client $C_i$;
\item
$d_i = d$: total deposit of server $S$;
\item
$n_i$: total number of tasks;
\item
task set $\{t_{i,1},\cdots,t_{i,n_i}\}$ and
associated costs $\{c_{i,1},\cdots,c_{i,n_i}\}$ and $\{\hat{c}_{i,1},\cdots,\hat{c}_{i,n_i}\}$;
\item
$s_i^* = \{s_{i,1}^*, s_{i,2}^*,\dots,s_{i,n_i}^*\}$: the optimal solution of \eqref{eq:f(x)fomal}. 
\end{itemize}

{\bf Output:} $\{b_{i,j}\}$ and $\{d_{i,j}\}$ that satisfy \eqref{eq:min_delay-c1} - \eqref{eq:min_delay-c4} and \eqref{eq:s=b+d}. \\

{\bf Phase I:} Initialization.
\begin{algorithmic}[1]
\For{$j\in\{1,\cdots,n_i\}$}
    \State
    $d_{i,j}\leftarrow c_{i,j}$;
    ~$b_{i,j}\leftarrow c_{i,j}+\frac{c_{i,j}\hat{c}_{i,j}}{c_{i,j}+\hat{c}_{i,j}}$;
    ~$s_j\leftarrow 0$
    \Comment{meet constraints~(\ref{eq:min_delay-c3}) and (\ref{eq:min_delay-c4})}
\EndFor
\end{algorithmic}
~~\\
{\bf Phase II:} Sequentially Allocation of the Remaining Budget/Deposit.
\begin{algorithmic}[1]
\State $id = 1$
\State $b'_i\leftarrow b_i-\sum_{j=1}^{n_i}(c_{i,j}+\frac{c_{i,j}\hat{c}_{i,j}}{c_{i,j}+\hat{c}_{i,j}})$ \Comment{allocating the remaining budget to satisfy \eqref{eq:min_delay-c1} and \eqref{eq:s=b+d}}
\For{$j\in\{1,\cdots,n_i\}$} 
    \State $id = i$;
    \While{$b'_i\geq\delta$ and $(b_{i,j} + d_{i,j}) \leq s_{i,j}^* - \delta$} 
    \State $b_{i,j} \leftarrow b_{i,j} +\delta$;~$b'_i\leftarrow (b'_i-\delta)$
    \EndWhile
\EndFor
\State $d'_i\leftarrow d_i-\sum_{j=1}^{n_i}\hat{c}_{i,j}$ \Comment{allocating the remaining deposit to satisfy \eqref{eq:min_delay-c2} and \eqref{eq:s=b+d}}
\For{$j\in\{id, id+1,\cdots,n_i\}$} 
    \While{$b'_i\geq\delta$ and $(b_{i,j} + d_{i,j}) \leq s_{i,j}^* - \delta$} 
    \State $d_{i,j} \leftarrow d_{i,j} +\delta$;~$d'_i\leftarrow (d'_i-\delta)$
    \EndWhile
\EndFor
\end{algorithmic}
\end{algorithm}

It is clear that the allocation given by Algorithm \ref{alg:s_to_b_and_d} satisfies \eqref{eq:min_delay-c1} - \eqref{eq:min_delay-c4} as well as \eqref{eq:s=b+d}. Actually, there could be many different allocations of $\{b_{i,j}\}$ and $\{d_{i,j}\}$. As long as those allocations satisfy \eqref{eq:min_delay-c1} - \eqref{eq:min_delay-c4} and \eqref{eq:s=b+d}, they give us the same optimal value. Therefore, to find an optimal allocation of the budget and deposit, we could
\begin{enumerate}
    \item Find the optimal solution $s_i^*$ of \eqref{eq:f(x)formal}.
    \item Find a decomposition of $s_i^*$ into $\{b_{i,j}\}$ and $\{d_{i,j}\}$ which satisfies \eqref{eq:min_delay-c1} - \eqref{eq:min_delay-c4} and \eqref{eq:s=b+d}.
\end{enumerate}
}

\comment{
\section{10 tasks with different $c$ and $\frac{\hat{c}}{c}$}
In section~\ref{sec:single_clt}, we compare different allocation methods with tasks having same cost $c$, same $\hat{c}$ or same $\frac{\hat{c}}{c}$. We further analyze 10 tasks with different cost $c$ and different $\frac{\hat{c}}{c}$. The 10 tasks have cost $c\in\{10,20,\cdots,100\}$ and the corresponding $\hat{c}_i=c_i*i$ for $i\in\{1,2,\cdots10\}$.
\begin{figure}[htb!]
  \begin{subfigure}{0.48\columnwidth}
    \includegraphics[width=\columnwidth]{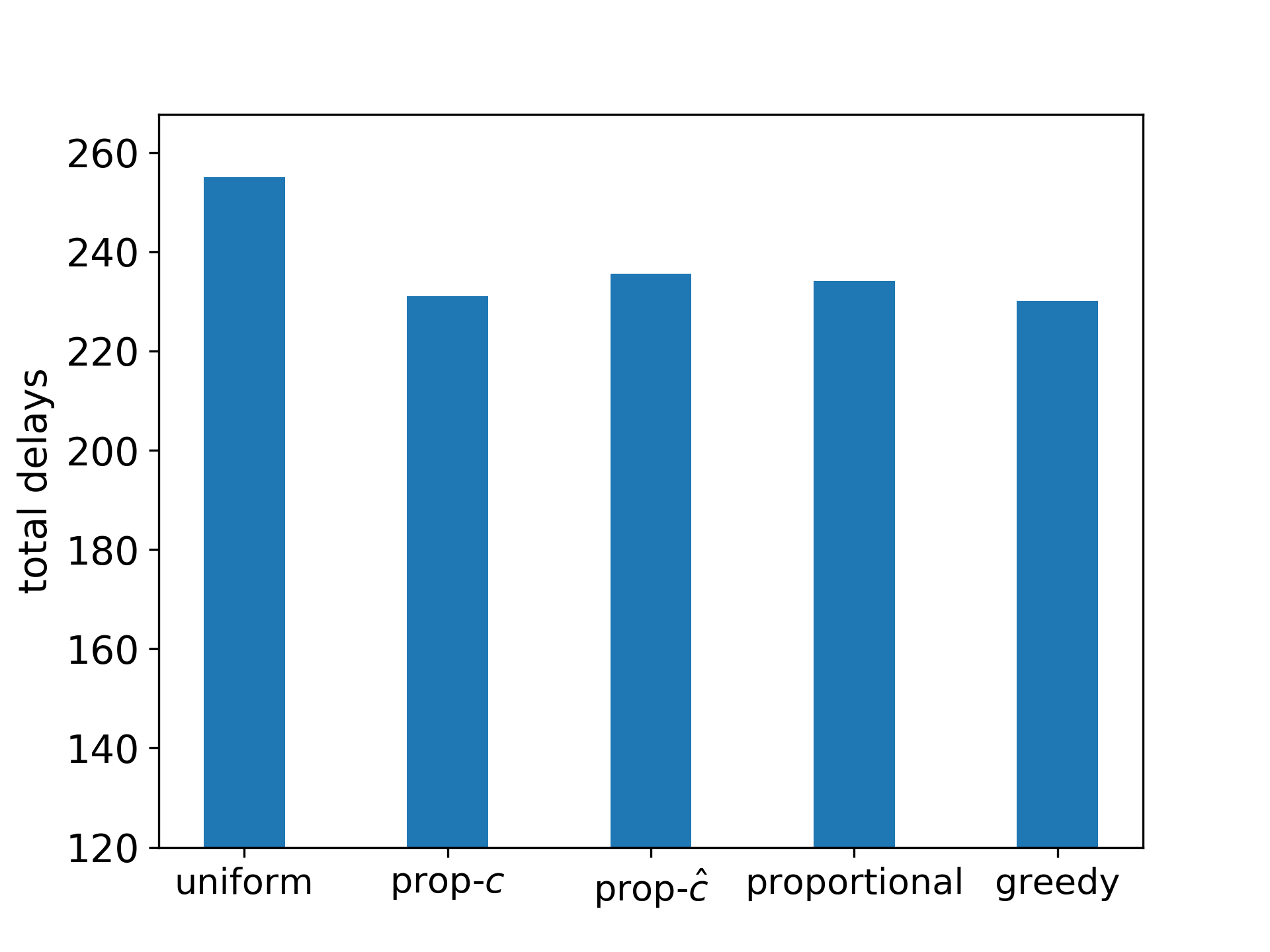}
    \caption{total delay for 10 tasks} \label{static:diffratio-a}
  \end{subfigure}%
  \hspace*{\fill}   
  \begin{subfigure}{0.48\columnwidth}
    \includegraphics[width=\columnwidth]{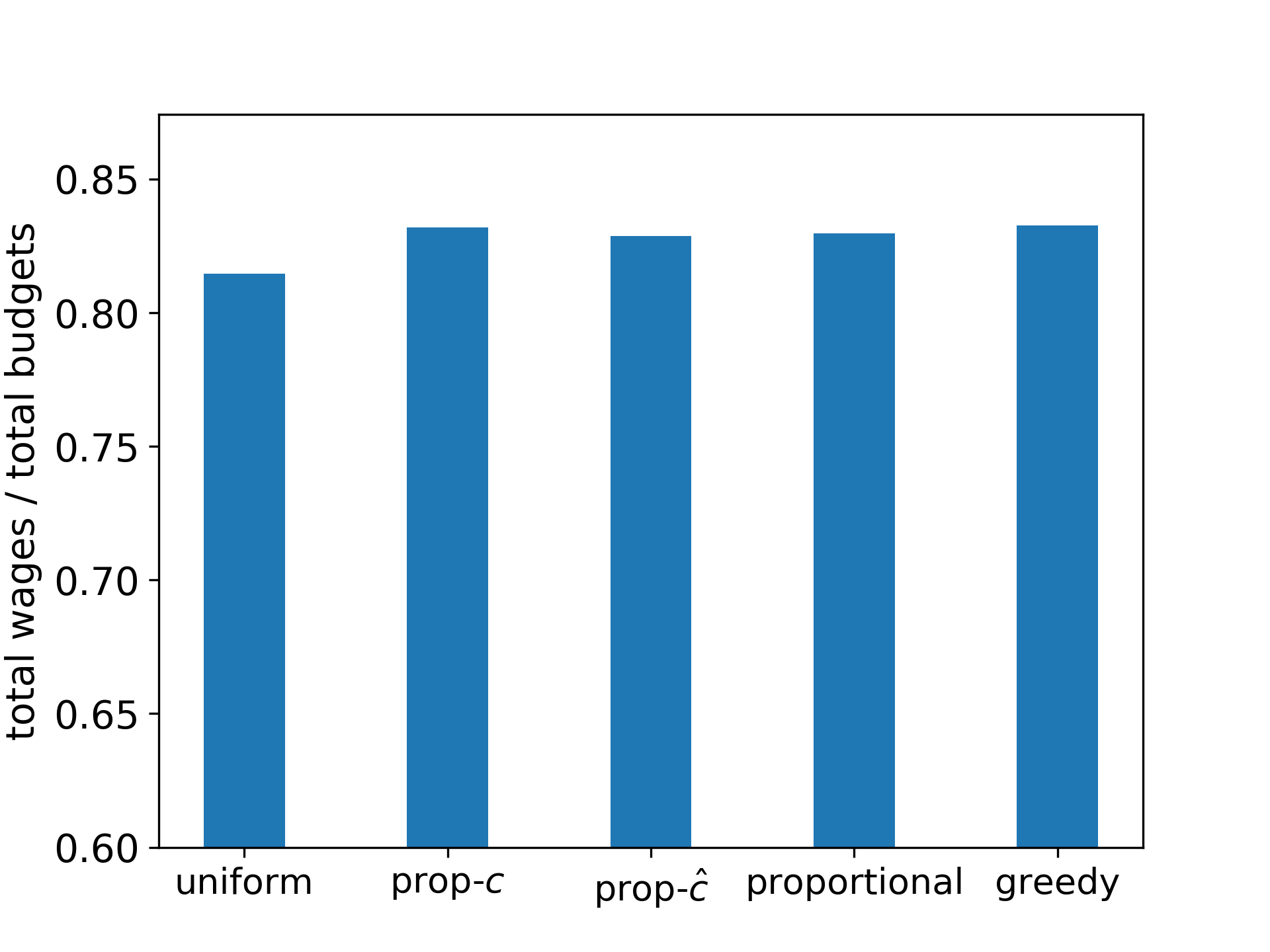}
    \caption{total wage / total budget} \label{static:diffratio-b}
  \end{subfigure}%
  \hfill   
  \begin{subfigure}{0.48\columnwidth}
    \includegraphics[width=\columnwidth]{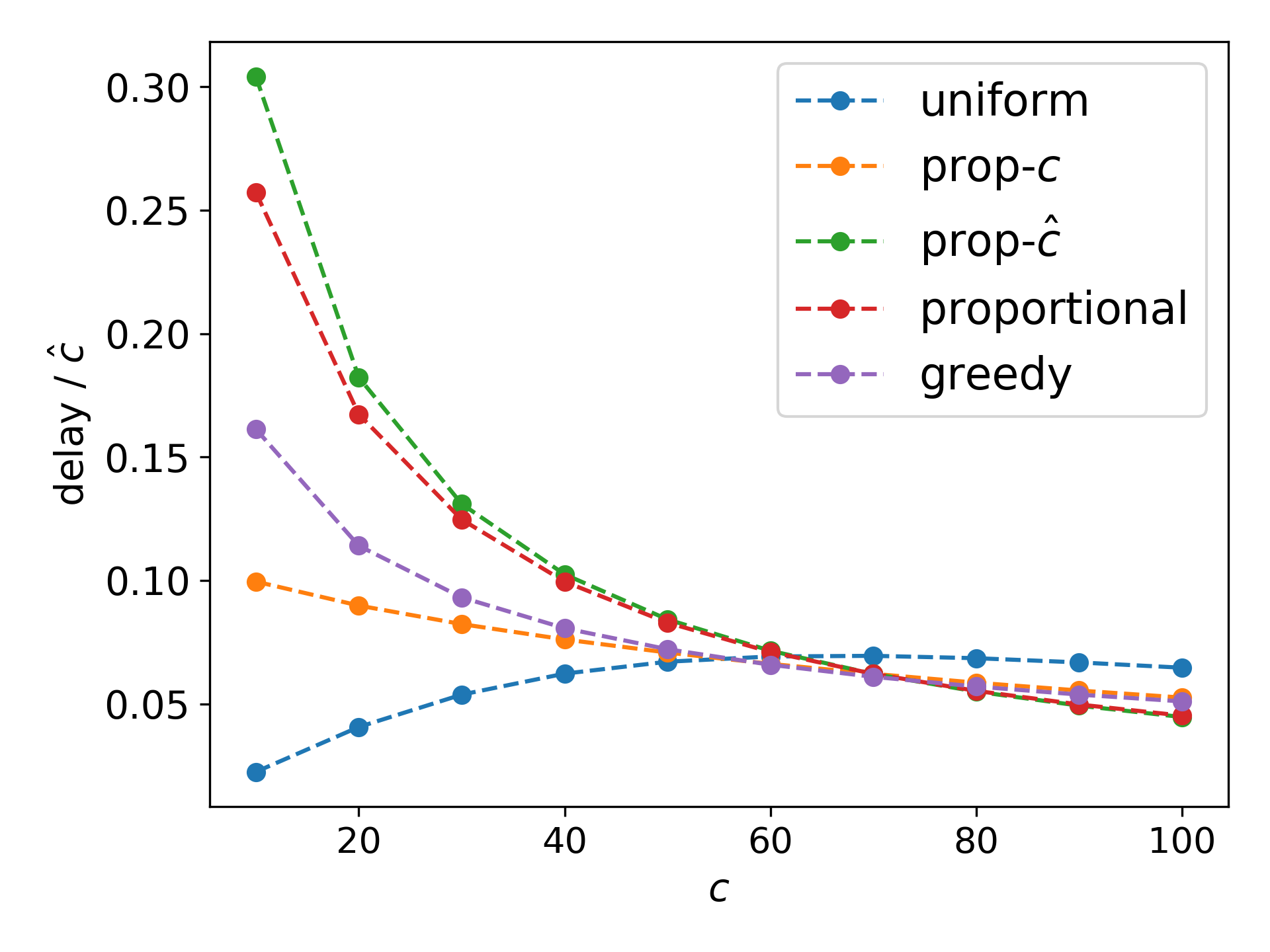}
    \caption{delay/$\hat{c}$ v.s. $c$} \label{static:diffratio-c}
  \end{subfigure}%
  \hspace*{\fill}   
  \begin{subfigure}{0.48\columnwidth}
    \includegraphics[width=\columnwidth]{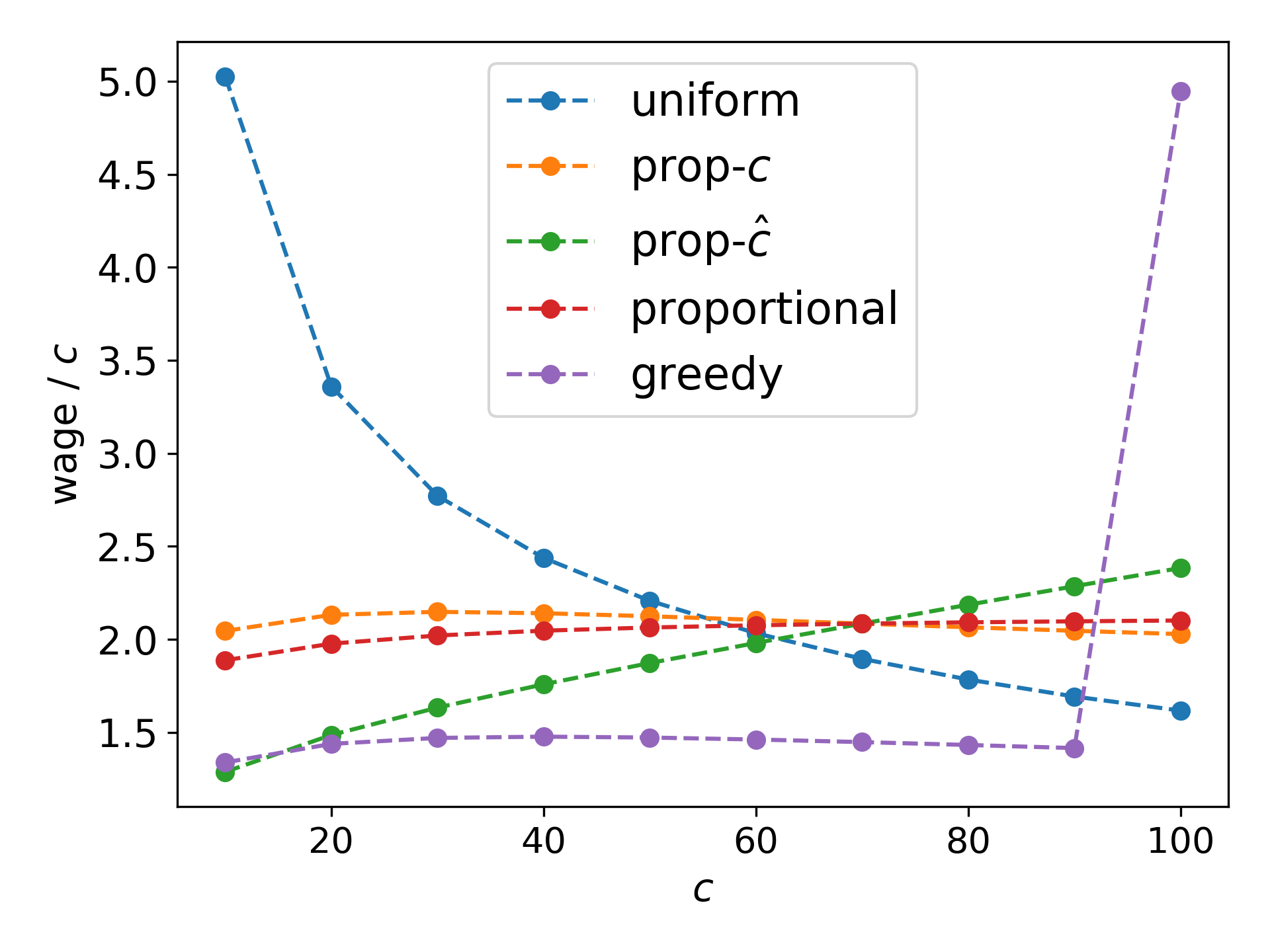}
    \caption{wage/$c$ v.s. $c$} \label{static:diffratio-d}
  \end{subfigure}%
\caption{10 tasks with $c=[10,\cdots,100]$ and $\frac{\hat{c}}{c}=[1,\cdots,10]$} 
\label{static:diff-ratio}
\end{figure}
}






\end{document}